\documentclass[opre,sglanonrev]{informs4}

\newcommand{\I}{\mathbf{1}}
\newcommand{\E}{\mathbb{E}}
\newcommand{\p}{\mathbb{P}}
\newcommand{\UU}{\mathcal{U}}
\newcommand{\diff}{\mathrm{d}}
\newcommand{\LA}{\boldsymbol\lambda}
\newcommand{\BA}{\boldsymbol\beta}
\newcommand{\dl}{{\rm DefLeg}}
\newcommand{\pl}{{\rm PreLeg}}

\newcommand{\Beta}{\mathrm{Beta}}
\newcommand{\B}{\mathrm{B}}

\newcommand{\uf}{\textrm{uf}}

\newtheorem{Proposition}{Proposition}[section]
\newtheorem{Lemma}{Lemma}[section]
\newtheorem{Theorem}{Theorem}[section]

\newtheorem{Definition}{Definition}[section]
\newtheorem{Remark}{Remark}[section]
\usepackage{booktabs}
\usepackage{subcaption}
\usepackage{algorithm}
\usepackage{caption}
\usepackage[toc,page]{appendix}

\counterwithin{algorithm}{section}

\usepackage{enumitem}    
\usepackage{chngcntr}    
\graphicspath{ {images/} }

\counterwithin{algorithm}{section}

\newenvironment{breakablealgorithm}[1]{%
	\par\vspace{\medskipamount}%
	\hrule height .8pt %
	\vspace{2pt}%
	\captionsetup{type=algorithm, labelsep=colon, labelfont=bf, singlelinecheck=off, justification=raggedright}
	\captionof{algorithm}{#1}
	\vspace{-6pt}%
	\hrule%
	\vspace{6pt}%
	\begin{enumerate}[wide, labelindent=0pt, label=\textbf{Step \arabic*:}]
	}{%
	\end{enumerate}%
	\vspace{2pt}%
	\hrule height .8pt%
	\par\vspace{\medskipamount}%
}

\RequirePackage{tgtermes}
\RequirePackage{bm}
\RequirePackage{endnotes}

\OneAndAHalfSpacedXII 

\usepackage{algpseudocode}
\usepackage{chngcntr}
\counterwithin{algorithm}{section}
\usepackage{enumitem}

\usepackage{tabularx}

\usepackage{natbib}
\bibpunct[, ]{(}{)}{,}{a}{}{,}%
\def\bibfont{\small}%
\def\bibsep{\smallskipamount}%
\EquationsNumberedThrough    

\TheoremsNumberedThrough     
\ECRepeatTheorems  %

\MANUSCRIPTNO{MOOR-0001-2024.00}

\begin{document}


\RUNAUTHOR{Bu, Cai, Xia and Yang}

\RUNTITLE{Perfectly Fitting CDO Prices Across Tranches}

\TITLE{Perfectly Fitting CDO Prices Across Tranches: A Theoretical Framework with Efficient Algorithms}

\ARTICLEAUTHORS{%
\AUTHOR{Lan Bu}
\AFF{Financial Technology Thrust, Hong Kong University of Science and Technology (Guangzhou) \EMAIL{bulan@pku.edu.cn}}
\AUTHOR{Ning Cai}
\AFF{Financial Technology Thrust, Hong Kong University of Science and Technology (Guangzhou) \EMAIL{ningcai@hkust-gz.edu.cn}}
\AUTHOR{Chenxi Xia}
\AFF{School of Mathematical Sciences, Peking University \EMAIL{xiacx@pku.edu.cn}}
\AUTHOR{Jingping Yang}
\AFF{School of Mathematical Sciences, Peking University \EMAIL{yangjp@math.pku.edu.cn}}

} 

\ABSTRACT{
This paper addresses a key challenge in CDO modeling: achieving a perfect fit to market prices across all tranches using a single, consistent model. The existence of such a perfect-fit model implies the absence of arbitrage among CDO tranches and is thus essential for unified risk management and the pricing of nonstandard credit derivatives.
To address this central challenge, we face three primary difficulties: standard parametric models typically fail to achieve a perfect fit; the calibration of standard parametric models inherently relies on computationally intensive simulation-based optimization; and there is a lack of formal theory to determine when a perfect-fit model exists and, if it exists, how to construct it. We propose a theoretical framework to overcome these difficulties. We first introduce and define two compatibility levels of market prices—weak compatibility and strong compatibility. Specifically, market prices across all tranches are said to be weakly (resp. strongly) compatible if there exists a single model (resp. a single conditionally i.i.d. model) that perfectly fits these market prices.  
We then derive sufficient and necessary conditions for both levels of compatibility by establishing a relationship between compatibility and linear programming (LP) problems. 
Furthermore, under either condition, we construct a corresponding concrete copula model that achieves a perfect fit. Notably, our framework not only allows for efficient verification of weak compatibility and strong compatibility through LP problems but also facilitates the construction of the corresponding copula models that achieve a perfect fit, eliminating the need for simulation-based optimization. 
The practical applications of our framework are demonstrated in risk management (e.g.,  designing effective hedging strategies and estimating loss distributions for credit portfolios) and the pricing of nonstandard credit derivatives.
}






\maketitle




\section{Introduction} \label{sec:intro}

\subsection{Background and Motivations} \label{sec_sub:background_motivation}

Collateralized Debt Obligations (CDOs) constitute a significant segment of the credit markets, as underscored both by their substantial trading volume with the notional value of iTraxx/CDX tranches traded in 2023 estimated at \$335 billion (\citealp{godec2024fixed}) and by their instrumental role as primary tools for correlation trading (\citealp{mounfield2009synthetic}). 
A critical objective in CDO modeling is to develop a single, consistent model that achieves a ``perfect fit" to market prices across all tranches—by this, we mean a model capable of generating prices that exactly align with market prices across all tranches \citep{hull2006valuing}. Such a consistent model with perfect fit (called a ``perfect-fit" model hereafter) is essential for important applications, such as unified risk management and the pricing of nonstandard credit derivatives. It not only offers a cohesive framework for designing self-consistent hedging strategies for credit portfolios and calculating credit portfolios' loss distributions, but also enables the arbitrage-free pricing of nonstandard credit derivatives, including CDOs with nonstandard attachment and detachment points and CDOs with a nonstandard number of underlying names.

There are three primary difficulties in achieving the ``perfect fit" objective. First, standard parametric models generally fall short of attaining a perfect fit to market prices across all tranches. Regarding the seminal and influential Gaussian copula model proposed by \cite{li2000}, identifying a single Gaussian copula model that perfectly fits market prices across all tranches proves difficult. Indeed, ``correlation smiles" emerge (\citealp{moosbrucker2006explaining}), indicating that different tranches imply different correlations, thereby rendering Gaussian copula models inadequate for achieving a perfect fit. A more serious limitation is that even when a single Gaussian copula model is used to fit the market prices of an individual tranche, the implied correlation is often non-unique or even non-existent (e.g., \citealp{brigo2010credit} and \citealp{mcneil2015quantitative}). In the literature, a variety of different parametric copula models have been proposed, including elliptical copulas (e.g., the Student's $t$ copula), Archimedean copulas (e.g., the Clayton and nested Archimedean families; see, e.g., \citealp{prange2009correlation} and \citealp{hofert2010sampling}), and other advanced models based on the normal inverse Gaussian (NIG) or $\alpha$-stable distributions (e.g., \citealp{kalemanova2007normal} and \citealp{prange2009correlation}). While these parametric models offer greater flexibility, the related research indicates that these models also fail to achieve a perfect fit across all tranches. 

Second, 
due to the general inadequacy of standard parametric models to achieve a perfect fit, a common approach when calibrating them to market prices across all tranches is to minimize some calibration error metrics. Given that closed-form pricing formulas for CDO tranches are typically unavailable, this approach usually entails solving simulation-based optimization problems, which are computationally intensive as many key quantities involved in the numerical optimization have to be computed via simulation.  

Third, and most importantly, to the best of our knowledge, the literature lacks a general formal theory to determine when a perfect-fit model exists and, if it exists, how to construct it. A pioneering study in identifying a perfect-fit model across all tranches is \cite{hull2006valuing}. Nonetheless, it seems that their attention is confined to a relatively narrow class of models with specific assumptions; they  assume that default times are independent conditional on a random hazard rate that can take one of a finite number of states, and then calibrate the probabilities of these states to achieve a perfect fit. 

The celebrated work of \cite{hull2006valuing} motivates us to consider general models with minimal assumptions for achieving a perfect fit and to establish a general formal theory concerning the existence and construction of a general perfect-fit model. Moreover, it will be desirable to develop efficient algorithms to verify the existence of a general perfect-fit model and, if it exists, to construct such a model, without resorting to computationally intensive simulation-based optimization. 

\subsection{Our Contributions} 

In this paper, we propose a theoretical framework to examine when a general perfect-fit model exists and, if it exists, how to construct it. Our study focuses on the existence and construction of general ``perfect-fit" copula functions. This is without loss of generality because, according to Sklar's theorem (\citealp{CopulaMethod}), any probabilistic model for CDOs, including structural models of asset values (e.g., \citealp{kijima2010pricing}, \citealp{collin2012relative}, and \citealp{seo2018rare}) and reduced-form models of default intensities (e.g., \citealp{peng2008connecting} and \citealp{Frey2012}), implies a joint distribution of default times and, in turn, a corresponding copula function.

Specifically, we first introduce a new concept—weak compatibility—as the most natural and comprehensive characterization of the consistency of CDO market prices across all tranches. CDO market prices across all tranches on the same underlying portfolio are said to be weakly compatible if there exists any copula (equivalent to the existence of any probabilistic model) that can perfectly fit them. Notably, weak compatibility represents the minimal requirement for the theoretical consistency of market prices, as it guarantees the existence of a unified pricing measure (i.e., an equivalent martingale measure), thus ensuring the absence of arbitrage in the market. Furthermore, it enables consistent pricing of nonstandard credit derivatives on the same portfolio, such as CDOs with nonstandard attachment and detachment points. 

While weak compatibility addresses the minimal requirement for the theoretical consistency of market prices, practical applications often demand more structure. For instance, a general copula that meets weak compatibility may have a complex functional form, making it difficult to apply to solving certain practically important problems, such as the efficient pricing of CDOs with a nonstandard number of underlying names and the computation of credit portfolios' loss distributions for risk management. To address these practical needs, we introduce another new concept: strong compatibility; CDO market prices across all tranches on the same underlying portfolio are said to be strongly compatible if there exists a copula from the class of conditionally independent and identically distributed (i.i.d.) models that can perfectly fit them. Beyond enabling solutions to the aforementioned practical problems, strong compatibility is also well-motivated by the fact that both structural and reduced-form credit models, when applied under a homogeneity assumption, implicitly rely on a conditionally i.i.d. structure. 

Indeed, it is straightforward to see that strong compatibility implies weak compatibility. Furthermore, failure to be weak compatibility essentially implies ``incompatibility", meaning that there exists no copula that perfectly fits the CDO market prices across all tranches and thus underscoring the fact that weak compatibility represents the minimal requirement for the theoretical consistency of market prices.
See Figure~\ref{Comp1} for their relationship. 

\begin{figure}[ht] 
		\caption{Categorization of CDO Market Prices Across All Tranches 
			Based on Their Compatibility}
\begin{center}	\includegraphics[width=0.56\textwidth,height=0.3\textwidth]{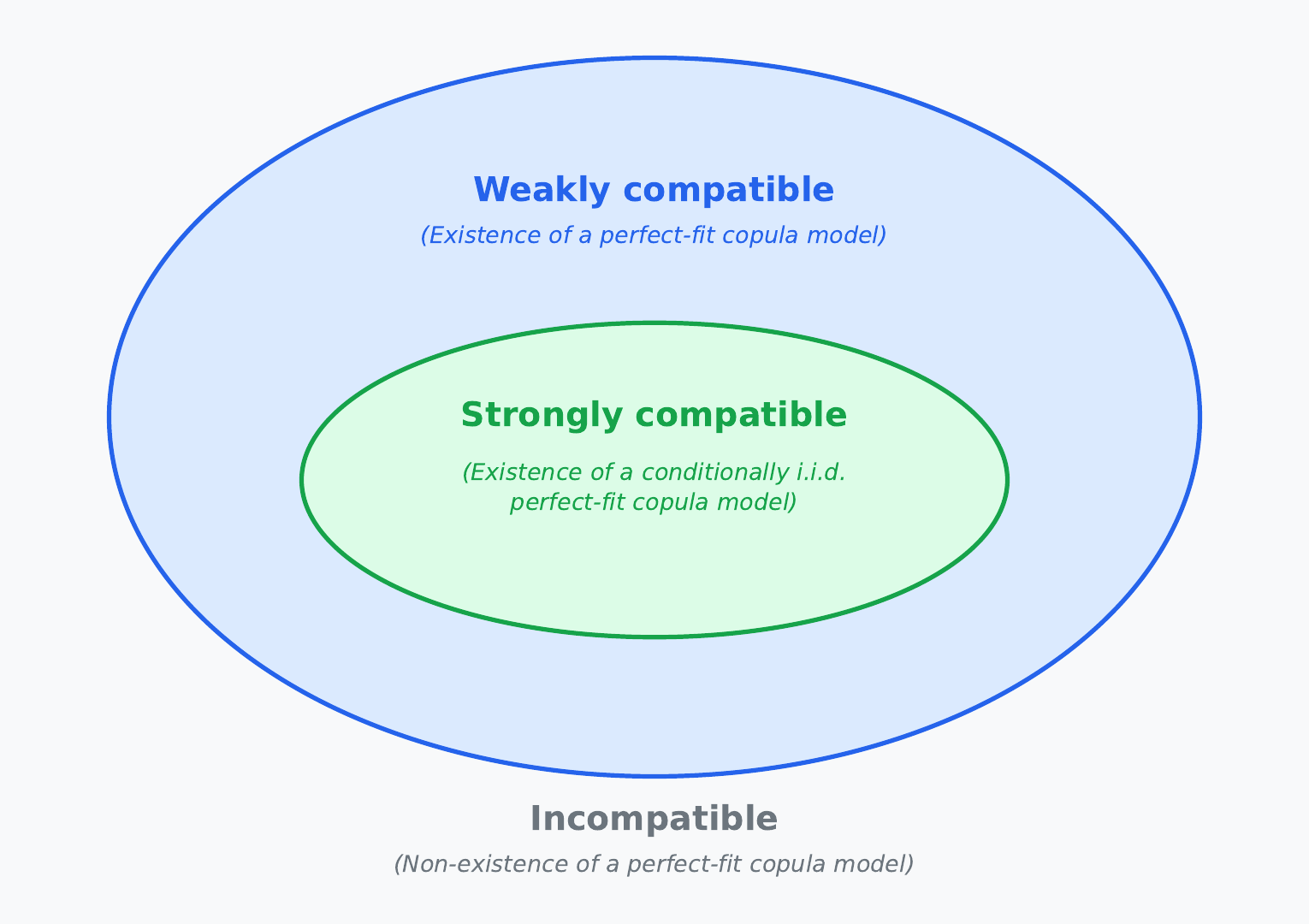}
\end{center} \vspace{0.2cm}
	\begin{minipage}{\textwidth}
\small \linespread{1.2}\selectfont 
	\vspace{0.15cm}

	{{\footnotesize{\emph{Notes.} This figure illustrates the categorization of CDO market prices across all tranches on the same underlying portfolio based on their compatibility: (i) incompatible (grey), (ii) weakly compatible (blue), and (iii) strongly compatible (green). 
		Indeed, ``strongly compatible" implies ``weakly compatible". Furthermore, failure to be ``weakly compatible" implies ``incompatible", underscoring the fact that ``weakly compatible" represents the minimal requirement for the theoretical consistency of market prices.}}}	
		\end{minipage}
\label{Comp1} 
\end{figure} 
\vspace{-0.3cm}

For both weak compatibility and strong compatibility, we derive a sufficient and necessary condition by establishing a relationship between each concept and a linear programming (LP) problem. Specifically, for weak compatibility, we propose a DPM-based approach to reduce its verification to an LP problem, where the DPM, abbreviated from the ``default probability matrix", summarizes information about the default count distribution at a finite number of prespecified time points. For strong compatibility, a novel and flexible family of copulas termed ``gamma-distorted copulas" is introduced to establish the  sufficient and necessary condition. Furthermore, for each, when the established sufficient and necessary condition holds, we demonstrate how to construct a corresponding concrete copula model that achieves a perfect fit. Our framework not only enables the efficient verification of weak compatibility and strong compatibility through LP problems but also facilitates the construction of the corresponding concrete copula models, eliminating the need for computationally intensive simulation-based optimization.

In summary, the main contributions of our paper are threefold. 
\begin{enumerate}[label=$\bullet$]
\item We establish a formal theoretical framework to address two questions: when a general copula model exists that can perfectly fit CDO market prices across all tranches, and if it exists, how to construct such a perfect-fit copula model. To this end, we put forward two new concepts—weak compatibility and strong compatibility—to fill the literature gap,  which respectively address the existence of a general copula model and a general conditionally i.i.d. copula model that can achieve a perfect fit.

\item For both weak compatibility and strong compatibility, we derive a sufficient and necessary condition by reducing its verification to an LP problem. In addition, for each, when the associated condition holds, we demonstrate how to construct a corresponding concrete perfect-fit copula model. Both the verification and construction processes are efficient because the related LP problems can be readily solved numerically, without resorting to computationally intensive simulation-based optimization.

\item We apply weak compatibility and strong compatibility to several important practical applications. Specifically, 
the former is employed to construct effective model-independent hedging strategies for credit portfolios and to price CDOs with nonstandard attachment and detachment points, while the latter is used to compute the loss distributions of credit portfolios and to price CDOs with a nonstandard number of underlying names. See 
Table~\ref{tab:compatibility_applications} for a summary of these applications. 


\end{enumerate}

\begin{table}[ht]
\centering
\caption{Applications of Weak Compatibility and Strong Compatibility}\vspace{0.1cm}
\begin{tabularx}{\textwidth}{l@{\hspace{1em}}l@{\hspace{1em}}l}
\toprule
& Weak Compatibility & Strong Compatibility \\
\midrule
\addlinespace
\shortstack[l]{Pricing nonstandard \\ credit derivatives} & \shortstack[l]{Pricing CDOs with nonstandard \\ attachment and detachment points} & \shortstack[l]{Pricing CDOs with a nonstandard \\ number of underlying names} \\
\addlinespace
\shortstack[l]{Conducting risk \\ management} & \shortstack[l]{Constructing model-independent \\ hedging strategies for credit portfolios} & \shortstack[l]{Estimating the loss distributions \\of credit portfolios} \\
\bottomrule
\end{tabularx}
\label{tab:compatibility_applications}
\end{table}

\subsection{Comparison with \cite{hull2006valuing}} \label{secsub:compare_hullwhite}

The work by \cite{hull2006valuing} stands as a seminal contribution, marking the first explicit emphasis on the criticality of achieving a perfect fit to CDO market prices across all tranches. Their ``implied copula" method offers an influential demonstration of how such calibration might be practically accomplished, and their work serves as a core motivation for our research. Our paper builds on the ``perfect fit" objective they lay out and establishes a formal theoretical foundation for this problem. Specifically, our theoretical framework extends their analysis along three key dimensions.


First, our theoretical framework is substantially more general. 
The implied copula method of \cite{hull2006valuing} relies on a specific functional form derived from the assumption that default times are independent conditional on a random hazard rate that can take one of a finite number of states. In contrast, our weak compatibility framework adopts the most general perspective, considering the set of all valid copulas without imposing any \textit{a priori} structural assumptions. 
This generality enables us to establish the model-independent conditions for price consistency. Furthermore, even our strong compatibility framework remains highly general, encompassing all conditionally i.i.d. copula models—with the model considered in \cite{hull2006valuing} included as a special case. 

Second, our theoretical framework is developed for the general multi-period setting inherent in standard CDO contracts. While the implied copula method of \cite{hull2006valuing} values cash flows across multiple periods, its dependence structure is effectively static. 
In their model, the default count distributions at all time points are governed by a single set of calibrated state probabilities, which implies that if the default distribution for any single date is known, the distributions for all other dates are automatically determined. 
By contrast, our strong compatibility framework provides the necessary degrees of freedom to characterize the conditions for a perfect fit in a truly multi-period setting. 
Therefore, our strong compatibility framework can be viewed not only as an extension of  \cite{hull2006valuing} to the general class of conditionally i.i.d. copulas but also as a multi-period extension of the modeling approach in \cite{hull2006valuing}. 

Third, and most importantly, 
we establish the formal theoretical conditions for the existence of a perfect-fit model, or more precisely, for the existence of a general copula model capable of achieving a perfect fit. 
The constructive method of \cite{hull2006valuing} successfully achieved a perfect fit to their dataset and has proven effective in many practical market scenarios (\citealp{brigo2010credit}).
However, since they consider only a relatively narrow class of models, their method offers no guarantees of the existence of a perfect-fit model. 
Furthermore, they do not delineate the general conditions under which a perfect-fit model exists. Our paper addresses this theoretical gap by deriving the sufficient and necessary conditions for both weak compatibility and strong compatibility. In doing so, we answer the fundamental question of \textit{when} a perfect-fit model is possible, a question that precedes the task of determining \textit{how} to construct it. 

The remainder of our paper is organized as follows. 
Section \ref{TDC} 
formally defines weak compatibility and strong compatibility. 
Sections \ref{WC} and \ref{CM1} present the main results for weak compatibility and strong compatibility, respectively, whereas Section \ref{sec:appl} discusses their various applications.
Empirical results are reported in Section \ref{sec:empirical}. All the proofs and some auxiliary results are deferred to the e-companion.


\section{Definitions of Weak Compatibility and Strong Compatibility}\label{TDC}

\subsection{Valuation Framework for CDO Tranches} \label{sec_sub:value_CDO}

Before introducing the new concepts of weak compatibility and strong compatibility, we first outline the valuation framework for CDO tranches as preliminary groundwork. Consider a CDO that references a portfolio of $n$ underlying assets or names. For $j=1,2,\ldots,n$, let $\tau_j$ denote the default time of the $j$-th asset. 
According to Sklar's theorem  (\citealp{CopulaMethod}), the joint distribution function $F_{\mathrm{joint}}(t_1, \ldots, t_n)$ of $\tau_1,\ldots,\tau_n$ can be decomposed into their respective marginal distribution functions $F_1(\cdot),\ldots,F_n(\cdot)$ and a copula function $C$ that describes the dependence structure:
\begin{equation*}
	F_{\mathrm{joint}}(t_1, \ldots, t_n) = C(F_1(t_1), \ldots, F_n(t_n)) \quad\text{for 
	$t_j\geq 0$ with $j=1,\cdots,n$}.
\end{equation*}

For simplicity, we assume a constant recovery rate $R_j$ for the $j$-th name in this standard framework.
Then the cumulative loss process of the portfolio up to time $t$ is given by
\begin{equation*}
	L_t = \sum_{j = 1}^n (1-R_j)A_j \I_{\{\tau_j \leq t\}} \quad\text{for $t\geq 0$,}
\end{equation*}
where $A_j$ denotes the notional amount of the $j$-th reference entity and 
$\I_B$ (for any event $B$) represents an indicator function that equals $1$ if $B$ occurs and  $0$ otherwise.

Let $0<T_1< \cdots < T_m = T$ denote the predefined payment dates of the CDO contract, with $T_0 := 0$ and $T$ representing the contract's maturity. 
Given that its cash flows are unaffected by any defaults occurring beyond maturity $T$, we assume,   
for mathematical completeness and notational convenience, the existence of a subsequent time point $T_{m+1}>T$, by which all assets in the portfolio are certain to have defaulted. 
Then it holds that $F_j(T_{m+1}) = 1$ for all $j=1,\ldots,n$.

The CDO structure partitions the total portfolio loss into $M$ tranches, each characterized by an attachment point $a_l$ and a detachment point $b_l$ ($0 \leq a_l < b_l \leq 1$). 
These tranches are denoted by $[a_l, b_l]$ for $l=1,\ldots, M$. 
The loss allocated to tranche $[a_l, b_l]$ up to time $t$, denoted by $L_t^{[a_l, b_l]}$, is the portion of the total portfolio loss $L_t$ that falls within this tranche:
\begin{equation*}
	L_t^{[a_l, b_l]} = {(L_t - a_l)}^+ - {(L_t - b_l)}^+\quad \text{for $l=1,\ldots, M$},
\end{equation*}
where ${(x)}^+ := \max(x,0)$ for any $x\in\mathbb R$. 

The value of a CDO tranche derives from the net cash flows between the protection buyer and seller. 
These cash flows comprise two legs: the premium leg (payments from the buyer to the seller) and the default leg (payments from the seller to the buyer). 

Specifically, the discounted value of the premium leg for tranche $[a_l, b_l]$ can be expressed as the sum of any up-front payment and the discounted value of the running spread payments on the outstanding tranche notional:
\begin{equation}
	\pl^{[a_l,b_l]} = \text{uf}^{[a_l, b_l]}(b_l - a_l) + s^{[a_l, b_l]} \sum_{i=1}^m \left[D(T_i)(T_i-T_{i-1})(b_l - a_l - L^{[a_l,b_l]}_{T_i})\right]\quad \text{for $l=1,\ldots, M$},
	\notag
\end{equation}
where $\text{uf}^{[a_l, b_l]}$ and $s^{[a_l, b_l]}$ denote the up-front payment and the running spread for tranche $[a_l, b_l]$, respectively, and $D(\cdot)$ represents the risk-free discount function. 
Market conventions dictate how these parameters are quoted. 
The equity and junior tranches are typically quoted via up-front payments (with $\text{uf}^{[a_l, b_l]}$ as the market variable and $s^{[a_l, b_l]}$ as a fixed running spread), while other tranches are quoted by their running spread (with $s^{[a_l, b_l]}$ as the market variable and $\text{uf}^{[a_l, b_l]}$ as $0$). 
See Table \ref{quote_example} in Section \ref{sec:empirical} for an example of market quotes.

%

The discounted value of the default leg for tranche $[a_l, b_l]$ represents the present value of payments made by the protection seller to cover realized losses on that tranche:
\begin{equation*}
	\dl^{[a_l, b_l]} = \int_0^T D(t)dL_t^{[a_l, b_l]}\quad \text{for $l=1,\ldots, M$}.
\end{equation*}
For computational purposes, it is common in the literature to assume that defaults (and thus the incremental losses $dL_t^{[a_l, b_l]}$) occur only at discrete time points, specifically the midpoints of the $m$ coupon payment periods $[T_0, T_1],\ldots,[T_{m-1}, T_m]$ (see, e.g., \citealp{peng2008connecting} and \citealp{papageorgiou2009multiscale}). 
Under this midpoint timing assumption, we can obtain 
\begin{equation}
	\dl^{[a_l, b_l]} = \sum_{i=1}^m D\left(\frac{T_{i - 1} + T_i}{2}\right)(L_{T_i}^{[a_l, b_l]}-L_{T_{i-1}}^{[a_l, b_l]})\quad \text{for $l=1,\ldots, M$}.\notag 
	\label{dl_def}
\end{equation}

Consider a long protection position in tranche $[a_l, b_l]$. Its net present value (NPV) $V^{[a_l, b_l]}$ is then given by the value received (default leg) minus the value paid (premium leg):
\begin{equation}
	V^{[a_l, b_l]} = \dl^{[a_l, b_l]} - \pl^{[a_l,b_l]}\quad \text{for $l=1,\ldots, M$}.
	\label{vab}
\end{equation}
For pricing and calibration purposes, we are interested in the expected NPVs 
$$v^{[a_l, b_l]}:=\E[V^{[a_l, b_l]}] \quad \text{for $l=1,\ldots, M$}$$ 
under a risk-neutral measure. 
If we assume that the observed market prices (the running spreads $s^{[a_l, b_l]}$ and the up-front payments $\text{uf}^{[a_l, b_l]}$) are fair, then the theoretical expected NPV $v^{[a_l, b_l]}$ calculated using these market quotes should be zero for all traded tranches, i.e., $v^{[a_l, b_l]} = 0$ for $l=1,\ldots, M$.



\subsection{Definitions of Weak Compatibility and Strong Compatibility} \label{sec_sub:dfn_ws_compatibility}

The valuation process of CDO tranches requires several inputs. The marginal default distributions $F_1(\cdot),\ldots, F_n(\cdot)$ can typically be derived from the market prices of CDS for the individual names in the portfolio. 
The risk-free discount function $D(\cdot)$ can be obtained from market interest rate curves. 
Given these market-implied inputs, the expected NPV $v^{[a_l, b_l]}$ of a tranche depends on two key components: the dependence structure of the underlying assets, described by the copula $C$, and the market quotes for the tranches, $\mathbf{s} = (s^{[a_1,b_1]}, \uf^{[a_1,b_1]}, \ldots, s^{[a_M,b_M]}, \uf^{[a_M,b_M]})$. We can thus explicitly express the expected NPV as $v^{[a_l,b_l]}(C; \mathbf{s})$ for $l = 1, \ldots, M$.


Accordingly, calibrating a CDO model capable of perfectly fitting market prices across all tranches (i.e., identifying a copula $C$ (from a chosen family)  whose model-implied fair prices replicate the observed market prices) is equivalent to identifying a copula $C$ (from a chosen family) that simultaneously satisfies the fair value condition for all market-quoted tranches:
\begin{align}\label{eq:vC0}
	v^{[a_l, b_l]}(C; \mathbf{s}) = 0\quad \text{for $l=1,\ldots, M$}.
\end{align}
Building on this idea, we formally define compatibility of market prices across all tranches with respect to a specified family of copulas. 

\begin{Definition}\label{def: compat}({\bf Compatibility of Market Prices Across All Tranches With Respect to a Specified Family of Copulas $\mathcal{C}$})
	{\upshape Consider an $M$-tranche CDO with a known discount function $D(\cdot)$, recovery rates $R_1,\ldots,R_n$, and marginal distributions $F_1(\cdot), \ldots, F_n(\cdot)$ (typically derived from market data). Let $\mathcal{C}$ denote a specified family of $n$-dimensional copulas. 
		The observed market prices $\mathbf{s}= (s^{[a_1,b_1]}, \uf^{[a_1,b_1]}, \ldots, s^{[a_M,b_M]}, \uf^{[a_M,b_M]})$ across all $M$ tranches is said to be} compatible with respect to $\mathcal{C}$ {\upshape if there exists a copula $C \in \mathcal{C}$ such that \eqref{eq:vC0} holds.
		In this case, we also say that} the copula $C$ achieves a perfect fit to market prices across all tranches.
		\end{Definition}
		
		Consider two important general families of copulas, and then we can define two levels of compatibilities: weak compatibility and strong compatibility.
		
		\begin{Definition}\label{def: weak_compat}({\bf Weak Compatibility of Market Prices Across All Tranches})
{\upshape  Consider the general family $\mathcal{C}_0$ of copulas, which includes all $n$-dimensional copulas. Compatibility with respect to $\mathcal{C}_0$ is termed} weak compatibility.
\end{Definition}

\begin{Definition}\label{def: strong_compat}({\bf Strong Compatibility of Market Prices Across All Tranches})
{\upshape  Consider the general family $\mathcal{C}_1$ of copulas, which includes all $n$-dimensional, conditionally i.i.d. copulas. Compatibility with respect to $\mathcal{C}_1$ is termed} strong compatibility.
\end{Definition}

Since $\mathcal{C}_1$ is a subset of $\mathcal{C}_0$, strong compatibility implies weak compatibility. Indeed, weak compatibility confirms the existence of at least one general probabilistic model whose model-implied fair prices can replicate the market prices. In contrast, strong compatibility imposes a stricter requirement: such a model must belong to the class of conditionally i.i.d. models. 

Notably, weak compatibility represents the minimal requirement for the theoretical consistency of market prices. If market prices satisfy weak compatibility, there exists at least one general copula $C \in \mathcal{C}_0$ that can achieve a perfect fit to market prices across all tranches and hence can simultaneously justify all observed tranche prices. Furthermore, weak compatibility implies the absence of arbitrage opportunities among the CDO tranches. In addition, weak compatibility also enables the theoretically consistent pricing of nonstandard credit derivatives tied to the same portfolio, such as CDOs with nonstandard attachment and detachment points. Conversely, a failure to satisfy weak compatibility indicates that no single probabilistic model can rationalize the market prices across all tranches, pointing to a potential market anomaly.

In contrast, strong compatibility necessitates achieving a perfect fit to market prices across all tranches using a conditionally i.i.d. model. This requirement is also well-motivated for two key reasons. First, standard credit models, including both structural and reduced-form ones, implicitly rely on such a ``conditionally i.i.d." structure when applying the homogeneity assumption. Indeed, the dependence structure of these standard models is characterized by a set of common risk factors, and conditional on these factors, the individual default times become independent of one another. Second, this ``conditionally i.i.d." structure offers two critical advantages. (i) It enables the calibrated distribution of the common risk factors to be consistently applied in pricing nonstandard credit derivatives tied to a portfolio of difference size, such as CDOs with a nonstandard number of underlying names; (ii) It is particularly well-suited for CDO portfolio risk management, as it simplifies the estimation of loss distributions through a two-step simulation process: first, a path for the common risk factors is simulated, and then, conditional on the path, the individual defaults can be efficiently simulated as independent events.


\begin{Remark}
{\bf (Comparison with Traditional Calibration Methods of Parametric Copula Models)} {\upshape Traditional research in the CDO literature usually first specifies a particular family of parametric copula models (e.g., Gaussian copulas or Archimedean copulas) and then derives the model-implied market quotes as functions of the model parameters using \eqref{eq:vC0}. 
It is noteworthy that the use of \eqref{eq:vC0} here is indirect: while the market quotes referenced in \eqref{eq:vC0} are assumed to be actual market observations, they utilize \eqref{eq:vC0} to obtain model-implied market quotes $\mathbf{s}_{\mathrm{model}}$ satisfying $v^{[a_l, b_l]}(C; \mathbf{s}_{\mathrm{model}}) = 0$ for $l = 1 \ldots, M$. 
Finally, the calibration is completed by minimizing a distance metric between the model-implied market quotes $\mathbf{s}_{\mathrm{model}}$ and the observed market quotes $\mathbf{s}$. 
This approach is traditionally adopted in the literature primarily because specific families of parametric copula models are relatively restrictive, and as such, they generally fail to provide a perfect fit. 
Furthermore, given that the model-implied market quotes typically lack closed-form expressions and need to be computed via simulation, this approach involves solving a simulation-based optimization problem, which is computationally intensive. 

Motivated by \cite{hull2006valuing}, instead of deriving the model-implied market quotes as functions of the model parameters using \eqref{eq:vC0} and then minimizing the calibration error, we consider a general class of copula models and then directly identify a copula $C$ that solves \eqref{eq:vC0} exactly, where the market quotes are observed in the market. It turns out that for the two important general classes of copula models $\mathcal C_0$ and $\mathcal C_1$,  our proposed method offers not only greater mathematical convenience but also higher computational efficiency without resorting to simulation-based optimization techniques, as will be detailed in the subsequent sections. }  
\end{Remark}
	
	
\section{Weak Compatibility}
\label{WC}

In this section, we aim to derive a sufficient and necessary condition for weak compatibility by establishing a connection between it and a linear programming (LP) problem through the introduction of a new concept—the default probability matrix for the underlying portfolio of a CDO.  If weak compatibility is satisfied, we demonstrate how to construct a corresponding concrete copula $C\in \mathcal C_0$ that achieves a perfect fit. 

\subsection{The Default Probability Matrix (DPM)}

We begin by introducing the DPM for the underlying portfolio of a CDO. This matrix can be used to rephrase the key Eqs. \eqref{eq:vC0}, thereby transforming the problem of identifying a copula $C\in \mathcal C_0$ into that of determining the corresponding DPM.

Let $N_t := \sum_{j = 1}^n \I_{\{\tau_j \leq t\}}$ denote the default count (i.e., the number of defaults) by time $t$. The distributions of default counts at the predefined payment dates $T_1, \ldots, T_m$ can be summarized by the DPM $Q=\{q_{ij}\}_{1 \leq i \leq m, 0 \leq j \leq n}$, whose elements are defined as
\begin{equation}
	q_{ij} := \p(N_{T_i} = j)\quad\text{for $1 \leq i \leq m$ and $0 \leq j \leq n$}.\notag
\end{equation}

Given the known marginal distribution $F_j(\cdot)$ for each default time $\tau_j$, the DPM $Q$ is then determined by the copula $C$ of the default times $\tau_1, \ldots, \tau_n$. 
For any subset $S \subseteq \{1, \ldots, n\}$, let $C_S$ denote the $S$-marginal copula function of $C$ (i.e., the joint distribution function of the uniform random variables $U_k := F_k(\tau_k)$ for $k \in S$). 
Using the inclusion-exclusion principle, each element 
$q_{ij}$ of $Q$ can be derived as follows in a straightforward way. 
\begin{equation}
	q_{ij} = \sum_{\substack{S \subseteq \{1, \ldots, n\} \\ |S| \geq j}} {(-1)}^{|S| - j} \binom{|S|}{j} C_S(F_k(T_i); k \in S).
	\label{qC}
\end{equation}
This formula explicitly illustrates how the joint default probabilities (captured by the copula $C$ and the marginal distributions $F_k(\cdot)$) determine the probabilities $q_{ij}$ for $1 \leq i \leq m$ and $0 \leq j \leq n$.  

%

The following proposition establishes a sufficient and necessary condition for a matrix $\hat Q\in{\mathbb{R}^{m \times (n + 1)}}$ to qualify as a DPM corresponding to some distribution of default times.
\begin{Proposition}\label{equiv}
	For $0 = T_0 < T_1 < \cdots < T_m = T$ and a matrix $\hat Q = {\{\hat q_{ij}\}}_{1 \leq i \leq m, 0 \leq j \leq n}\in{\mathbb{R}^{m \times (n + 1)}}$, the following three statements are equivalent.
	\begin{enumerate}
		\item[(i)] The matrix $\hat{Q}$ is a DPM, that is, there exist non-negative random variables $\hat\tau_j$ (for $1\leq j \leq n$) such that $\mathbb P(\hat N_{T_i} = j) = \hat q_{ij}$, where $\hat N_t := \sum_{j=1}^n \I_{\{\hat \tau_j\leq t\}}$.     
		\item[(ii)] The matrix $\hat{Q}$ is a DPM for exchangeable default times, that is, there exist non-negative exchangeable random variables $\hat\tau_j$ (for $1\leq j \leq n$) such that $\mathbb P(\hat N_{T_i} = j) = \hat q_{ij}$, where $\hat N_t := \sum_{j=1}^n \I_{\{\hat \tau_j\leq t\}}$.     
		\item[(iii)] The matrix $\hat{Q}={\{\hat q_{ij}\}}_{1 \leq i\leq m, 0\leq j\leq n}$ satisfies the following linear constraints:
		\begin{equation}\label{qcond}
			\left\{\begin{array}{ll}
				\sum_{j=0}^n \hat q_{ij} = 1 &\quad {\rm for} \, 1 \leq i\leq m,\\
				\sum_{k\geq j} \hat q_{ik} \leq \sum_{k\geq j} \hat q_{i + 1, k} & \quad {\rm for} 1 \leq i \leq m - 1 \, {\rm and} \, 0 \leq j \leq n, \\ 
				\hat q_{ij}\geq 0 & \quad {\rm for} \, 1 \leq i \leq m \, {\rm and} \, 0\leq j\leq n.\\
			\end{array}
			\right.
		\end{equation}
	\end{enumerate}
\end{Proposition}
{\it Proof.} See Section \ref{proof11} in the e-companion. \hfill $\Box$

\begin{Remark}\label{rem:equivDPM_1}
	\upshape{The equivalence of (i) and (ii) in Proposition \ref{equiv} implies that for any credit portfolio with default times $\tau_1, \ldots, \tau_n$, one can construct an equivalent portfolio with ``exchangeable" default times $\tau_1', \ldots, \tau_n'$ that share the same DPM as the original one. Consequently, if the valuation of a credit derivative depends only on the default counts, one may assume that the distribution of $\tau_1, \ldots, \tau_n$ is exchangeable. This finding justifies the commonly used assumption in the literature (e.g., \citealp{burtschell2009comparative} and \citealp{collin2024integrated}) 
	that all assets in a portfolio share identical marginal distributions of default times and pairwise correlations. }
\end{Remark} 
\begin{Remark}\label{rem:equivDPM_2}
	\upshape{(iii) in Proposition \ref{equiv} provides a sufficient and necessary condition for a matrix $\hat{Q}$ to qualify as a valid DPM. This condition is intuitively necessary: first, the probabilities sum to one; second, cumulative default counts are non-decreasing over time. Interestingly, Proposition \ref{equiv} establishes that this intuitively necessary condition also proves sufficient. }
\end{Remark}

\subsection{Expressing the Expected NPVs $v^{[a_l, b_l]}(C; \mathbf{s})$ of Tranches in Terms of the DPM}

This subsection focuses on expressing the expected NPVs $v^{[a_l, b_l]}(C; \mathbf{s})$  (for $1\leq l \leq M$) of tranches in terms of the DPM. This, in turn, allows us to rephrase the key Eqs. \eqref{eq:vC0} using the DPM.   

In the remainder of this paper, we adopt two common assumptions from the literature (e.g., \citealp{wang09pricing}, \citealp{burtschell2009comparative}, \citealp{collin2012relative}, and \citealp{seo2018rare}). First, the notional amounts and recovery rates of all the names in the portfolio are the same, i.e., $A_j = A = \frac{1}{n}$ and $R_j = R$ for all $1\leq j \leq n$. Second, all the marginal distributions of default times are identical, i.e., $F_j(\cdot) = F(\cdot)$ for all $1\leq j \leq n$, which can be justified by Proposition \ref{equiv} (see Remark \ref{rem:equivDPM_1}).


\begin{Proposition}\label{pricingprop}
	{\bf(Expressing the Expected NPVs $v^{[a_l, b_l]}(C; \mathbf{s})$ of Tranches in Terms of the DPM)}
	\quad Given the market prices of CDO tranches $\mathbf{s} = (s^{[a_1,b_1]}, \uf^{[a_1,b_1]}, \ldots,$ $s^{[a_M,b_M]}, \uf^{[a_M,b_M]})$, then we have 
	\begin{equation}
		v^{[a_l, b_l]}(C; \mathbf{s}) = {\LA^{[a_l, b_l]}}(\mathbf{s})'Q\BA^{[a_l, b_l]} - \gamma^{[a_l, b_l]}(\mathbf{s})\quad \, {\rm for}\, l=1,\ldots,M,   
		\label{vc}
	\end{equation}
	where $Q={\{q_{ij}\}}_{1 \leq i\leq m, 0\leq j\leq n}$ is the DPM of the underlying portfolio, $\BA^{[a_l, b_l]} = (\beta_0^{[a_l, b_l]},\ldots,\beta_n^{[a_l, b_l]})'$, $\beta_j^{[a_l, b_l]} = {\left(j(1-R)/n - a_l\right)}^+ - {\left(j(1 - R)/n - b_l\right)}^+$ for $j=0,1,\ldots,n$,  $\LA^{[a_l, b_l]}(\mathbf{s}) = (\lambda_1^{[a_l, b_l]}(\mathbf{s}),$\\$ \ldots, \lambda_m^{[a_l, b_l]}(\mathbf{s}))'$, 
	\begin{align}
		\lambda_i^{[a_l, b_l]}(\mathbf{s}) &= s^{[a_l, b_l]}D(T_i)(T_i - T_{i-1}) + D\left(\frac{T_{i-1}+T_i}{2}\right) - D\left(\frac{T_{i+1}+T_i}{2}\right)\I_{\{i < m\}}
		\label{vspread_1}
	\end{align}
	for $i=1,\ldots,m$, and   		     
	\begin{align}
		\gamma^{[a_l, b_l]}(\mathbf{s}) &= (b_l - a_l) \text{uf}^{[a_l, b_l]} + (b_l - a_l)s^{[a_l, b_l]}\sum_{i=1}^m D(T_i)(T_i-T_{i-1}).
		\label{vspread_2}
	\end{align}
\end{Proposition}
{\it Proof.} See Section \ref{proof11} in the e-companion. \hfill $\Box$

\subsection{A Sufficient and Necessary Condition for Weak Compatibility}

According to Proposition~\ref{pricingprop}, identifying a copula $C\in\mathcal C_0$ via 
$v^{[a_l, b_l]}(C; \mathbf{s}) =0$ is equivalent to identifying a DPM $Q$, as the right-hand side (RHS) of \eqref{vc} depends on $C$ solely through $Q$. Moreover, Proposition~\ref{pricingprop} implies that $v^{[a_l, b_l]}$ depends affinely on the elements $q_{ij}$ ($1 \leq i \leq m$ and $0 \leq j \leq n$) of the DPM $Q$. It turns out that this nice affine structure enables the reduction of the problem of identifying a copula $C\in\mathcal C_0$ to an LP problem, or more precisely, a linear feasibility problem. Theorem \ref{ImpliedMatrix} below establishes a sufficient and necessary condition for weak compatibility in terms of an LP problem, whereas Theorem \ref{thm:construct} provides a method to construct a corresponding concrete copula $C\in\mathcal C_0$ when weak compatibility is satisfied.   


\begin{Theorem}\label{ImpliedMatrix}	
	{\bf (A Sufficient and Necessary Condition for Weak Compatibility)} 
	The market prices of CDO tranches $\mathbf{s} = (s^{[a_1,b_1]}, \uf^{[a_1,b_1]}, \ldots, s^{[a_M,b_M]}, \uf^{[a_M,b_M]})$ satisfy weak compatibility if and only if there exists $\{\hat q_{ij}\}_{1 \leq i \leq m, 0\leq j \leq n}$ that satisfies the following linear constraints:
	\begin{equation}
		\left\{
		\begin{array}{ll}
			\sum_{i=1}^m \sum_{j=0}^n \lambda_i^{[a_l,b_l]}(\mathbf{s}) \beta_j^{[a_l,b_l]} \hat q_{ij} - \gamma^{[a_l, b_l]}(\mathbf{s}) = 0 &\quad\text{\upshape{for} $1\leq l \leq M$},\\
			\sum_{j=0}^n \hat q_{ij} = 1 &\quad\text{\upshape{for} $1 \leq i\leq m$},\\
			\sum_{j=0}^n j\hat q_{ij} = n F(T_i) &\quad\text{\upshape{for} $1 \leq i \leq m$}, \\
			\sum_{k\geq j} \hat q_{ik} \leq \sum_{k\geq j} \hat q_{i + 1, k} &\quad\text{\upshape{for} $1 \leq i\leq m - 1$ \upshape{and} $0\leq j \leq n$},\\
			\hat q_{ij} \geq 0 &\quad \text{\upshape{for} $1 \leq i \leq m$ \upshape{and} $0\leq j\leq n$.}\\
		\end{array}
		\right.
		\label{linearsystem}
	\end{equation}
\end{Theorem}
{\it Proof.} See Section \ref{proof12} in the e-companion.\hfill $\Box$

\begin{Theorem}\label{thm:construct}	
	{\bf (A Method to Construct a Copula $\hat C\in\mathcal C_0$ When Weak Compatibility is Satisfied)} 
	Suppose that ${\{\hat q_{ij}\}}_{1\leq i \leq m, 0\leq j \leq n}$ is a solution to \eqref{linearsystem}. 
	Then there exists an exchangeable copula $\hat{C}$ that perfectly fits all market prices; specifically, $\hat{C}$ is defined as follows.
	Given $y_j \in \{0, 1, \ldots, m + 1\}$ (for $1 \leq j \leq n$), define $u_j := F(T_{y_j})$. Then the value of the copula $\hat C$ at $(u_1, \ldots, u_n)$ is defined as  
	\begin{equation}
		\hat C(u_1, \ldots, u_n) := \frac{1}{n!} \sum_{\sigma \in \mathcal{G}_n} \min_{1 \leq j \leq n}\left\{\sum_{k \geq \sigma(j)}\hat{q}_{y_j, k}\right\},
		\label{implied_copula}
	\end{equation}
	where $\mathcal{G}_n$ is the set of all permutations on the set $\{1, 2, \ldots, n\}$.
\end{Theorem}
{\it Proof.} See Section \ref{proof12} in the e-companion.\hfill $\Box$

\subsection{An Efficient Algorithm for Verifying Weak Compatibility and Constructing a Corresponding Concrete Copula When Satisfied}\label{sec_sub:weak_algo}

Based on Theorems \ref{ImpliedMatrix} and \ref{thm:construct}, we can develop the following algorithm to verify weak compatibility and construct a corresponding concrete copula when weak compatibility is satisfied. Since the core component of this algorithm is to solve an LP problem, or more precisely, a linear feasibility problem, the algorithm exhibits high efficiency. 
\vspace{0.1cm}

\begin{breakablealgorithm}{\bf Verifying Weak Compatibility and Constructing a Concrete Copula}
	\item[\textbf{Step 1:}] \textbf{Calculate Coefficients.} Compute the coefficients $\lambda_{i}^{[a_{l},b_{l}]}(\mathbf{s})$, $\gamma^{[a_{l},b_{l}]}(\mathbf{s})$, and $\beta_{j}^{[a_{l},b_{l}]}$ based on market data, as defined in Proposition~\ref{pricingprop}. 
	\item[\textbf{Step 2:}] \textbf{Verify Weak Compatibility by Solving an LP Problem.} Solve the linear feasibility problem \eqref{linearsystem} for the DPM. 
	If a feasible solution exists, the market prices satisfy weak compatibility; otherwise, they do not. 
	\item[\textbf{Step 3:}] \textbf{Construct a Corresponding Concrete Copula When Weak Compatibility is Satisfied.} If the market prices satisfy weak compatibility, use a feasible DPM $\hat{Q}$ obtained in Step 2 to construct a concrete copula $\hat C\in\mathcal C_0$ as defined in \eqref{implied_copula}. 
\end{breakablealgorithm}\label{algo:weak}
\vspace{0.1cm}

Denote by $\mathcal{Q}$ the feasible region of \eqref{linearsystem}. We refer to a feasible $\hat{Q}\in\mathcal{Q}$, if it exists, as a ``market-implied DPM" for the CDO. The market implied DPM may not be unique. Indeed, in such a case, the set $\mathcal{Q}$ encapsulates all possible DPMs that are consistent with the observed market prices of standard CDO tranches. The non-uniqueness of the market-implied DPM is reasonable. This is because $\mathcal{C}_{0}$ is extensive, encompassing all possible CDO models, while the observable market data, namely the prices of a finite number of standard CDO tranches, provides only partial information regarding the underlying CDO model. Consequently, multiple CDO models may exist, and accordingly multiple market-implied DPMs, all of which perfectly fit the market prices. Essentially, the feasible set $\mathcal{Q}$ precisely characterizes all possible CDO models that achieve a perfect fit to the market prices. 


\subsection{Extension to Weak Bid-Ask Compatibility}
\label{sec:bid-ask}
In practice, market prices of CDO tranches are quoted in the form of bid prices and ask prices. Notably, our weak compatibility framework can be readily extended to incorporate the market bid-ask prices. In such a scenario, 
rather than perfectly fitting a single price point, the requirement is that there must exist a copula model whose implied fair price lies within the observed bid-ask interval for all tranches.

We introduce a new concept termed \textbf{Weak Bid-Ask Compatibility} to address this issue, i.e., the theoretical consistency of market bid-ask prices of CDO tranches.
For any copula $C$, its model-implied fair price denoted by $\mathbf{s}_{\text{model}}$ refers to the price that equates the expected NPVs of a long protection position to zero, i.e., $v^{[a_l, b_l]}(C; \mathbf{s}_{\text{model}}) = 0$ for $1\leq l \leq M$.
Weak bid-ask compatibility, in turn, requires the existence of a copula $C \in \mathcal{C}_{0}$ whose model-implied fair price $\mathbf{s}_{\text{model}}$ is between the market bid price $\mathbf{s}_{\text{bid}}$ and ask price $\mathbf{s}_{\text{ask}}$: 
\begin{align}\label{eq:between}
	\mathbf{s}_{\text{bid}} \leq \mathbf{s}_{\text{model}} \leq \mathbf{s}_{\text{ask}}.
\end{align}
Since the expected NPV $v^{[a_l, b_l]}(C; \mathbf{s})$ is a monotonically decreasing function of the price $\mathbf{s}$ for each tranche $l=1,\ldots, M$, ``$\mathbf{s}_{\text{bid}} \leq \mathbf{s}_{\text{model}}$" and ``$\mathbf{s}_{\text{model}} \leq \mathbf{s}_{\text{ask}}$" are equivalent to ``$v^{[a_l, b_l]}(C; \mathbf{s}_{\text{bid}}) \geq v^{[a_l, b_l]}(C; \mathbf{s}_{\text{model}}) = 0$" and ``$v^{[a_l, b_l]}(C; \mathbf{s}_{\text{ask}})\leq v^{[a_l, b_l]}(C; \mathbf{s}_{\text{model}}) = 0 $", respectively. Therefore, condition \eqref{eq:between} can be re-expressed as:
\begin{align}\label{eq:between_equiv}
	v^{[a_l, b_l]}(C; \mathbf{s}_{\text{bid}}) \geq 0\quad \textrm{and}\quad v^{[a_l, b_l]}(C; \mathbf{s}_{\text{ask}})\leq 0 \quad \text{for } l = 1, \ldots, M.
\end{align}
Compared with condition \eqref{eq:between}, condition \eqref{eq:between_equiv} is more straightforward  to verify. Accordingly, we adopt it as the formal definition of weak bid-ask compatibility. 
\begin{Definition}[Weak Bid-Ask Compatibility]\label{def:bid-ask-weak}
    {\upshape  
    Market bid-ask prices $(\mathbf{s}_{\text{bid}},\mathbf{s}_{\text{ask}})$ of all $M$ CDO tranches are said to be \emph{weakly bid-ask compatible} if there exists a copula $C \in \mathcal{C}_{0}$ such that the expected NPVs from a long protection position satisfy \eqref{eq:between_equiv}.
}
\end{Definition}

The verification of weak bid-ask compatibility is analogous to Theorem \ref{ImpliedMatrix}, leading to the following Proposition \ref{weak_bid-ask-corollary} which reduces the verification problem to a linear feasibility problem. 
Similar to Algorithm \ref{algo:weak}, we can also develop an algorithm to verify weak bid-ask compatibility and, when satisfied, constructing a corresponding concrete copula. The details are omitted due to the analogy.

\begin{Proposition}[A Sufficient and Necessary Condition for Weak Bid-Ask Compatibility]\label{weak_bid-ask-corollary}
    The market bid-ask prices ($\mathbf{s}_{\text{bid}}, \mathbf{s}_{\text{ask}}$) satisfy weak bid-ask compatibility if and only if there exists a matrix $\{\hat{q}_{ij}\}_{1 \leq i \leq m, 0 \leq j \leq n}$ that satisfies the following linear constraints:
    \begin{equation}\label{eq:weak_bid_ask}
        \begin{cases}
    \sum_{i=1}^{m}\sum_{j=0}^{n}\lambda_{i}^{[a_{l},b_{l}]}(\mathbf{s}_{\text{bid}})\beta_{j}^{[a_{l},b_{l}]}\hat{q}_{ij}-\gamma^{[a_{l},b_{l}]}(\mathbf{s}_{\text{bid}})\geq 0 &\quad \text{\upshape{for} $1\leq l\leq M$}, \\
    \sum_{i=1}^{m}\sum_{j=0}^{n}\lambda_{i}^{[a_{l},b_{l}]}(\mathbf{s}_{\text{ask}})\beta_{j}^{[a_{l},b_{l}]}\hat{q}_{ij}-\gamma^{[a_{l},b_{l}]}(\mathbf{s}_{\text{ask}})\leq 0 &\quad \text{\upshape{for} $1\leq l\leq M$},\\
    			\sum_{j=0}^n \hat q_{ij} = 1 &\quad\text{\upshape{for} $1 \leq i\leq m$},\\
    \sum_{j=0}^n j\hat q_{ij} = n F(T_i) &\quad\text{\upshape{for} $1 \leq i \leq m$}, \\
    \sum_{k\geq j} \hat q_{ik} \leq \sum_{k\geq j} \hat q_{i + 1, k} &\quad\text{\upshape{for} $1 \leq i\leq m - 1$ \upshape{and} $0\leq j \leq n$},\\
    \hat q_{ij} \geq 0 &\quad \text{\upshape{for} $1 \leq i \leq m$ \upshape{and} $0\leq j\leq n$}.\\
        \end{cases}
    \end{equation}
\end{Proposition}
\noindent\emph{Proof.} The proof is similar to that of Theorem \ref{ImpliedMatrix} and is therefore omitted. \hfill $\Box$

\section{Strong Compatibility}\label{CM1}

In this section, we intend to derive a sufficient and necessary condition for strong compatibility. Additionally, if this condition is satisfied, we will provide a method to construct a corresponding concrete copula $C\in\mathcal C_1$ that achieves a perfect fit. Our approach is grounded in the theory of stochastic distortions. 
Specifically, we introduce a new and flexible family of stochastic distortions, termed \emph{gamma distortions}, which allows us to reframe the verification of strong compatibility as a series of LP problems, or more precisely, linear feasibility problems. 
Solving these LP problems yields not only a sufficient and necessary condition for strong compatibility but also the parameters required to construct a corresponding concrete copula $C\in\mathcal C_1$. 

\subsection{Conditionally I.I.D. Models and Stochastic Distortions}

As preliminary preparations, we first introduce the concept of stochastic distortions and then establish the close connection between conditionally i.i.d. copulas and stochastic distortions. 
\begin{Definition}[Stochastic Distortion, \citealp{Lin2018}]
	{\upshape A stochastic process $\{X(u):u\in[0,1]\}$ 
		is called a \emph{stochastic distortion} if it is non-decreasing a.s. (that is, $X(u_1) \leq X(u_2)$ a.s. for any $u_1 < u_2$) and satisfies that $X(0) = 0$ and $X(1) = 1$ a.s..}
\end{Definition}

If a stochastic distortion $X(u)$ satisfies $\E[X(u)] = u$ for any $u\in [0, 1]$, then 
\begin{equation}
	C^{X}(u_1, \ldots, u_n) := \E\left[\prod_{j=1}^n X(u_j)\right] \quad \text{for $(u_1, \ldots, u_n) \in [0, 1]^n$}
	\label{DistortedCopula}
\end{equation}
is a copula function (see Theorem 3.1 of \citealp{Lin2018}). In our paper, we will call 
$C^{X}$ 
the \emph{transformed copula by the stochastic distortion $X$}.

In Proposition \ref{ExchangeProp} below, we establish the equivalence between conditionally i.i.d. copulas and the existence of specific stochastic distortion representations. 
As a direct consequence,  each conditionally i.i.d.\ copula can be represented as a transformed copula by a stochastic distortion.

\begin{Proposition}[Stochastic Distortion Representations of Conditionally I.I.D. Copulas]\label{ExchangeProp} 
	Let ${\{U_k\}}_{k \in \mathbb{N}^+}$ be an infinite sequence of uniformly distributed random variables over $(0,1)$ (abbreviated as $\UU(0, 1)$ random variables hereafter) on a probability space $(\Omega, \mathcal{F}, \mathbb{P})$. Then the following statements are equivalent. 
	\begin{enumerate}
		\renewcommand{\labelenumi}{(\roman{enumi})}
		\item ${\{U_k\}}_{k \in \mathbb{N}^+}$ are i.i.d.\ conditional on some $\sigma$-field $\mathcal{H} \subseteq \mathcal{F}$.
		\item There exists a stochastic distortion $X$ with $\E[X(u)] = u$ for all $u \in [0, 1]$ such that for any $n\geq 2$, the copula $C$ of $(U_1, \ldots, U_n)$ is identical to the transformed copula $C^X$ by the stochastic distortion $X$, i.e., 
		\begin{equation}
			C(u_1, \ldots, u_n) = C^{X}(u_1, \ldots, u_n).
			\label{Dist2} 
		\end{equation}
			\end{enumerate}
\end{Proposition}
{\it Proof.} See Section \ref{proof21} in the e-companion.\hfill $\Box$

\begin{Remark}
{\upshape According to the celebrated de Finetti's Theorem (see, e.g., \citealp{mai2020infinite}), (i) in Proposition \ref{ExchangeProp} is equivalent to the statement that ${\{U_k\}}_{k \in \mathbb{N}^+}$ are infinitely exchangeable.}
\end{Remark}

\subsection{Valuation of CDO Tranches under Conditionally I.I.D. Models}

Next, we discuss the valuation methodology of CDO tranches under a conditionally i.i.d. model, where $\tau_1, \ldots, \tau_n$ are assumed to be i.i.d. conditional on some common factors represented by a $\sigma$-field $\mathcal{H}$. As established in Proposition~\ref{ExchangeProp}, such a conditionally i.i.d. structure implies the existence of a stochastic distortion $X$ that satisfies $\E[X(u)]=u$ and is given explicitly by 
$X(u) = \p[F(\tau_1)\leq u | \mathcal{H}].$ 
Then we can deduce that 
\begin{align}\label{eq:qij_cond_iid}
q_{ij} = \p (N_{T_i} = j) = \E[\p(N_{T_i} = j|\mathcal{H})] = \E\left[\binom{n}{j}X(F(T_i))^j(1 - X(F(T_i)))^{n-j}\right] 
\end{align}
for $1\leq i \leq m$ and $0 \leq j \leq n$. Plugging \eqref{eq:qij_cond_iid} into \eqref{vc} yields immediately the following explicit expression for the expected NPVs $v^{[a_l, b_l]}(C^X; \mathbf{s})$ of tranches
in terms of the stochastic distortion $X$: 
\begin{equation}
	v^{[a_l, b_l]}(C^X; \mathbf{s}) = \sum_{i = 1}^m \sum_{j = 0}^n \lambda^{[a_l, b_l]}_i(\mathbf{s}) \beta_j^{[a_l ,b_l]}  \binom{n}{j} \E\left[ X(F(T_i))^j (1-X(F(T_i)))^{n - j} \right] - \gamma^{[a_l, b_l]}(\mathbf{s})
	\label{distortion_v}
\end{equation}
for $1 \leq l \leq M.$

\begin{Remark}[Immediate Applications to Parametric Copula Models]
{\upshape The above valuation me-\\-thodology of CDO tranches based on stochastic distortions applies directly to many widely used families of parametric copula models, including Gaussian, Archimedean, and multivariate Fr\'echet copulas, 
as they are conditionally i.i.d. and thus admit stochastic distortion representations. Consequently, the expected tranche NPVs $v^{[a_l, b_l]}$ for any of these models can be computed by substituting the model's specific stochastic distortion $X$ into the general valuation formula \eqref{distortion_v}. 
}
\end{Remark}

\subsection{Gamma Distortions and Gamma-Distorted Copulas}

This subsection introduces a novel family of stochastic distortions, termed \textit{gamma distortions}. 
We show that applying a discrete form of the \emph{transformed copula by gamma distortion} to model default times, each element of the resulting DPM $Q$ is then a linear combination of the underlying discrete state probabilities, thereby facilitating the related analysis. 

The construction of gamma distortions is based on the gamma process, a well-known stochastic process extensively studied 
in the literature (see, e.g.,~\citealp{applebaum2009levy} and~\citealp{schoutens2010levy}). 
Recall that a gamma process $\{\Gamma(t; \gamma, \lambda):t\geq 0\}$ with parameters $\gamma>0$ and $\lambda>0$ is a pure-jump L\'evy process, which is increasing a.s. and satisfies that for any fixed time $t>0$, $\Gamma(t; \gamma, \lambda)$ follows a gamma distribution with shape parameter $\gamma t$ and rate parameter $\lambda$.

Consider an $N \in\mathbb N^+$ and two independent gamma processes  $\{\xi_t: t\geq 0\}$ and $\{\eta_t: t\geq 0\}$ with both parameters $\gamma$ and $\lambda$ equal to $1$. 
Suppose that $\{\phi(u):0 \leq u \leq 1\}$ is a non-decreasing stochastic process that is independent of $\{\xi_t\}$ and $\{\eta_t\}$ and satisfies the following conditions:
\begin{equation}
	\phi(0) = 0 \quad \text{a.s.}, \quad \phi(1) = N \quad  \text{a.s.}, \quad\text{and}\quad 
	\E[\phi(u)] = Nu.
	\label{cond_phipsi}
\end{equation}
Define a stochastic process $\{X(u):u\in[0,1]\}$ as follows, which proves to be a stochastic distortion. 
\begin{equation}
	X(u) := \frac{\xi_{\phi(u)}}{\xi_{\phi(u)} + \eta_{N - \phi(u)}} \quad \text{for $0 \leq u \leq 1$}.
	\label{Xdef}
\end{equation}
\vspace{-0.9cm}
\begin{Proposition}\label{propx}
	The stochastic process $X$ defined in \eqref{Xdef} is a stochastic distortion and satisfies  $\E[X(u)] = u$ for $0 \leq u \leq 1$.
\end{Proposition}
{\it Proof.} See Section \ref{proof21} in the e-companion. \hfill $\Box$

\begin{Definition}[Gamma Distortions, Gamma-Distorted Copulas, and Their Generators]
	{\upshape The stochastic distortion $X$ defined in \eqref{Xdef} is referred to as a \emph{gamma distortion}, the transformed copula by the gamma distortion $X$ is called a \emph{gamma-distorted copula}, and the involved non-decreasing stochastic process $\{\phi(u):0 \leq u \leq 1\}$ is called the \emph{generator of the gamma distortion $X$} or \emph{the generator of the gamma-distorted copula}.}
\end{Definition}
Below we introduce a particular type of gamma distortions and gamma-distorted copulas with a specific ``discrete" structure.

\begin{Definition}[Discrete Gamma Distortions and Discrete Gamma-Distorted Copulas]
	{\upshape 
		If the generator $\phi$  of a gamma distortion $X$ satisfies an additional ``discrete" condition that for each $i =1, \dots, m$, the random variable $\phi(F(T_i))$ takes only discrete values in $\{0, 1, \ldots, N\}$ with the probability mass function $p_{ik} = \p(\phi(F(T_i)) = k)$ for any $k=0,1,\ldots, N$, we refer to the associated gamma distortion $X$ and gamma-distorted copula as a \emph{discrete gamma distortion} and a \emph{discrete gamma-distorted copula}, respectively.}
\end{Definition}



Using \eqref {eq:qij_cond_iid}, we can obtain that for a discrete gamma-distorted copula, the elements of the associated DPM can be expressed as a linear combination of the probabilities $p_{ik}$ for $i =1, \dots, m$ and $k=0,1,\ldots, N$:
	\begin{equation} \label{qij_discrete}
	q_{ij} = 
	\sum_{k=0}^N h_{jk}^{(N)} p_{ik} \quad \text{for  $1\leq i \leq m$ and $0 \leq j \leq n$},
	\end{equation}
	where the coefficients $h_{jk}^{(N)}$ are defined as 
	\begin{equation}
		\begin{array}{ll}
			h_{j0}^{(N)} := \I_{\{j = 0\}} &\quad\text{for $0\leq j \leq n$},\\
			h_{jk}^{(N)} :=  \binom{n}{j} \frac{\B(k + j , N + n - k - j)}{\B(k, N - k)} &\quad \text{for $0 \leq j \leq n$ and $1 \leq k \leq N - 1$},\\
			h_{jN}^{(N)} := \I_{\{j = n\}} &\quad\text{for $0\leq j \leq n$}, 
		\end{array}
		\label{define_h}
	\end{equation}
	with $\B(\cdot,\cdot)$ denoting the Beta function.

\subsection{A Sufficient and Necessary Condition for Strong Compatibility}

While discrete gamma-distorted copulas represent only a specific subclass of conditionally i.i.d. copulas, this specific family exhibits a remarkable universality. Specifically, we find that for any of the entire set of strongly compatible market prices except those lying on the boundary of this set (this boundary has a zero Lebesgue measure), there always exists a discrete gamma-distorted copula that achieves a perfect fit. Furthermore, whether market prices of CDO tranches can be perfectly fit by such a copula---and thus whether they are strongly compatible---can be verified efficiently by solving LP problems. 

Theorem \ref{strongly_compatible} below provides a sufficient and necessary condition for strong compatibility. When the condition is satisfied, Theorem \ref{construct_copula} below offers a method to construct a corresponding concrete conditionally i.i.d. copula $C\in\mathcal C_1$ that achieves a perfect fit. 

\begin{Theorem}\label{strongly_compatible} {\bf{(A Sufficient and Necessary Condition for Strong Compatibility)}}
	Let $\mathcal{A}\subseteq \mathbb{R}^{2M}$ denote the set of strongly compatible market prices and $\mathcal{A}^\circ$ denote the interior of $\mathcal{A}$.
	Then for each vector $\mathbf{s} \in \mathcal{A}^\circ$, there exists an $N\in\mathbb N^+$ such that the following system of linear constraints for $\{\hat p_{ik}\}_{1 \leq i \leq m, 0 \leq k \leq N}$ has a feasible solution.
	\begin{equation}
		\begin{cases}
			\sum_{i=1}^m \sum_{k=0}^N g_{ikl}^{(N)}(\mathbf{s}) \hat p_{ik} - \gamma^{[a_l, b_l]}(\mathbf{s}) = 0 \, &\quad\text{\upshape{for} $1\leq l \leq M$},\\
			\sum_{k=0}^N \hat p_{ik} = 1 \, &\quad\text{\upshape{for} $1 \leq i\leq m$},\\
			\sum_{k=0}^N k\hat p_{ik} = N F(T_i) \, &\quad\text{\upshape{for} $1 \leq i \leq m$}, \\
			\sum_{j \geq k} \hat p_{ij} \leq \sum_{j \geq k} \hat p_{i + 1, j} \, &\quad\text{\upshape{for} $1 \leq i \leq m - 1$ and $0 \leq k \leq N$}, \\
			\hat p_{ik} \geq 0 \, &\quad\text{\upshape{for} $1 \leq i \leq m$ and $0\leq k\leq N$},\\
		\end{cases}
		\label{linearsystem2}
	\end{equation}
	where $g_{ikl}^{(N)}(\mathbf{s}) = \sum_{j=0}^n h_{jk}^{(N)} \lambda_i^{[a_l,b_l]}(\mathbf{s}) \beta_j^{[a_l,b_l]}$ for $1 \leq i \leq m$, $0 \leq k \leq N$, and $1 \leq l \leq M$, 
	 $h_{jk}^{(N)}$ is defined in \eqref{define_h}, and $\lambda_i^{[a_l,b_l]}(\mathbf{s})$ and  $\beta_j^{[a_l,b_l]}$ are given in Proposition \ref{pricingprop}.
	
	Conversely, if for the market prices $\mathbf{s}$, there exists an $N \in\mathbb N^+$ such that the system of linear constraints \eqref{linearsystem2} has a feasible solution, then the market prices $\mathbf{s}$ is strongly compatible.
\end{Theorem}
{\it Proof.} See Section \ref{proof22} in the e-companion. \hfill $\Box$

\begin{Remark}
	{\upshape The sufficient and necessary condition provided in Theorem \ref{strongly_compatible} is indeed a sufficient and necessary condition up to the set $\mathcal{A}$'s boundary with zero Lebesgue measure, where the set $\mathcal{A}$ denotes the set of strongly compatible market prices. Since the set $\mathcal{A}$'s boundary has a zero Lebesgue measure, it suffices for practical applications. 
	}
\end{Remark}

As for the first part of Theorem \ref{strongly_compatible}, we can obtain a stronger result in the following Proposition \ref{prop_interior}, which will be used to prove that Algorithm  \ref{strong_compatibility_algorithm} will terminate in a finite number of iterations. 

\begin{Proposition}\label{prop_interior}
	For each vector $\mathbf{s} \in \mathcal{A}^\circ$, there exists an $N_0\in\mathbb N^+$ such that the system of linear constraints \eqref{linearsystem2} has a feasible solution for any $N \geq N_0$.
\end{Proposition}
{\it Proof.} See Section \ref{proof23} in the e-companion. \hfill $\Box$

Next, we shall develop a method to construct a concrete copula $C\in \mathcal{C}_1$ when strong compatibility is satisfied. Indeed, in such a case, we can construct a discrete gamma-distorted copula. 
\begin{Proposition}\label{prop:generator}
	Suppose that for some $N \in\mathbb N^+$, $\{\hat p_{ik}\}_{1 \leq i \leq m, 0 \leq k \leq N}$ is a feasible solution to the system of linear constraints \eqref{linearsystem2_delete_1}. 
	\begin{equation}
		\begin{cases}
			\sum_{k=0}^N \hat p_{ik} = 1 \, &\quad\text{\upshape{for} $1 \leq i\leq m$},\\
			\sum_{k=0}^N k\hat p_{ik} = N F(T_i) \, &\quad\text{\upshape{for} $1 \leq i \leq m$}, \\
			\sum_{j \geq k} \hat p_{ij} \leq \sum_{j \geq k} \hat p_{i + 1, j} \, &\quad\text{\upshape{for} $1 \leq i \leq m - 1$ and $0 \leq k \leq N$}, \\
			\hat p_{ik} \geq 0 \, &\quad\text{\upshape{for} $1 \leq i \leq m$ and $0\leq k\leq N$}.
		\end{cases}
		\label{linearsystem2_delete_1}
	\end{equation}
	Then the stochastic process $\{\phi(u):0\leq u\leq 1\}$ defined below is non-decreasing and  satisfies \eqref{cond_phipsi}.
	\begin{equation}
		\phi(u) := \begin{cases}
			\sum_{k=0}^N k\I_{\left\{\sum_{j<k}\hat{p}_{ij} < U \leq \sum_{j\leq k}\hat{p}_{ij}\right\}}  &\text{\upshape{for} $u = F(T_i)$ and $i=0,\cdots\!,m+1$},\\
			\frac{(F(T_{i+1}) - u)\phi(F(T_i))+(u - F(T_i))\phi(F(T_{i + 1}))}{F(T_{i + 1}) - F(T_i)}\!\! &\text{\upshape{for} $u\in\! ( F(T_i), F(T_{i + 1}))$ and $i=0,\cdots\!,m$},  
		\end{cases}
		\label{implied_phi}
	\end{equation}
	where $U\sim \UU(0, 1)$, $\hat{p}_{0k} := \I_{\{k = 0\}}$ for $k=0,1,\ldots, N$, and $\hat{p}_{m + 1, k} := \I_{\{k = N\}}$ for $k=0,1,\ldots, N$. 	
	Moreover, the stochastic process $\{\phi(u):0\leq u\leq 1\}$ satisfies that 
	\begin{equation}\label{eq:phi_distribution}
		\p(\phi(F(T_i)) = k) = \hat p_{ik} \quad \text{\upshape{for} $1 \leq i \leq m$ \upshape{and} $0 \leq k \leq N$}.
	\end{equation}
\end{Proposition}
{\it Proof.} See Section \ref{proof21} in the e-companion. \hfill $\Box$

Proposition \ref{prop:generator} implies that any feasible solution to the system of linear constraints \eqref{linearsystem2_delete_1} can induce a stochastic process $\{\phi(u):0\leq u\leq 1\}$ that is qualified to serve as the generator of a discrete gamma-distorted copula. 
Theorem \ref{construct_copula} below will show that if the feasible solution also satisfies the first set of constraints of \eqref{linearsystem2},  the discrete gamma-distorted copula constructed based on $\{\phi(u):0\leq u\leq 1\}$ can perfectly fit the market prices across all tranches.

\begin{Theorem}\label{construct_copula}{\bf {(A Method to Construct a Copula $C\in \mathcal{C}_1$ When Strong Compatibility is Satisfied)}}
	Suppose that for some $N \in\mathbb N^+$, $\{\hat p_{ik}\}_{1 \leq i \leq m, 0 \leq k \leq N}$ is a feasible solution to the system of linear constraints \eqref{linearsystem2}. Then the discrete gamma-distorted copula $C^X$ with the generator $\phi$ defined in \eqref{implied_phi} via this feasible solution 
	perfectly fits the market prices across all tranches $[a_1, b_1], \ldots, [a_M, b_M]$. 
\end{Theorem}
{\it Proof.} See Section \ref{proof22} in the e-companion. \hfill $\Box$

\subsection{An Efficient Algorithm for Verifying Strong Compatibility and Constructing a Corresponding Concrete Conditionally i.i.d. Copula When Satisfied}

Unlike weak compatibility, verifying strong compatibility cannot be reduced to a single LP problem. 
This complexity arises because the linear feasibility problem \eqref{linearsystem2} depends on the parameter $N$, which is unknown \textit{a priori}. 
To overcome this, we propose an algorithm that iteratively verifies compatibility, proceeding sequentially across tranches. 
For each tranche, the algorithm solves a series of LP problems for an increasing sequence of $N$ values. 
Despite this iterative structure, the process is efficient because each step only involves solving computationally tractable LP problems.

\subsubsection{The Key Ideas of Our Iterative Algorithm}\label{sec_subsub:key_strong_comp_algo} 
Before presenting our iterative algorithm, we first articulate the key ideas. 
Let $\mathcal{A}^l \subseteq \mathbb{R}^{2l}$ denote the set of strongly compatible market prices for the first $l$ tranches (for $1\leq l\leq M$), and let $(\mathcal{A}^l)^\circ$ denote its interior. 
We proceed by assuming that $(s^{[a_1, b_1]}, \uf^{[a_1, b_1]}, \ldots, s^{[a_l, b_l]}, \uf^{[a_l, b_l]})\in(\mathcal{A}^l)^\circ$.
Consider a system of linear constraints, which is the same as \eqref{linearsystem2} except that the first set of linear constraints of \eqref{linearsystem2} are replaced by $\sum_{i=1}^m \sum_{k=0}^N g_{ikl'}^{(N)}(\mathbf{s}) \hat p_{ik} - \gamma^{[a_{l'}, b_{l'}]}(\mathbf{s}) = 0$ for $1\leq l' \leq l$. Proposition~\ref{prop_interior} ensures that this system of linear constraints is feasible for all sufficiently large $N$, and we denote its feasible region by $\mathcal{P}_N^l$.


A crucial step in our iterative algorithm is that given that the prices of the first $l$ tranches are strongly compatible, we need to determine the price range for the $(l+1)$-th tranche that preserves strong compatibility. 
If the market price of the $(l+1)$-th tranche falls within this range, we conclude that the prices of the first $l+1$ tranches are strongly compatible and proceed to the next tranche. 
Otherwise, the full set of $M$ market prices is not strongly compatible, and the algorithm terminates.

The determination of this strong compatibility price range depends on how the $(l+1)$-th tranche is quoted. 
If it is quoted via up-front payment with fixed spread $s^{[a_{l+1}, b_{l+1}]}$, we denote the range by $[\underline{\uf}^{[a_{l+1}, b_{l+1}]}, \overline{\uf}^{[a_{l+1}, b_{l+1}]}]$, where the upper bound $\overline{\uf}^{[a_{l+1}, b_{l+1}]}$ and the lower bound $\underline{\uf}^{[a_{l+1}, b_{l+1}]}$ are given by
\begin{equation}
	\begin{cases}
		\overline{\uf}^{[a_{l+1}, b_{l+1}]} = \sup\left\{x \in \mathbb{R}: (s^{[a_1, b_1]}, \uf^{[a_1,b_1]}, \ldots, s^{[a_l, b_l]}, \uf^{[a_l,b_l]}, s^{[a_{l+1}, b_{l+1}]}, x)\in \mathcal{A}^{l+1}\right\},\\
		\underline{\uf}^{[a_{l+1}, b_{l+1}]} = \inf\left\{x \in \mathbb{R}: (s^{[a_1, b_1]}, \uf^{[a_1,b_1]}, \ldots, s^{[a_l, b_l]}, \uf^{[a_l,b_l]}, s^{[a_{l+1}, b_{l+1}]}, x)\in \mathcal{A}^{l+1}\right\}.
	\end{cases}
	\label{eq:uf_ub_lb}
\end{equation}
If the $(l+1)$-th tranche is quoted via spread, the range $[\underline{s}^{[a_{l+1}, b_{l+1}]}, \overline{s}^{[a_{l+1}, b_{l+1}]}]$ is given analogously.

However, the strong compatibility price range given above is not directly computable. Notably, we can propose an indirect approach to computing it. Specifically, we can define a sequence of the so-called ``$N$-dependent price ranges'', which can be computed easily and furthermore can be proved to converge to the strong compatibility price range as $N$ tends to infinity.  
For any fixed $N \in \mathbb{N}^+$ for which the feasible set $\mathcal{P}_N^l$ is non-empty, we can compute the $N$-dependent price range, $[\underline{\uf}^{[a_{l+1}, b_{l+1}]}_N, \overline{\uf}^{[a_{l+1}, b_{l+1}]}_N]$ or $[\underline{s}^{[a_{l+1}, b_{l+1}]}_N, \overline{s}^{[a_{l+1}, b_{l+1}]}_N]$, by solving a series of LP or linear fractional programming (LFP) problems (LFP problems can also be readily transformed into LP problems). We compute the bounds of these $N$-dependent price ranges in the following way.
\begin{itemize}[leftmargin=11pt]
    \item If the $(l+1)$-th tranche is quoted via up-front payment (with fixed spread $s^{[a_{l+1}, b_{l+1}]}$), $\overline{\uf}^{[a_{l+1}, b_{l+1}]}_N$ and $\underline{\uf}^{[a_{l+1}, b_{l+1}]}_N$ can be computed by solving the following LP problems:
    \begin{equation}
    	\begin{cases}
    		\overline{\uf}^{[a_{l+1}, b_{l+1}]}_N =\ \sup\limits_{\{\hat p_{ik}\} \in \mathcal{P}_N^l} \frac{\sum_{i=1}^m \sum_{k=0}^N g_{i,k,l+1}^{(N)}(\mathbf{s})\hat{p}_{ik}}{{b_{l+1} - a_{l+1}}} - s^{[a_{l+1}, b_{l+1}]}\sum_{i=1}^m D(T_i)(T_i - T_{i-1}),  \\
    		 \underline{\uf}^{[a_{l+1}, b_{l+1}]}_N =\ \inf\limits_{\{\hat p_{ik}\} \in \mathcal{P}_N^l} \frac{\sum_{i=1}^m \sum_{k=0}^N g_{i,k,l+1}^{(N)}(\mathbf{s})\hat{p}_{ik}}{{b_{l+1} - a_{l+1}}} - s^{[a_{l+1}, b_{l+1}]}\sum_{i=1}^m D(T_i)(T_i - T_{i-1}) . 
    	\end{cases}
    	\label{eqs:upfront_LP}
    \end{equation}
    \item If the $(l+1)$-th tranche is quoted via spread (with $\uf^{[a_{l+1}, b_{l+1}]} = 0$),  $\overline{s}^{[a_{l+1},b_{l+1}]}_N$ and $\underline{s}^{[a_{l+1},b_{l+1}]}_N$ can be computed by solving the following LFP problems:
        \begin{equation}
    	\begin{cases}
    		\overline{s}^{[a_{l+1},b_{l+1}]}_N = \sup\limits_{\{\hat{p}_{ik}\}\in\mathcal{P}_{N}^l} \frac{\sum_{i=1}^{m}\sum_{k=0}^{N}\sum_{j=0}^{n} \left[D\left(\frac{T_{i}+T_{i-1}}{2}\right)-D\left(\frac{T_{i}+T_{i+1}}{2}\right)1_{\{i<m\}}\right]\beta_{j}^{[a_{l+1}, b_{l+1}]}h_{jk}^{(N)} \hat{p}_{ik}}{\sum_{i=1}^{m}D(T_{i})(T_{i}-T_{i-1})\left((b_{l+1}-a_{l+1})-\sum_{j=0}^{n}\sum_{k=0}^{N}\beta_{j}^{[a_{l+1}, b_{l+1}]}h_{jk}^{(N)}\hat{p}_{ik}\right)}, \\
    		\underline{s}^{[a_{l+1},b_{l+1}]}_N = \inf\limits_{\{\hat{p}_{ik}\}\in\mathcal{P}_{N}^l} \frac{\sum_{i=1}^{m}\sum_{k=0}^{N}\sum_{j=0}^{n} \left[D\left(\frac{T_{i}+T_{i-1}}{2}\right)-D\left(\frac{T_{i}+T_{i+1}}{2}\right)1_{\{i<m\}}\right]\beta_{j}^{[a_{l+1}, b_{l+1}]}h_{jk}^{(N)} \hat{p}_{ik}}{\sum_{i=1}^{m}D(T_{i})(T_{i}-T_{i-1})\left((b_{l+1}-a_{l+1})-\sum_{j=0}^{n}\sum_{k=0}^{N}\beta_{j}^{[a_{l+1}, b_{l+1}]}h_{jk}^{(N)}\hat{p}_{ik}\right)}.
    	\end{cases}
    	\label{eqs:upfront_LFP}
    \end{equation}
\end{itemize}


Proposition \ref{bound_convergence} below establishes that these $N$-dependent price ranges converge to the strong compatibility price range as $N$ tends to infinity.
\begin{Proposition}\label{bound_convergence}
    Suppose that $(s^{[a_1, b_1]}, \uf^{[a_1, b_1]}, \ldots, s^{[a_l, b_l]}, \uf^{[a_l, b_l]}) \in (\mathcal{A}^l)^\circ$ for $1\leq l\leq M-1$.
    Consider the $(l+1)$-th tranche. If it is quoted via up-front payment, then it holds that
        \[
            \lim_{N \to +\infty}\overline{\uf}^{[a_{l+1}, b_{l+1}]}_N = \overline{\uf}^{[a_{l+1}, b_{l+1}]} \quad
            \text{and}\quad \lim_{N \to +\infty}\underline{\uf}^{[a_{l+1}, b_{l+1}]}_N = \underline{\uf}^{[a_{l+1}, b_{l+1}]}.
        \]
    If it is quoted via spread, then it holds that
    \[
        \lim_{N \to +\infty}\overline{s}^{[a_{l+1}, b_{l+1}]}_N = \overline{s}^{[a_{l+1}, b_{l+1}]} \quad \text{and}\quad \lim_{N \to +\infty}\underline{s}^{[a_{l+1}, b_{l+1}]}_N = \underline{s}^{[a_{l+1}, b_{l+1}]}.
    \]
\end{Proposition}
{\it Proof.} See Section \ref{bound_convergence_proof} in the e-companion. \hfill $\Box$

\begin{Remark}[Transforming LFP problems into LP problems]\label{CCtransform}
\upshape{
The computation of $\overline{s}^{[a_{l+1},b_{l+1}]}_N$ and $\underline{s}^{[a_{l+1},b_{l+1}]}_N$ requires solving LFP problems. 
These problems can be efficiently solved by transforming them into equivalent LP problems via the Charnes-Cooper transformation (\citealp{cooper1962programming}).
For an outline of the transformation procedure, see Section \ref{sec_app:LFP2LP} in the e-companion. 
}
\end{Remark}

\subsubsection{Our Iterative Algorithm}\label{sec_subsub:algo_strong_comp}
Based on the key ideas articulated in Section \ref{sec_subsub:key_strong_comp_algo}, we now present the algorithm for verifying strong compatibility and constructing the corresponding copula if strong compatibility is satisfied. 
Proposition~\ref{bound_convergence} ensures that this algorithm terminates in a finite number of steps.

\begin{breakablealgorithm}{{\bf Verifying Strong Compatibility and Constructing a Concrete Copula}}\label{strong_compatibility_algorithm}
	\item[\textbf{Step 1:}] \textbf{Initialization.} 
	Choose a strictly increasing sequence of integers $N_1 < N_2 < \cdots$ and a tolerance $\epsilon > 0$. 
	
	\item[\textbf{Step 2:}] \textbf{Iterative Verification Across Tranches.} 
    Execute an outer loop from $l = 0$ to $M-1$. For each $l$, perform the following procedure to determine whether the prices of the first $(l+1)$ tranches are strongly compatible.
	\begin{enumerate}[leftmargin=5em]
		\item[\textbf{Case 1:}] \textbf{The $(l+1)$-th tranche is quoted via spread.} Perform an inner loop for $y = 1, 2, \ldots$. For each $y$, set $N=N_y$ and execute the following checks.
        \begin{enumerate}[label=(\roman*)]
            \item Compute the $N$-dependent price range $[\underline{s}^{[a_{l+1},b_{l+1}]}_{N_y}, \overline{s}^{[a_{l+1},b_{l+1}]}_{N_y}]$ by solving the corresponding LFP problems \eqref{eqs:upfront_LFP}, which can be transformed into LP problems as mentioned in Remark \ref{CCtransform}.
            
            \item If the market price $s^{[a_{l+1},b_{l+1}]}$ is within the computed $N$-dependent price range $[\underline{s}^{[a_{l+1},b_{l+1}]}_{N_y}, \overline{s}^{[a_{l+1},b_{l+1}]}_{N_y}]$, then the market prices of the first $(l+1)$ tranches are strongly compatible. Exit the inner loop (over $y$) and move to the next tranche.
            
            \item If $y > 1$, check if the changes in bounds are within the pre-specified tolerance $\epsilon$, i.e., whether $|\underline{s}^{[a_{l+1},b_{l+1}]}_{N_y} - \underline{s}^{[a_{l+1},b_{l+1}]}_{N_{y-1}}| < \epsilon$ and $|\overline{s}^{[a_{l+1},b_{l+1}]}_{N_y} - \overline{s}^{[a_{l+1},b_{l+1}]}_{N_{y-1}}| < \epsilon$. 
            If these inequalities hold but the market price $s^{[a_{l+1},b_{l+1}]}$ remains outside the computed $N$-dependent price range, we claim that the full set of market prices is not strongly compatible, and the whole algorithm terminates.
        \end{enumerate}
        \item[\textbf{Case 2:}] \textbf{The $(l+1)$-th tranche is quoted via up-front payment.} The procedure is analogous to Case 1 except that all spread notations are replaced by their up-front payment counterparts and the LFP problems \eqref{eqs:upfront_LFP} in step (i) are replaced by the LP problems \eqref{eqs:upfront_LP}. 
	\end{enumerate} 
	
	\item[\textbf{Step 3:}] \textbf{Final Construction.} 
	If the loop in Step 2 reports that the market prices of all $M$ tranches are strongly compatible, then solve the linear feasibility problem \eqref{linearsystem2} using the final value of $N_y$ determined in the last iteration of Step 2.
	Use the resulting feasible solution to construct the generator $\phi(u)$ as defined in \eqref{implied_phi}, immediately leading to a gamma-distorted copula $C^X$ that achieves a perfect fit to the market prices across all tranches.
\end{breakablealgorithm}

\begin{Remark}[Comparison with the Implied Copula Method of \cite{hull2006valuing}]
	\upshape{\ \ \ It is instructive to compare our strong compatibility framework with the influential implied copula method of \cite{hull2006valuing}. 	Our framework follows in their pioneering footsteps and extends their analysis along three key dimensions. First, our strong compatibility framework is far more general, including all conditionally i.i.d. copula models, with the model considered in \cite{hull2006valuing} as a special case. Second,  
		our strong compatibility framework accounts for the general multi-period setting inherent
		in standard CDO contracts. Third, we establish the formal theoretical condition for the existence of a conditional i.i.d. perfect-fit model and further, if it exists, develop an efficient algorithm for constructing a concrete perfect-fit copula. For more details, please see the comprehensive summary in Section \ref{secsub:compare_hullwhite}.}
\end{Remark}

\subsection{Strong Bid-Ask Compatibility}
Similar to weak compatibility, the strong compatibility framework can also be extended to accommodate market bid-ask spreads. 
Parallel to Definition \ref{def:bid-ask-weak}, we define ``strong bid-ask compatibility" as follows.
\begin{Definition}
	[Strong Bid-Ask Compatibility]\label{def_bid_ask_compatibility}
	{\upshape Market bid-ask prices $(\mathbf{s}_{bid}, \mathbf{s}_{ask})$ of all $M$ CDO tranches are said to be \emph{strongly bid-ask compatible} if there exists a copula $C \in \mathcal{C}_1$ such that the expected NPVs from a long protection position satisfy \eqref{eq:between_equiv}.
	}
\end{Definition}

The verification of strong bid-ask compatibility through a linear feasibility problem is similar to Theorem \ref{strongly_compatible}, as detailed in the following proposition.  
In addition, similar to Algorithm \ref{strong_compatibility_algorithm}, we can also develop an algorithm to verify strong bid-ask compatibility and, when satisfied, constructing a corresponding concrete copula. The details are omitted due to the similarity.
\begin{Proposition}[A Sufficient and Necessary Condition for Strong Bid-Ask Compatibility]\label{strong_bid-ask-corollary}
	Let $\mathcal{A}_{\textrm{bid-ask}} \subseteq \mathbb{R}^{4M}$ denote the set of market bid-ask prices that satisfy strong bid-ask compatibility and $\mathcal{A}_{\textrm{bid-ask}}^\circ$ the interior of $\mathcal{A}_{\textrm{bid-ask}}$. 
	Then for each vector $(\mathbf{s}_{bid}, \mathbf{s}_{ask})\! \in\! \mathcal{A}_{\textrm{bid-ask}}^\circ$, there exists an $N\in\mathbb N^+$ such that the following system of linear constraints for $\{\hat p_{ik}\}_{1 \leq i \leq m, 0 \leq k \leq N}$ has a feasible solution.
	\begin{equation}
		\begin{cases}
			\sum_{i=1}^{m}\sum_{k=0}^{N}g_{ikl}^{(N)}(\mathbf{s}_{bid})\hat{p}_{ik}-\gamma^{[a_{l},b_{l}]}(\mathbf{s}_{bid}) \geq 0 \quad &\quad \text{\upshape{for} $1 \leq l \leq M$}, \\ 
			\sum_{i=1}^{m}\sum_{k=0}^{N}g_{ikl}^{(N)}(\mathbf{s}_{ask})\hat{p}_{ik}-\gamma^{[a_{l},b_l]}(\mathbf{s}_{ask}) \leq 0 \quad &\quad \text{\upshape{for} $1 \leq l \leq M$}, \\
			\sum_{k=0}^N \hat p_{ik} = 1 \, &\quad \text{\upshape{for} $1 \leq i\leq m$},\\
			\sum_{k=0}^N k\hat p_{ik} = N F(T_i) \, &\quad \text{\upshape{for} $1 \leq i \leq m$}, \\
			\sum_{j \geq k} \hat p_{ij} \leq \sum_{j \geq k} \hat p_{i + 1, j} \, &\quad \text{\upshape{for} $1 \leq i \leq m - 1$ \upshape{and} $1 \leq k \leq N$}, \\
			\hat p_{ik} \geq 0 \, &\quad \text{\upshape{for} $1 \leq i \leq m$ \upshape{and} $0\leq k\leq N$}.\\
		\end{cases}
		\label{linearsystem_bid_ask}
	\end{equation}
	
	Conversely, if there exists an $N \in\mathbb N^+$ such that the system of linear constraints \eqref{linearsystem_bid_ask} has a feasible solution, then $(\mathbf{s}_{bid}, \mathbf{s}_{ask})$ satisfies strong bid-ask compatibility.
\end{Proposition}
\noindent{\it Proof.} See Section \ref{bound_convergence_proof} in the e-companion. \hfill$\Box$

\section{Applications of Weak Compatibility and Strong Compatibility}\label{sec:appl}


\subsection{Applications of Weak Compatibility}

Weak compatibility not only ensures a consistent, arbitrage-free pricing framework for the observed CDO tranches but also has important practical applications. 
Specifically, satisfying weak compatibility enables construction of a model-independent strategy for hedging CDO tranches 
as well as the derivation of model-independent bounds for nonstandard credit derivatives tied to the same underlying portfolio, such as CDOs with nonstandard attachment and detachment points.

\subsubsection{A Model-Independent Hedging Strategy}\label{sec:hedging}

For managers of CDO portfolios, one of the most critical risk management activities is hedging against shifts in the underlying CDS index spread. 
This practice, known as index spread delta hedging, is a cornerstone of CDO risk management; see, e.g., \cite{chen2008sensitivity}, \cite{frey2010dynamic}, and \cite{masol2011comparing}. 
The primary objective is to compute the spread delta $\delta^{[a_l, b_l]}$ (for $1\leq l\leq M$), which represents the notional amount of the CDS index required to hedge a position in a CDO tranche. It is defined as the ratio of the change in the CDO tranche's value (i.e., $\Delta v^{[a_l, b_l]}$) to the change in the CDS index's value (denoted by $\Delta v^{CDS}$) for a small shift in the CDS index spread, i.e., $\delta^{[a_l, b_l]} := \frac{\Delta v^{[a_l, b_l]}}{\Delta v^{CDS}}$.

The central challenge in computing the spread delta is to determine how the underlying default distribution changes in response to a shift in the CDS index spread. Our approach frames this as an information updating problem. We begin with a market-implied DPM $\hat{Q}$ that is supposed to be consistent with the initial market prices. When the CDS index spread shifts, this prior distribution must be updated to a new posterior distribution $\tilde{Q}$. Intuitively, a small shift in the CDS index spread should not result in a substantial deviation of the new posterior distribution from the prior one. 
Accordingly, we quantify the distribution deviation using the widely adopted relative entropy (also known as Kullback-Leibler divergence; \citealp{kullback1951information}) and select the posterior distribution that minimizes its relative entropy with the prior. 

Algorithm \ref{alg:hedging} below provides specific procedures for calculating the spread delta, which in turn leads to a practically implementable, model-independent hedging strategy immediately. The empirical studies conducted in Section \ref{sec:hedging_performance} will illustrate its effective hedging performance. 

\begin{breakablealgorithm}{{\bf Calculating the Model-Independent Index Spread Delta}}\label{alg:hedging}
	
	\item[\textbf{Step 1:}] Apply a small shift (e.g., $1$bp) to the CDS index spread and derive the new market-implied marginal default distribution function $\tilde{F}(\cdot)$ using the standard market practice (see Section  \ref{sec_hazard_rate} in the e-companion for the details of this standard market practice; one may also see, e.g., \citealp{hull2016options}).
	
	\item[\textbf{Step 2:}] Solve the following convex optimization problem (it essentially minimizes the posterior DPM's relative entropy with the prior DPM) to find the posterior DPM $\tilde{Q} = \{\tilde q_{ij}\}_{1\leq i \leq m, 0 \leq j \leq n}$. Here the term $\epsilon$ is a small positive constant (e.g., $10^{-20}$) introduced to prevent division by zero. \\
	\textbf{Minimize:}
	\[
	\sum_{i = 1}^m \sum_{j=0}^n \tilde q_{ij} \log\left(\frac{\tilde q_{ij}}{\hat{q}_{ij} + \epsilon} \right)
	\]
	\textbf{Subject to:}
	\[
	\begin{cases}
		\sum_{j=0}^{n}\tilde{q}_{ij}=1, & \quad \text{for $1\leq i\leq m$}, \\
		\sum_{j=0}^{n}j\tilde{q}_{ij}=n\tilde{F}(T_{i}), & \quad \text{for $1\leq i\leq m$}, \\
		\sum_{k\ge j}\tilde{q}_{ik}\leq \sum_{k\ge j}\tilde{q}_{i+1,k} & \quad \text{for $1\leq i\leq m-1 $ and $0\leq j\leq n$}, \\
		\tilde{q}_{ij}\geq 0 & \quad \text{for $1\leq i\leq m$ and $0\leq j\leq n$.}
	\end{cases}
	\]

	\item[\textbf{Step 3:}] Calculate the change in the CDO tranche's value using the prior and posterior DPMs: $\Delta v^{[a_l, b_l]} = \LA^{[a_{l},b_{l}]}(\mathbf{s})^{\prime}(\tilde{Q} - \hat{Q})\BA^{[a_{l},b_{l}]}$.
	Then calculate the corresponding change in the CDS index's value $\Delta v^{CDS}$ using the standard method  (see Section \ref{sec_hazard_rate}  in the e-companion for the details of this standard method; one may also refer to, e.g., \citealp{cont2011dynamic}). 
	
	\item[\textbf{Step 4:}] Compute the spread delta as the ratio of the two value changes:
	$
	\delta^{[a_l, b_l]} = \frac{\Delta v^{[a_l, b_l]}}{\Delta v^{CDS}}.
	$
\end{breakablealgorithm}

\subsubsection{Model-Independent and Arbitrage-Free Pricing Bounds for CDOs with Nonstandard Attachment and Detachment Points}
A key application of the weak compatibility framework is to determine the arbitrage-free price range for a nonstandard CDO tranche $[\tilde{a},\tilde{b}]$ with the nonstandard attachment point $\tilde{a}$ and detachment point $\tilde{b}$. Furthermore, this pricing approach is model-independent because it relies only on the observed market prices of standard tranches, without depending on any specific parametric copula model.

The core idea is that the theoretical price of the aforementioned nonstandard tranche is determined by the underlying DPM. Since the market prices of standard tranches only constrain the DPM to a feasible set $\mathcal{Q}$ rather than determining it uniquely, the price of the nonstandard tranche is not a single value but a range. This range is bounded by the maximum and minimum possible values of the tranche's theoretical price, optimized over all valid DPMs within the feasible set $\mathcal{Q}$. This range is critically important for market participants: a market quote for the nonstandard tranche that falls outside this range implies an arbitrage opportunity. 

The following Algorithm \ref{algo:weak_nonstandard_price_bound} outlines the procedures for computing these price bounds, and some numerical results will be provided in Section \ref{sec_sub:numerical_weak_nonstandard_price_bound}.

\begin{breakablealgorithm}{{\bf Calculating Model-Independent and Arbitrage-Free Pricing Bounds}}\label{algo:weak_nonstandard_price_bound}
	\item[\textbf{Case 1:}] \textbf{The tranche is quoted via spread.} 
	In such a case, the maximum market-implied spread $\overline{s}^{[\tilde{a},\tilde{b}]}$ is given by the optimal value of the following LFP problem:
	\[ 
        \overline{s}^{[\tilde{a},\tilde{b}]} = \sup_{\hat{Q}\in\mathcal{Q}} \frac{\sum_{i=1}^{m}\sum_{j=0}^{n}[D(\frac{T_{i}+T_{i-1}}{2})-D(\frac{T_{i}+T_{i+1}}{2})1_{\{i<m\}}]\beta_{j}^{[\tilde{a},\tilde{b}]}\hat{q}_{ij}}{\sum_{i=1}^{m}D(T_{i})(T_{i}-T_{i-1})(\tilde{b}-\tilde{a}-\sum_{j=0}^{n}\beta_{j}^{[\tilde{a},\tilde{b}]}\hat{q}_{ij})},
    \]
    where $\beta_{j}^{[\tilde{a},\tilde{b}]} = \left(\frac{j(1-R)}{n} - \tilde{a}\right)^+ - \left(\frac{j(1-R)}{n} - \tilde{b}\right)^+$ for $0 \leq j \leq n.$
	The minimum market-implied spread $\underline{s}^{[\tilde{a},\tilde{b}]}$ can be identified by replacing $\sup$ with $\inf$.
	We can transform these LFP problems into the LP problems via the Charnes-Cooper transformation (see Remark \ref{CCtransform} and Section \ref{sec_app:LFP2LP} in the e-companion for the detail of this transformation), and then obtain the model-independent and arbitrage-free pricing bounds by solving the resulting LP problems. 
    \item[\textbf{Case 2:}] \textbf{The tranche is quoted via up-front payment.}
    In such a case, the maximum market-implied up-front payment $\overline{\uf}^{[\tilde{a},\tilde{b}]}$ is given by the optimal value of the following LP problem:
    \[
        \overline{\uf}^{[\tilde{a},\tilde{b}]} = \sup_{\hat{Q}\in\mathcal{Q}}\frac{\sum_{i=1}^m \sum_{j=0}^n \lambda_i^{[\tilde{a}, \tilde{b}]}\beta_j^{[\tilde{a}, \tilde{b}]}\hat{q}_{ij}}{\tilde{b} - \tilde{a}} - s^{[\tilde{a}, \tilde{b}]}\sum_{i=1}^n D(T_i)(T_i - T_{i - 1}),
    \]
    where
    \begin{equation}\label{lambda_nonstandard}
        \lambda_i^{[\tilde{a}, \tilde{b}]} = s^{[\tilde{a}, \tilde{b}]}D(T_i)(T_i - T_{i-1}) + D\left(\frac{T_i + T_{i-1}}{2}\right) - D\left(\frac{T_{i+1} + T_i}{2}\right)\I_{\{i < m\}} \, \text{for}\, i = 1, \ldots, m.
    \end{equation}
    The minimum market-implied up-front payment $\underline{\uf}^{[\tilde{a},\tilde{b}]}$ can be found by replacing $\sup$ with $\inf$.
    Solving these LP problems then yields the model-independent and arbitrage-free pricing bounds.
\end{breakablealgorithm}


\begin{Remark}
	{\bf (Model-Independent and Arbitrage-Free Pricing of Other Portfolio Credit Derivatives)}
	\upshape {The model independent and arbitrage-free pricing approach presented above using the weak compatibility framework is not limited to pricing nonstandard CDO tranches with nonstandard attachment and detachment points. Rather, it applies to pricing \textit{any} claim contingent solely on the default count process $N_{t}$ at times $T_{i}$. Indeed,
	the elements $\hat{q}_{ij}$ of any feasible DPM $\hat{Q}$ represent the risk-neutral probabilities $\mathbb{P}(N_{T_{i}}=j)$. Therefore, this allows for model-independent and arbitrage-free pricing of various other portfolio credit derivatives, such as the $k$-th to default CDSs (basket default swaps) and options on the number of defaults. 
	}
\end{Remark}

\subsection{Applications of Strong Compatibility}

The strong compatibility framework offers effective tools for a range of applications. 
The conditionally i.i.d. structure, which is central to strong compatibility, is the key enabler for these applications as it enables the separation of systematic risk (the common factors) from idiosyncratic risk (the individual defaults). 
We would like to elaborate on two examples: estimating loss or NPV distributions for risk management, and pricing CDOs with a nonstandard number of underlying names.

\subsubsection{Estimating NPV Distributions for Risk Management}
For the purpose of risk management, it is essential to analyze the loss or NPV distribution of a portfolio composed of positions in multiple CDO tranches. The calibrated gamma-distorted copula is particularly well-suited for this task. This is because it models the dependence among defaults through a common set of stochastic factors and conditional on these common factors, the defaults then possess an i.i.d. structure, thereby facilitating the evaluation of the NPV distributions of all CDO tranches through simulation. 

The following Algorithm \ref{simulate_algorithm} provides detailed procedures for simulating the NPV distributions of all CDO tranches. 
This algorithm can be directly applied to estimate the NPV distribution of a trading portfolio with any combination of long or short positions in various CDO tranches and the underlying CDS index. Some numerical results will be provided in Section \ref{sec_subsub: numerical_NPV}.

\begin{breakablealgorithm}{{\bf Simulating NPV Distributions of the CDO Tranches }}\label{simulate_algorithm}
    \item[\textbf{Step 1:}] \textbf{Simulate Random Drivers.} 
    Sample  $n+1$ independent uniform random variables $U, U_{1}, \ldots, U_{n} \sim \mathcal{U}(0,1)$. 
	Moreover, sample two independent paths of gamma processes $\{\xi_{t}\}$ and $\{\eta_{t}\}$.

    \item[\textbf{Step 2:}] \textbf{Construct a Stochastic Distortion Path.}
    Use a feasible solution $\{\hat{p}_{ik}\}$ of \eqref{linearsystem2} and the sampled random variable $U$ to construct a single path for the generator $\{{\phi}(u): 0 \leq u \leq 1\}$ as defined in \eqref{implied_phi}. Then, construct the corresponding path for the stochastic distortion process:
    $ X(F(t)) = \frac{\xi_{{\phi}(F(t))}}{\xi_{{\phi}(F(t))} + \eta_{N-{\phi}(F(t))}}. $

    \item[\textbf{Step 3:}] \textbf{Simulate Defaults and Tranche Values.}
    For each time point $t \in \{T_1, ..., T_m\}$, calculate the aggregate number of defaults:
    $ N_{t} = \sum_{j=1}^{n} 1_{\{U_{j} \le X(F(t))\}}. $
    Then, calculate the realized NPV for each tranche $[a_l, b_l]$ based on this simulated path of default counts:
    $$ V^{[a_{l},b_{l}]} = \sum_{i=1}^{m}\lambda_{i}^{[a_{l},b_{l}]}(\mathbf{s}) \left[ \left(\frac{(1-R)N_{T_{i}}}{n}-a_{l}\right)^{+} - \left(\frac{(1-R)N_{T_{i}}}{n}-b_{l}\right)^{+} \right] - \gamma^{[a_{l},b_{l}]}(\mathbf{s}).$$
%
\end{breakablealgorithm}

\subsubsection{Pricing CDOs with a Nonstandard Number of Underlying Names}

The strong compatibility framework provides a consistent method for valuing CDOs on a portfolio with a nonstandard number of underlying names $\tilde{n}$. 
Let $N$ be a sufficiently large integer such that \eqref{linearsystem2} has a feasible solution, and denote its feasible set by $\mathcal{P}_{N}$.
$\mathcal{P}_{N}$ represents the calibrated systematic risk profile, which is 
not inherently tied to the portfolio's original size $n$. 
Thus, it can be consistently applied to a CDO tranche $[\tilde{a}, \tilde{b}]$ with a nonstandard number of underlying names $\tilde{n}$, provided that the new portfolio shares the same systematic risk profile. 

Algorithm \ref{algo:nonstandard CDO 2} below provides detailed procedures regarding how to calculate price bounds for CDOs with a nonstandard number of underlying names. This algorithm can be readily extended to price other credit derivatives such as basket default swaps written on a nonstandard portfolio composed of CDO tranches with a nonstandard number of underlying names.

\begin{breakablealgorithm}{\bf (Calculating Price Bounds for CDOs with a Nonstandard Number of Underlying Names)}\label{algo:nonstandard CDO 2}
    \item[\textbf{Case 1:}] \textbf{The tranche is quoted via spread.} 
    In such a case, the maximum market-implied spread $\overline{s}^{[\tilde{a},\tilde{b}]}$ is given by the optimal value of the following LFP problem:
    \[ 
		\overline{s}^{[\tilde{a},\tilde{b}]} = \sup_{\{\hat{p}_{ik}\}\in\mathcal{P}_{N}} \frac{\sum_{i=1}^{m}\sum_{k=0}^{N}\sum_{j=0}^{\tilde{n}} [D(\frac{T_{i}+T_{i-1}}{2})-D(\frac{T_{i}+T_{i+1}}{2})1_{\{i<m\}}]\tilde{\beta}_{j}^{[\tilde{a},\tilde{b}]}\tilde{h}_{jk}^{(N)} \hat{p}_{ik}}{\sum_{i=1}^{m}D(T_{i})(T_{i}-T_{i-1})(\tilde{b}-\tilde{a}-\sum_{j=0}^{\tilde{n}}\sum_{k=0}^{N}\tilde{\beta}_{j}^{[\tilde{a},\tilde{b}]}\tilde{h}_{jk}^{(N)}\hat{p}_{ik})},
	\]
    where $		\tilde{\beta}_{j}^{[\tilde{a},\tilde{b}]} = \left(\frac{j(1-R)}{\tilde{n}}-\tilde{a}\right)^{+} - \left(\frac{j(1-R)}{\tilde{n}}-\tilde{b}\right)^{+}$ for $0 \leq j \leq \tilde n$, 
    			$\tilde h_{j0}^{(N)} := \I_{\{j = 0\}}$ for $0\leq j \leq \tilde n$,
    $\tilde h_{jk}^{(N)} :=  \binom{\tilde n}{j} \frac{\B(k + j , N + \tilde n - k - j)}{\B(k, N - k)}$ for $0 \leq j \leq \tilde n$ and $1 \leq k \leq N - 1$, and 
    $\tilde h_{jN}^{(N)} := \I_{\{j = \tilde n\}}$ for $0\leq j \leq \tilde n$.
    The minimum market-implied spread $\underline{s}^{[\tilde{a},\tilde{b}]}$ can be identified by replacing $\sup$ with $\inf$.
    We can transform these LFP problems into the LP problems via the Charnes-Cooper transformation (see Remark \ref{CCtransform} and Section \ref{sec_app:LFP2LP} in the e-companion for the detail of this transformation), and then obtain the price bounds by solving the resulting LP problems.
    \item[\textbf{Case 2:}] \textbf{The tranche is quoted via up-front payment.}
    In such a case, the maximum market-implied up-front payment $\overline{\uf}^{[\tilde{a},\tilde{b}]}$ is given by the optimal value of the following LP problem:
    \[ 
		\overline{\uf}^{[\tilde{a},\tilde{b}]} = \sup_{\{\hat{p}_{ik}\}\in\mathcal{P}_{N}} \frac{\sum_{i=1}^m \sum_{j=0}^n \sum_{k=0}^N \lambda_i^{[\tilde{a},\tilde{b}]} \tilde{\beta}_{j}^{[\tilde{a},\tilde{b}]} \tilde h_{jk}^{(N)}\hat{p}_{ik}}{\tilde{b} - \tilde{a}} - s^{[\tilde{a}, \tilde{b}]}\sum_{i=1}^m D(T_i)(T_i - T_{i - 1}),
	\]
	    where $\lambda_i^{[\tilde{a},\tilde{b}]}$ is defined as in \eqref{lambda_nonstandard}.
    The minimum market-implied up-front payment $\underline{\uf}^{[\tilde{a},\tilde{b}]}$ can be identified by replacing $\sup$ with $\inf$. 
    Solving these LP problems yields the price bounds immediately.
\end{breakablealgorithm}

\section{Empirical Studies}\label{sec:empirical}

In this section, we present empirical studies on the practical applications of both the weak and strong compatibility frameworks, illustrating how they can be utilized in risk management and the pricing of nonstandard credit derivatives. 
\vspace{-0.3cm}

\subsection{Data and Verification of Compatibility}
\subsubsection{Data}
We use the market data for the iTraxx Markit Europe S42 5-year tranches on March 28, 2025, which reference 125 names. This CDO comprises four tranches: $[0\%,3\%]$, $[3\%,6\%]$, $[6\%,12\%]$, and $[12\%,100\%]$. 
The first two tranches are quoted by up-front payments with a running spread of $100$ bps, while the last two are quoted by tranche spreads with an up-front payment of zero; see Table \ref{quote_example} for the corresponding market quotes of these tranches.
\vspace{-0.2cm}

\begin{table}[htbp]
	\caption{	Market Quotes for the iTraxx Markit Europe S42 5-Year Tranches on March 28, 2025 }
	\vspace{0.1cm}
		\begin{center}
			\begin{small}
	\begin{tabularx}{\textwidth}{l@{\hspace{36pt}}XXXX}
	\toprule
	Tranche & [0\%, 3\%] & [3\%, 6\%] & [6\%, 12\%] & [12\%, 100\%] \\
	\midrule
	Market Quote & 28.438\% & 4.531\% & 106.32 bps & 27.44 bps \\
	\bottomrule
\end{tabularx}\end{small}
		\end{center}
		\vspace{0.3cm}
\vspace{0.05cm}
	    \begin{minipage}{\textwidth}
		\small \linespread{1.2}\selectfont 
			\vspace{0.15cm}
		\emph{{\footnotesize{Notes.}}} 
		{\footnotesize{The first two tranches are quoted by up-front payments with a running spread of $100$ bps, while the last two are quoted by tranche spreads with an up-front payment of zero. 
	}}
		\end{minipage}
\label{quote_example}	
\end{table}
\vspace{-0.46cm}

The spread for the iTraxx Markit Europe S38 CDS index is $s^{\rm idx} = 58$ bps on March 28, 2025. The marginal distribution $F(\cdot)$ of the default times is calculated from this index spread using the standard market practice (see Section \ref{sec_hazard_rate}  in the e-companion or \citealp{hull2016options} for the details). 
In addition, following the standard assumption in the literature (e.g., \citealp{schloegl2005note}), we assume a constant recovery rate $R = 0.4$. The risk-free interest rate $r$ is derived from the Euro 
OIS curve based on the Euro Short-Term Rate (€STR); for March 28, 2025, this rate is $r = 2.417\%$.
All the LP problems in the empirical studies are solved using the CVXPY package in Python.

\subsubsection{Verification of Weak Compatibility and Strong Compatibility}

We first apply Theorems \ref{ImpliedMatrix} and \ref{strongly_compatible} to verify whether the market prices on March 28, 2025 satisfy weak and strong compatibility. 
The results show that for the market prices on this date, the linear feasibility problems  \eqref{linearsystem} and \eqref{linearsystem2} (with $N = 100$) both have feasible solutions. 
Therefore, we conclude that the market prices on this date satisfy both weak and strong compatibility.

Next, we conduct a further analysis to identify the boundaries of strong compatibility. 
Specifically, we fix the market prices of three of the four standard tranches and then solve for the strong compatibility price range of the remaining tranche.
As discussed in Section \ref{sec_subsub:key_strong_comp_algo}, this strong compatibility price range given by \eqref{eq:uf_ub_lb} 
is not directly computable. 
Instead, we can compute a sequence of $N$-dependent price ranges, which convergences to the strong compatibility price range as $N$ tends to infinity, as shown in Proposition \ref{bound_convergence}.  

Table \ref{tab:strong_compatibility_ranges} presents the $N$-dependent price ranges for each tranche corresponding to different values of the gamma distortion parameter $N$. Two observations emerge. First, the actual market price for each tranche, which is given in Table \ref{quote_example}, falls within the corresponding $N$-dependent price range for any $N$ used in Table \ref{tab:strong_compatibility_ranges}, reaffirming that the market prices on this day are strongly compatible. 
Second, the calculated $N$-dependent price ranges stabilize quickly as $N$ increases, indicating a fast convergence. Then the calculated $N$-dependent price range for a sufficiently large $N$ can be regarded approximately as the strong compatibility price range.  

%

The above analysis on the strong compatibility price ranges demonstrates how to determine whether a price of a specific tranche is strongly compatible with the rest of the market prices.  Indeed, any price within the computed strong compatibility price range is strongly compatible with the fixed market prices of the other three tranches. 

 
\vspace{-0.2cm}

\begin{table}[htbp]
	\caption{ $N$-Dependent Price Ranges for the Four Tranches for Different Values of the Gamma Distortion Parameter $N$ }
		\vspace{0.1cm}
		\begin{center}
		\begin{small}
	\begin{tabular*}{\textwidth}{@{\extracolsep{\fill}} l c c c c @{}}
		\toprule
		Tranche & [0\%, 3\%] & [3\%, 6\%] & [6\%, 12\%] & [12\%, 100\%] \\
		\midrule
		$N=50$ & [28.279\%, 28.936\%] & [4.372\%, 5.030\%] & [104.57, 111.78] bps & [27.33, 27.79] bps \\
		$N=75$ & [28.277\%, 28.939\%] & [4.371\%, 5.033\%] & [104.56, 111.83] bps & [27.33, 27.80] bps \\
		$N=100$ & [28.276\%, 28.941\%] & [4.369\%, 5.034\%] & [104.55, 111.87] bps & [27.32, 27.80] bps \\
		$N=125$ & [28.275\%, 28.941\%] & [4.368\%, 5.034\%] & [104.55, 111.89] bps & [27.32, 27.80] bps \\
		$N=150$ & [28.274\%, 28.942\%] & [4.368\%, 5.035\%] & [104.54, 111.91] bps & [27.32, 27.80] bps \\
		$N=175$ & [28.274\%, 28.942\%] & [4.367\%, 5.036\%] & [104.53, 111.92] bps & [27.32, 27.80] bps \\
		$N=200$ & [28.273\%, 28.942\%] & [4.367\%, 5.036\%] & [104.53, 111.92] bps & [27.32, 27.80] bps \\
		\bottomrule
	\end{tabular*}
	\end{small}
	\end{center}
	\vspace{0.3cm}
	 \vspace{0.05cm}
	    \begin{minipage}{\textwidth}
		\small \linespread{1.2}\selectfont 
			\vspace{0.15cm}
\emph{{\footnotesize{Notes.}}} 
{\footnotesize{The calculated $N$-dependent price ranges stabilize quickly as $N$ increases, indicating a fast convergence. Then the calculated $N$-dependent price range for a sufficiently large $N$ 
			can be regarded approximately as the strong compatibility price range. Any price within each strong compatibility price range is strongly compatible with the fixed market prices of the other three tranches.  }}
			\end{minipage}
	\label{tab:strong_compatibility_ranges}
	
\end{table}
\vspace{-0.7cm}

\subsection{Risk Management for CDO Portfolios}
\subsubsection{Constructing Effective Model-Independent Hedging Strategies for CDO Portfolios}
\label{sec:hedging_performance}

To construct a hedging strategy for CDO portfolios, it suffices to develop a hedging strategy for each CDO tranche. By applying Algorithm \ref{alg:hedging} developed in Section \ref{sec:hedging}, we can compute the index spread delta $\delta^{[a_{l},b_{l}]}$ for each tranche, thereby immediately yielding a model-independent hedging strategy. For comparison, we also compute the index spread deltas for all tranches using a standard Gaussian copula model. This standard Gaussian copula model serves as an appropriate benchmark because the corresponding index spread deltas are often very similar to those obtained from more complex dynamic models. For instance, \cite{cont2011dynamic} demonstrate that an affine jump-diffusion model and a top-down local intensity model produce index spread deltas remarkably close to those of the Gaussian copula model. 

Figure \ref{fig:delta_comparison} presents the index spread deltas—calculated for March 28, 2025 using our model-independent method and the Gaussian copula model—as well as the corresponding index spread deltas per unit width. The results highlight two advantages of our model-independent method, both of which align with financial intuition. 

First, the sum of the index spread deltas across the four tranches derived from our model-independent method is approximately one ($0.1970+0.1357+0.1554+0.5071 = 0.9952$). This is consistent with the theoretical expectation that a portfolio consisting of a long position in all tranches of the CDO, hedged with a short position in the underlying CDS index, should be nearly risk-free. 
In contrast, the sum of the index spread deltas across the four tranches calculated from the Gaussian copula model is $0.5992$, which deviates significantly from this intuitive principle.

Second, our model-independent method correctly captures the seniority-based risk structure. Financial intuition dictates that the index spread delta per unit width should decrease as tranche seniority increases. The right panel of Figure \ref{fig:delta_comparison} confirms that our method meets this expectation. 
In contrast, the Gaussian copula model fails this test. Its index spread delta per unit width is non-monotonic, counter-intuitively implying that the senior mezzanine tranche is more sensitive than the junior mezzanine tranche. This further underscores the capability of our model-independent approach to generate sensitivities consistent with the economic realities of CDO capital structures.

\begin{figure}[ht]
	\caption{Comparison of Index Spread Deltas (Left Panel) and Index Spread Deltas per Unit Width (Right Panel) for March 28, 2025 between Our Model-Independent Method (``This Study") and the Gaussian Copula Model (``Gaussian") 
}\vspace{0.1cm}
	\begin{subfigure}[b]{0.49\textwidth}
		\centering
		\includegraphics[width=\textwidth]{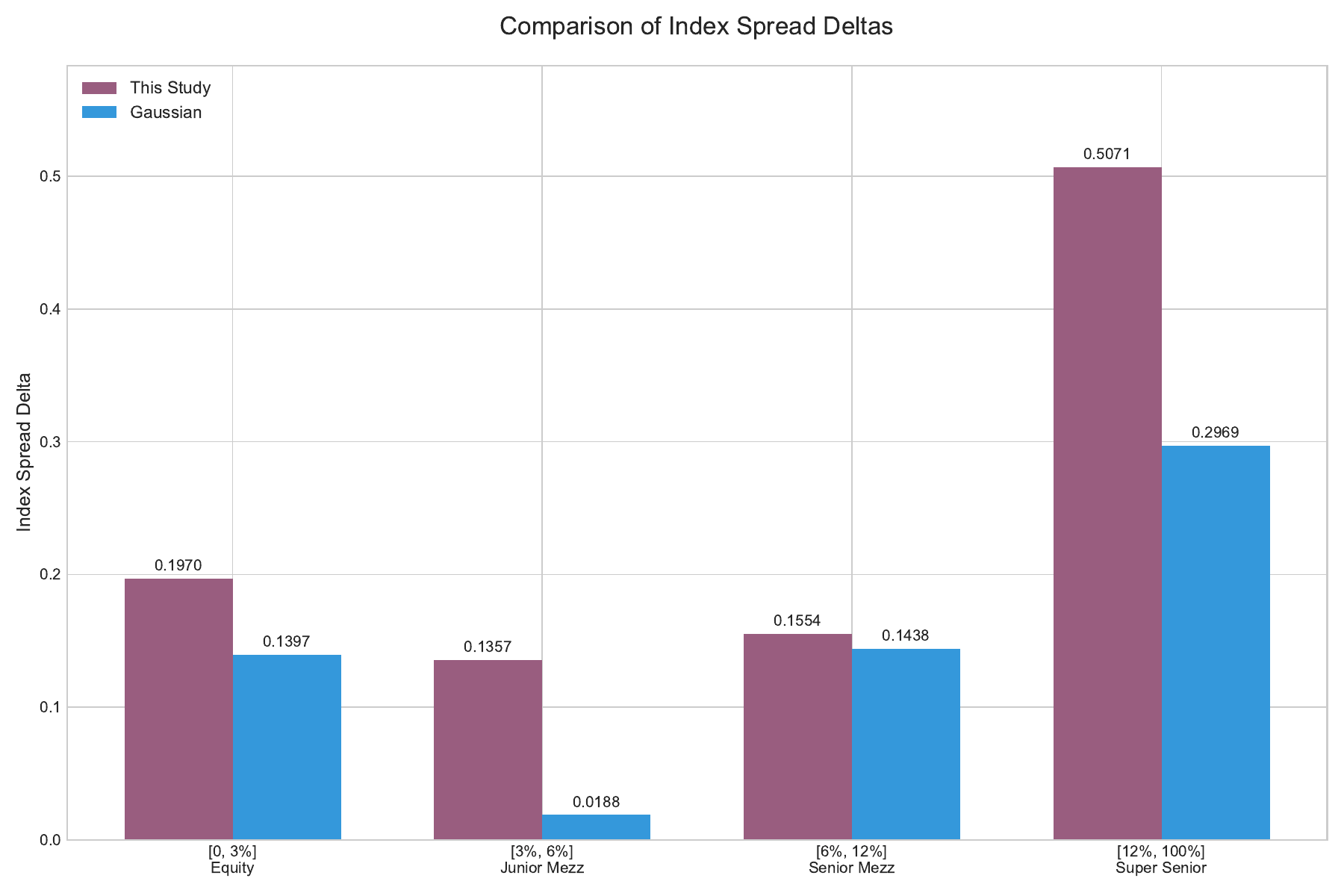}
		\label{fig:delta_abs}
	\end{subfigure}
	\hfill
	\begin{subfigure}[b]{0.49\textwidth}
		\centering
		\includegraphics[width=\textwidth]{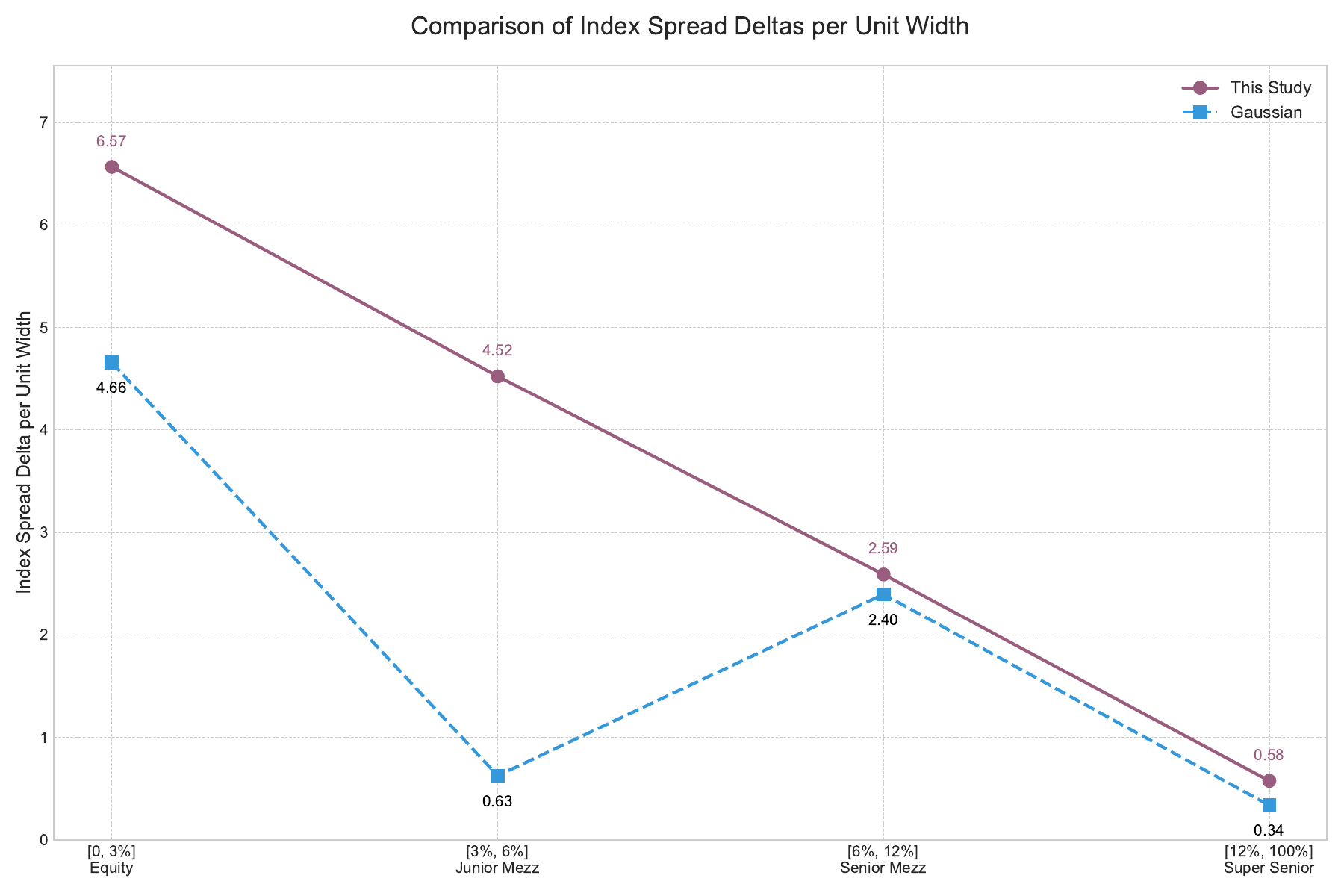}
		\label{fig:delta_per_width}
	\end{subfigure}
	    \begin{minipage}{\textwidth}
		\small \linespread{1.2}\selectfont 
		
\emph{{\footnotesize{Notes.}}} 
{\footnotesize{This figure presents the index spread deltas (left panel)—calculated for March 28, 2025 using our model-independent method and the Gaussian copula model—as well as the corresponding index spread deltas per unit width (right panel). The results highlight two advantages of our model-independent method, both of which align with financial intuition. First, the sum of the index spread deltas across the four tranches derived from our model-independent method is approximately one ($0.9952$), while the sum calculated from the Gaussian copula model is $0.5992$. 
Second, our model-independent method correctly captures the seniority-based risk structure; 
the index spread delta per unit width exhibits a monotonic decrease from the equity tranche to the super-senior tranche. In contrast, the Gaussian copula model fails this test.}} 
	\end{minipage}
	\label{fig:delta_comparison}
\end{figure}

\begin{figure}[ht]
\begin{center}
	\caption{Hedging Performance of Our Model-Independent Hedging Strategy (Daily P\&L of Four Tranches)}\vspace{0.1cm}  
	\includegraphics[width=\textwidth]{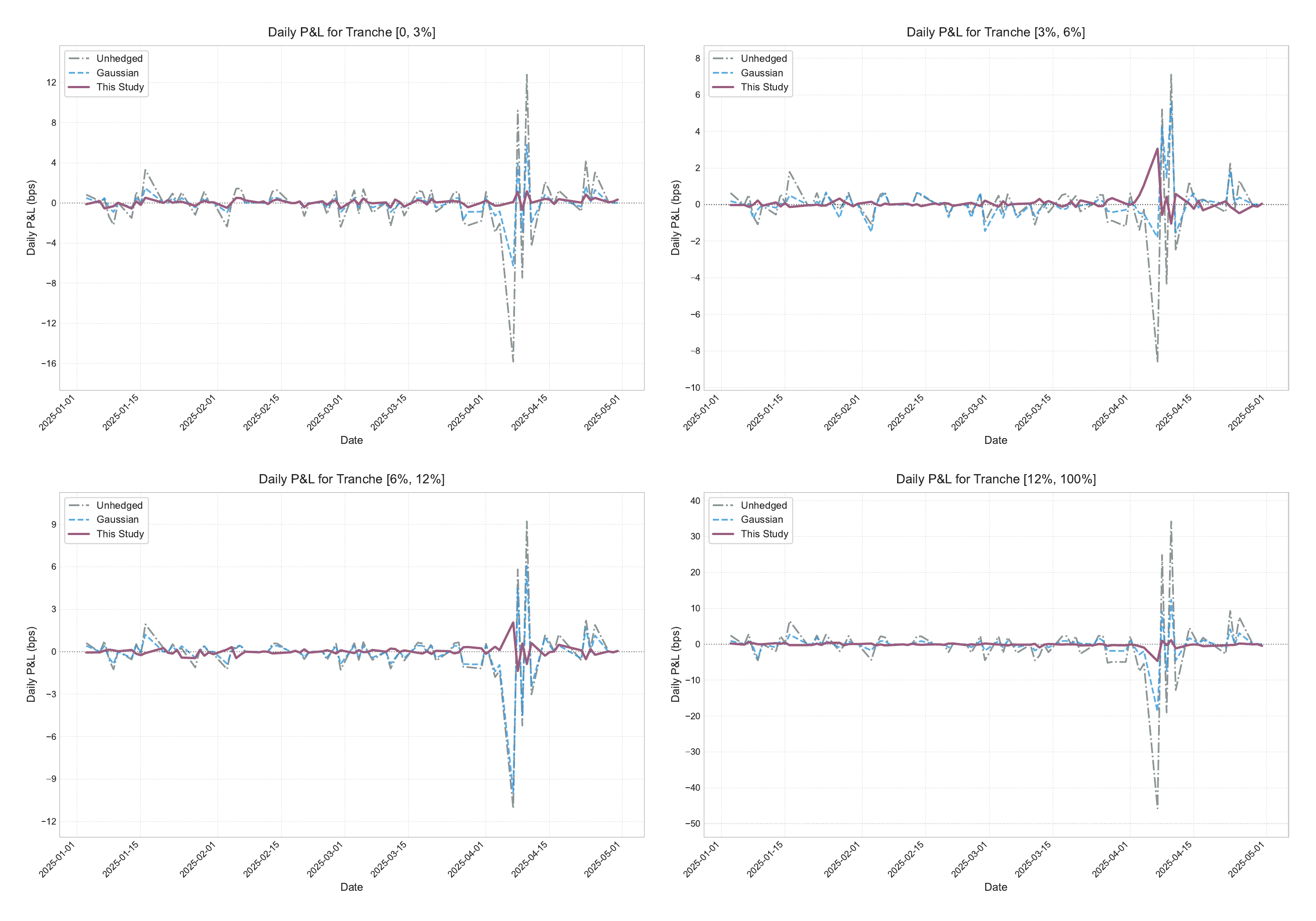}
		\label{fig:pnl_plot}
\end{center}
    \begin{minipage}{\textwidth}
	\small \linespread{1.2}\selectfont 
	
	\emph{{\footnotesize{Notes.}}} 
	{\footnotesize{This figure presents the daily P\&L of (i) the delta-hedged portfolios resulting from our model-independent hedging strategy (the solid curves), (ii) the unhedged portfolios (the dash-dotted curves), and (iii) the delta-hedged portfolios resulting from the Gaussian copula model (the dashed curves). It clearly illustrates the significant reduction in P\&L volatility achieved by our model-independent hedging strategy, when compared with the unhedged counterparts. 
	Furthermore, it can be seen that our model-independent hedging strategy also yields a substantially smaller P\&L volatility than the Gaussian copula model.}}
\end{minipage}
\end{figure}


We further conduct a backtest to empirically validate the effectiveness of our model-independent hedging strategy. Following \cite{cont2011dynamic}, we assess hedging performance by analyzing the daily profit and loss (P\&L) of the resulting delta-hedged portfolios. Specifically, the backtest is conducted over the period from January 1, 2025 to April 30, 2025; for each day in this window, we employ our model-independent hedging strategy to construct a delta-hedged portfolio and calculate its P\&L. Figure~\ref{fig:pnl_plot} presents the daily P\&L of (i) the delta-hedged portfolios resulting from our model-independent hedging strategy, (ii) the unhedged portfolios, and (iii) the  delta-hedged portfolios resulting from the Gaussian copula model. It clearly illustrates the significant reduction in P\&L volatility achieved by our model-independent hedging strategy, when compared with the unhedged counterparts. This is particularly evident for the equity and super-senior tranches, where our hedge remains robust in controlling P\&L fluctuations, even during the heightened market volatility observed in early April. Furthermore, it can be seen that our model-independent hedging strategy also yields a substantially smaller P\&L volatility than the Gaussian copula model. 


In addition, in the spirit of \cite{cont2011dynamic}, we quantitatively assess the P\&L volatility of the unhedged and two aforementioned delta-hedged portfolios by calculating the standard deviations of their respective daily P\&L series. The left panel of Figure~\ref{fig:std_analysis} presents these standard deviations, and we can see that our model-independent hedging strategy delivers superior hedging performance across all four tranches. 
Furthermore, as shown in the right panel of Figure~\ref{fig:std_analysis}, our hedging strategy achieves significant reductions in standard deviation, with reductions ranging from $74.2\%$ to $92.3\%$, relative to the unhedged portfolios. Across all tranches, our hedging strategy significantly outperforms the Gaussian copula model, yielding a substantially lower P\&L volatility.
This consistent reduction in standard deviation underscores the practical utility and robustness of our model-independent hedging strategy for CDO portfolios.

%

\begin{figure}[ht]
	\caption{Hedging Performance of Our Model-Independent Hedging Strategy (Comparison of P\&L Standard  Deviations)}\vspace{0.1cm} 
	\includegraphics[width=\textwidth,height=0.36\textwidth]{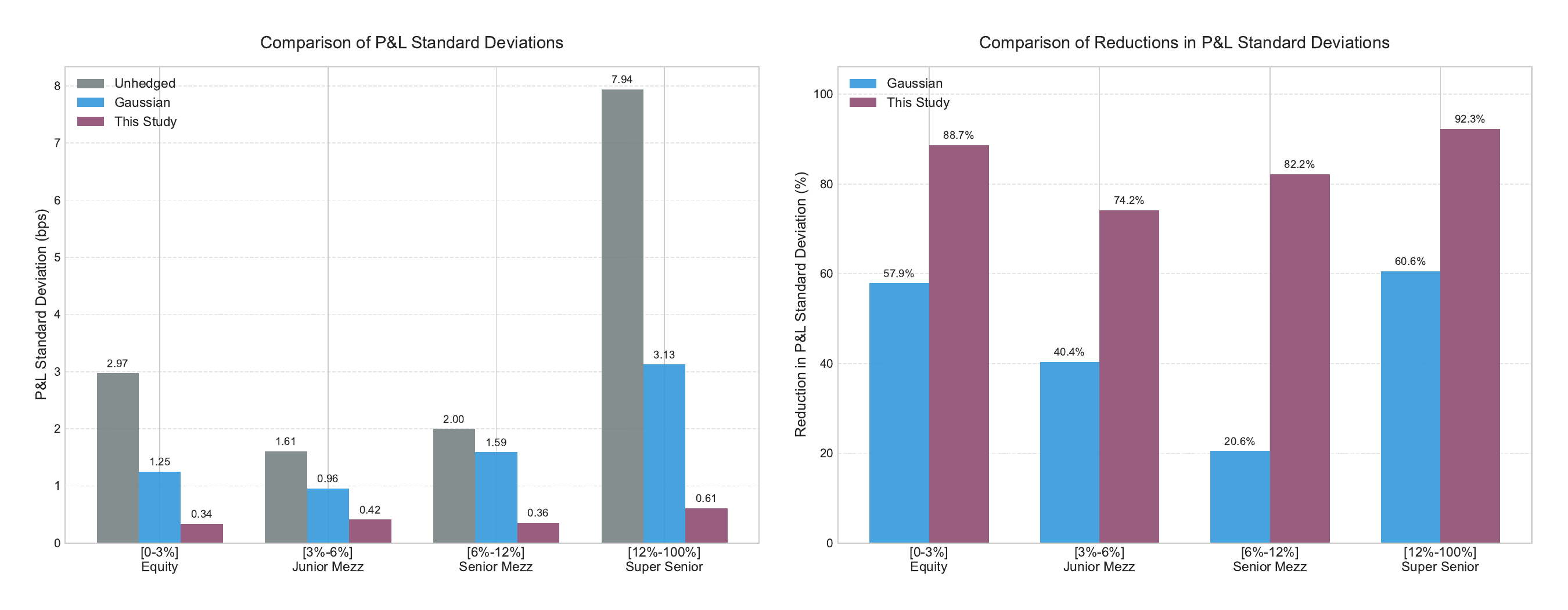}
			
			    \begin{minipage}{\textwidth}
				\small \linespread{1.2}\selectfont 
				
				\emph{{\footnotesize{Notes.}}} 
				{\footnotesize{The left panel presents the standard deviations of the daily P\&L series of (i) the delta-hedged portfolios resulting from our model-independent hedging strategy, (ii) the unhedged portfolios, and (iii) the delta-hedged portfolios resulting from the Gaussian copula model. We can see that our model-independent hedging strategy delivers superior hedging performance across all four tranches. The right panel shows that our hedging strategy achieves significant reductions in standard deviation, with reductions ranging from $74.2\%$ to $92.3\%$, relative to the unhedged portfolios. Furthermore, across all tranches, our hedging strategy significantly outperforms the Gaussian copula model, yielding a substantially lower P\&L volatility.}}
			\end{minipage}
	\label{fig:std_analysis}
\end{figure}

\subsubsection{NPV Distributions for CDO Tranches and Portfolios}\label{sec_subsub: numerical_NPV}

Analyzing the NPV distributions for CDO portfolios composed of multiple CDO tranches is critical for CDO portfolio risk management. However, traditional parametric models, which are often calibrated only to individual tranches, typically fail to support this analysis, as they cannot generate a single, consistent joint default distribution for the entire capital structure. In contrast, our strong compatibility framework enables precisely this portfolio-level analysis: the unified, consistent model underpinning it can perfectly fit all market prices, thereby enabling reliable estimation of NPV distributions for CDO portfolios.

We first conduct an empirical analysis on the NPV distributions of all CDO tranches, the building blocks of those of any CDO portfolio. 
By applying Algorithm \ref{simulate_algorithm}, we can simulate the NPV distributions of all CDO tranches (i.e., the distributions of $V^{[a_l, b_l]}$ for $l=1,\ldots,M$) with the simulated sample size of $1,000,000$.
The left panel of Figure \ref{fig:combined_dist} illustrates the resulting joint NPV distribution of the equity and senior mezzanine tranches. The ``reverse L'' shape shown in the figure clearly reflects the defining waterfall payment structure of CDO products: losses on a senior tranche begin to accrue only after the subordinated tranche has been completely wiped out. Similar waterfall payment structures are also observed in the joint NPV distributions of other tranche pairs. Furthermore,  the sample mean of the NPV $V^{[a_l,b_l]}$ for each tranche is statistically indistinguishable from zero, supporting the validity of our method. 
\begin{figure}[h!]
	\centering
	\caption{Simulated NPV Distributions for CDO Tranches and Portfolios}
	\begin{subfigure}[b]{0.48\textwidth}
		\centering
		\includegraphics[width=\textwidth,height=0.7\textwidth]{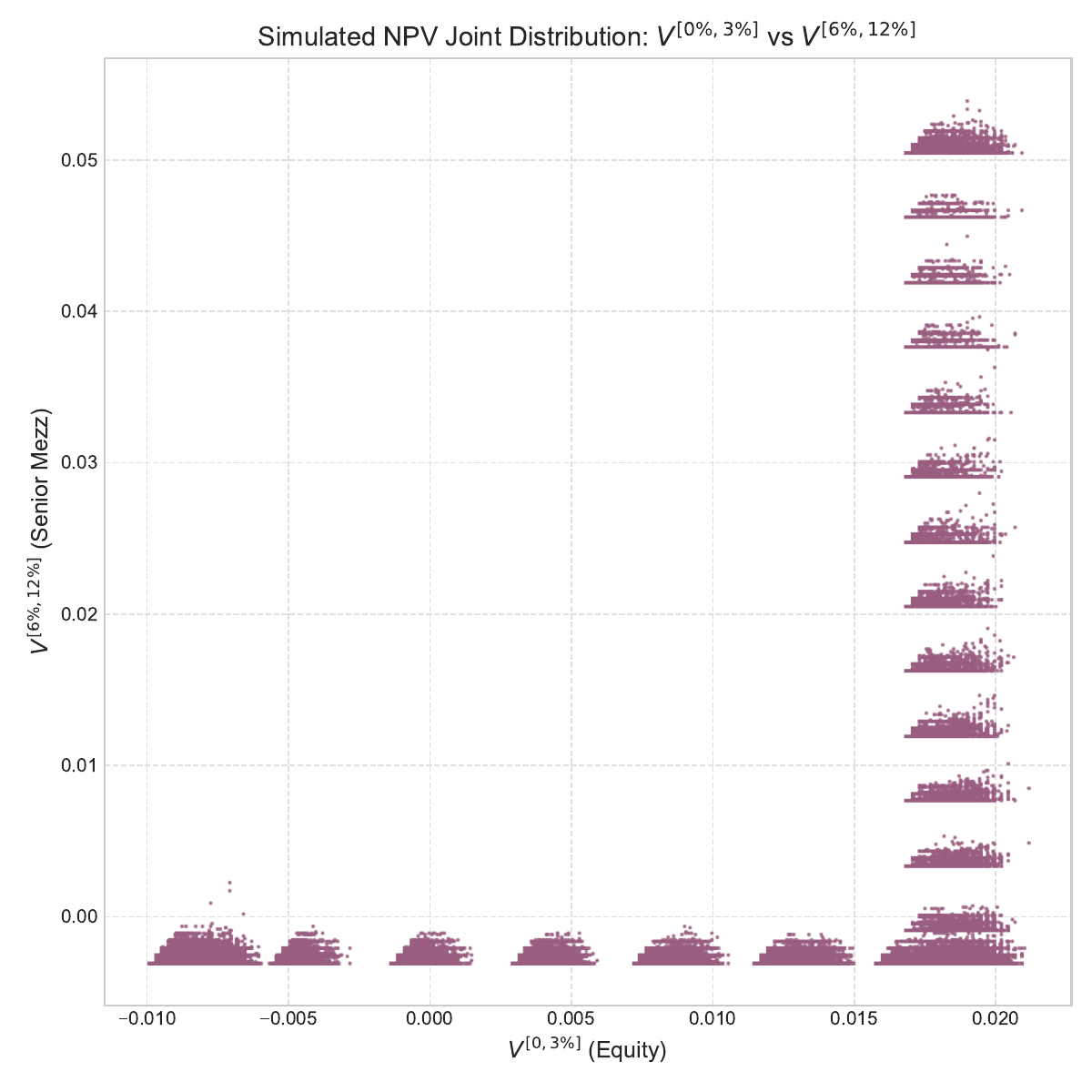} 
	\end{subfigure}
	\hfill
	\begin{subfigure}[b]{0.48\textwidth}
		\centering
		\includegraphics[width=\textwidth,height=0.7\textwidth]{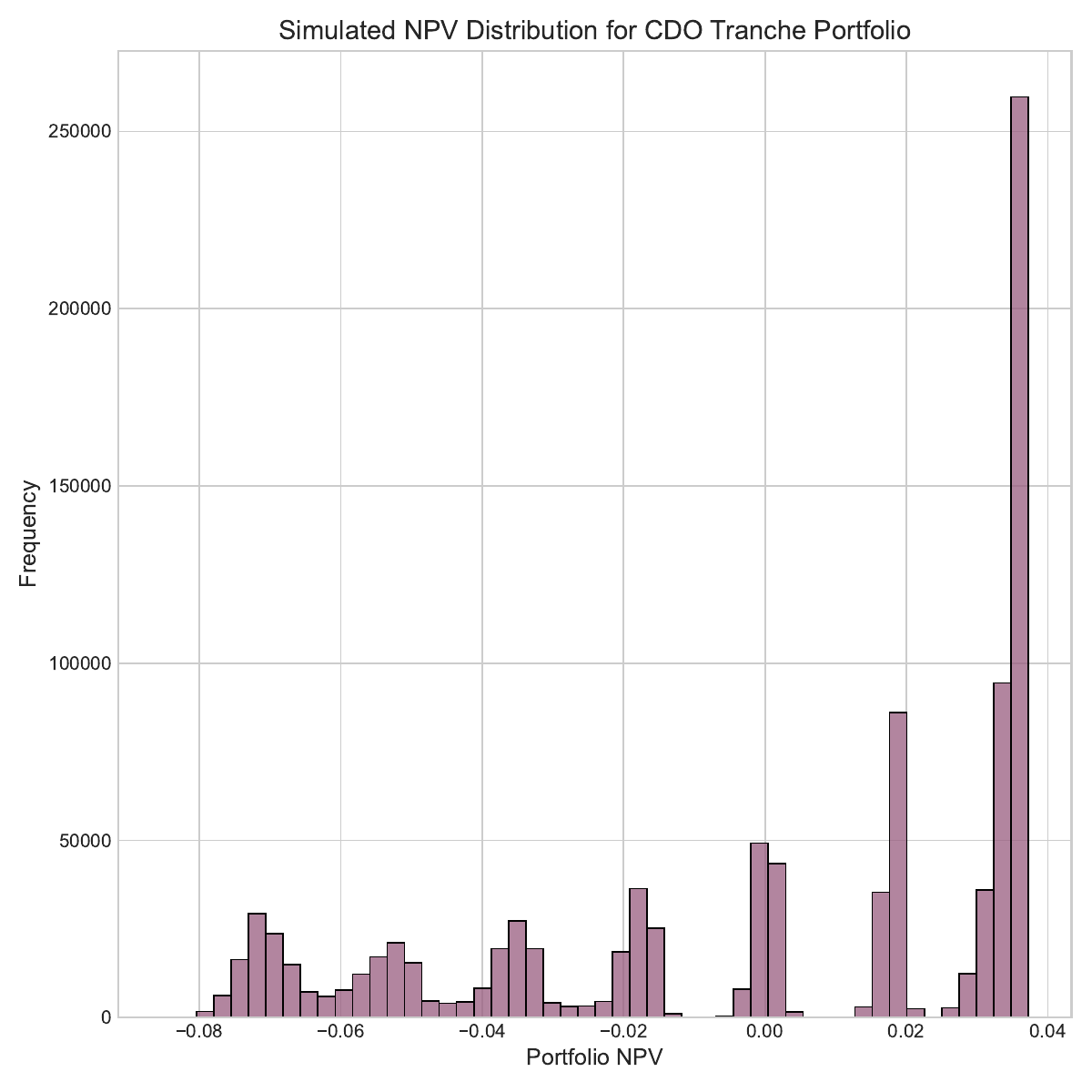}
	\end{subfigure}
	 
    \begin{minipage}{\textwidth}
	\small \linespread{1.2}\selectfont 
	
	\emph{{\footnotesize{Notes.}}} 
	{\footnotesize{The left panel illustrates the joint NPV distribution of the equity and senior mezzanine tranches (in terms of the scatter plot). The ``reverse L'' shape shown in the figure clearly reflects the defining waterfall payment structure of CDO products. 
		The right panel shows the NPV distribution of a specific CDO portfolio composed of short 4 units of the equity tranche, long 2 units of the junior mezzanine tranche, and short 1 unit of the senior mezzanine tranche. }} 
	\end{minipage}
	\label{fig:combined_dist}
\end{figure}
\vspace{-0.3cm}

By aggregating the simulated NPV values of individual tranches according to the specific CDO composition, we can analyze the NPV distribution 
of any CDO portfolio. The right panel of Figure \ref{fig:combined_dist} shows the NPV distribution of a specific CDO portfolio composed of short 4 units of the equity tranche, long 2 units of the junior mezzanine tranche, and short 1 unit of the senior mezzanine tranche. The resulting NPV distribution offers a practical tool for CDO portfolio risk management. 

\vspace{-0.3cm}

\subsection{Pricing Nonstandard CDOs}

\subsubsection{Pricing CDO Tranches with Nonstandard Attachment and Detachment Points}\label{sec_sub:numerical_weak_nonstandard_price_bound}
We apply Algorithm \ref{algo:weak_nonstandard_price_bound} to determine the model-independent and arbitrage-free price ranges for CDO tranches with nonstandard attachment and detachment points $[a, b]$. 

The left panel of Figure \ref{fig:nonstandard_ab_plots} presents the numerical results of the price ranges (in terms of spreads) for equity tranches with a fixed attachment point of $0\%$ and varying detachment points $b$ on March 28, 2025. As expected, the spread decreases as the detachment point $b$ increases. For very thick tranches (e.g., when $b$ approaches $100\%$), the spreads converge to the underlying index spread ($58$ bps). 
The arbitrage-free price range, denoted by the shaded area, is wider for detachment points typical of mezzanine tranches (approximately $20\%$ to $50\%$) and narrows for the most junior and senior tranche segments of the capital structure.
\vspace{-0.2cm}
\begin{figure}[ht]
	\centering
	\caption{Model-Independent and Arbitrage-Free Price Ranges for CDO Tranches with Nonstandard Attachment and Detachment Points}
	\begin{subfigure}[b]{0.48\textwidth}
		\centering
		\includegraphics[width=\textwidth,height=0.7\textwidth]{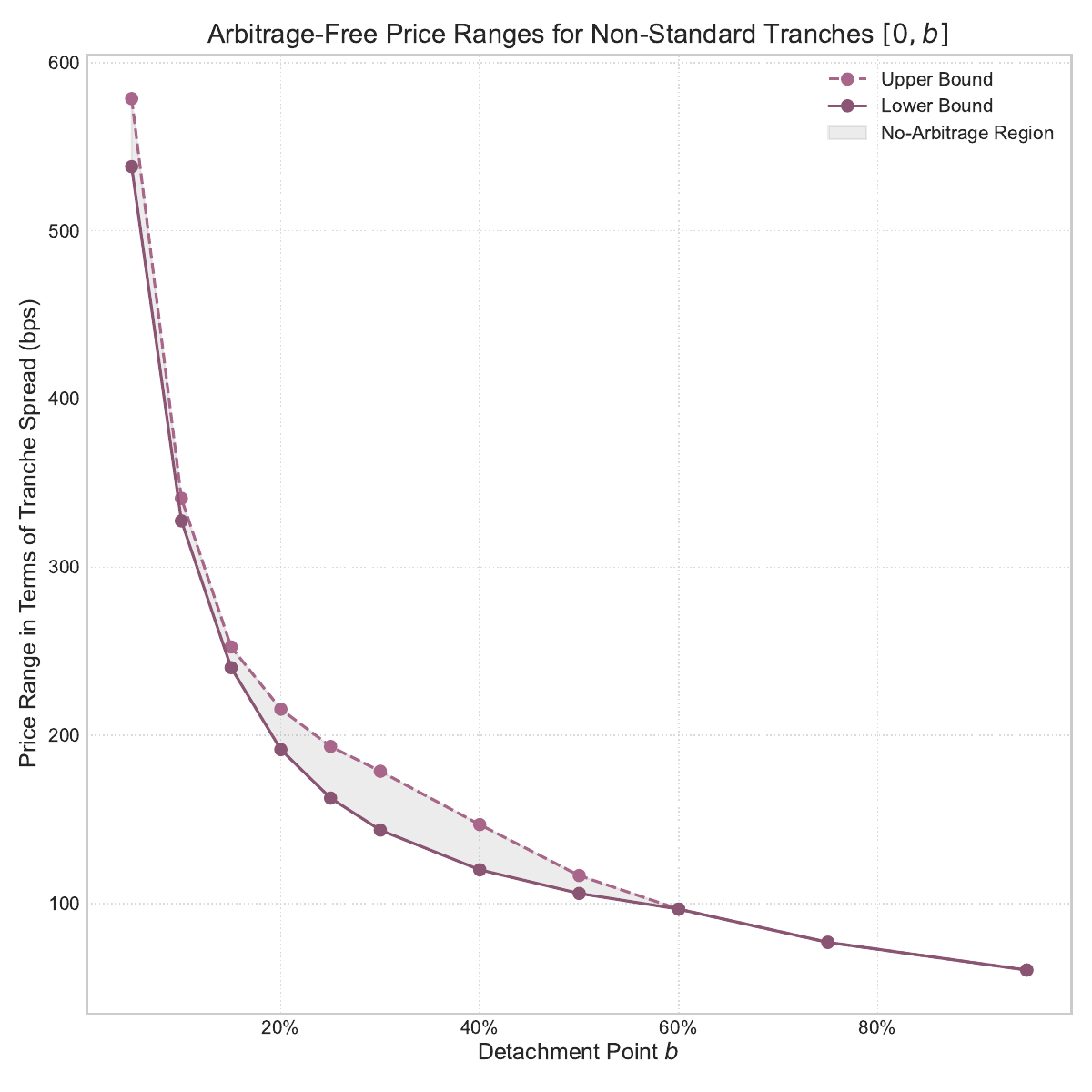}
	\end{subfigure}
	\hfill
	\begin{subfigure}[b]{0.48\textwidth}
		\centering
		\includegraphics[width=\textwidth,height=0.7\textwidth]{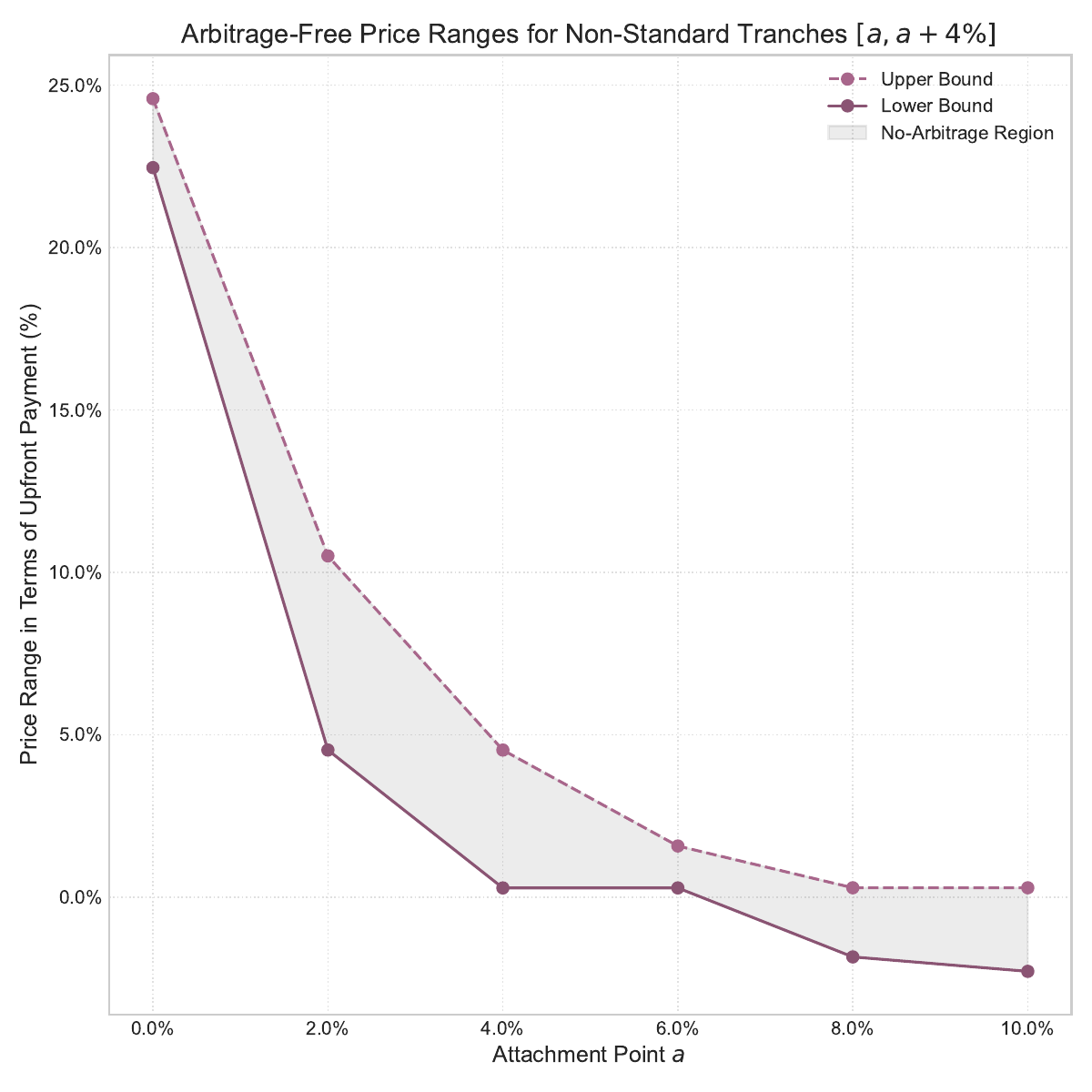}
	\end{subfigure}
    \begin{minipage}{\textwidth}
	\small \linespread{1.2}\selectfont 
	
	\emph{{\footnotesize{Notes.}}} 
	{\footnotesize{The left panel presents the numerical results of the model-independent and arbitrage-free price ranges (in terms of spreads) for equity tranches with a fixed attachment point of $0\%$ and varying detachment points $b$  on March 28, 2025. 
	The right panel  shows how the price ranges (in terms of up-front payments) for CDO tranches with varying attachment and detachment points $[a, a+4\%]$ having a constant width of $4\%$ change as $a$ increases (i.e., as tranche seniority varies). }}
		\end{minipage}
	\label{fig:nonstandard_ab_plots}
\end{figure}
\vspace{-0.4cm}

In addition, the right panel of Figure \ref{fig:nonstandard_ab_plots} shows how the price ranges (in terms of up-front payments) for CDO tranches with varying attachment and detachment points $[a, a+4\%]$ having a constant width of $4\%$ change as $a$ increases (i.e., as tranche seniority varies). 
It can be seen that the up-front payment is highest for the most junior tranche $[0\%, 4\%]$ and decreases monotonically as $a$ increases, eventually turning negative for senior tranches. Here a negative up-front payment means that the protection buyer receives a payment at inception, as the fixed running spread overcompensates for the low risk associated with these senior tranches.

\subsubsection{Pricing CDOs with a Nonstandard Number of Underlying Names}

We utilize Algorithm \ref{algo:nonstandard CDO 2} to compute the price ranges for CDOs with a nonstandard number of underlying names. Table \ref{tab:nonstandard_names} presents the calculated price ranges for the tranches of CDOs with varying nonstandard numbers of underlying names ($50$, $100$, $150$, and $200$), using the gamma-distorted copula with parameter $N=100$. We can see that for CDOs with a large number of underlying names (e.g., no fewer than $100$), the calculated price ranges are relatively narrow and close to the market prices of their standard counterpart with $125$ underlying names. 

\vspace{-0.2cm}

\begin{table}[htbp]
	\centering
	\caption{The Price Ranges for CDO Tranches with a Nonstandard Number of Names}
	\label{tab:nonstandard_names}
			\begin{center}
		\begin{small}
	\begin{tabular*}{\textwidth}{@{\extracolsep{\fill}} l c c c c @{}}
		\toprule
		\# of Names  & [0\%, 3\%] & [3\%, 6\%] & [6\%, 12\%] & [12\%, 100\%] \\
		\hline
		125 & 28.438\% & 4.531\% & 106.32 bps & 27.44 bps \\
		\hline
		50 & [25.917\%, 26.874\%] & [5.497\%, 6.444\%] & [107.89, 113.27] bps & [27.45, 27.74] bps \\
		100 & [28.091\%, 28.259\%] & [4.616\%, 4.800\%] & [106.25, 107.36] bps & [27.44, 27.49] bps\\
		150 & [28.573\%, 28.718\%] & [4.314\%, 4.475\%] & [105.43, 106.26] bps & [27.40, 27.44] bps \\
		200 & [28.760\%, 29.140\%] & [3.952\%, 4.382\%] & [104.38, 106.45] bps & [27.35, 27.44] bps \\
		\bottomrule
	\end{tabular*}		\end{small}
			\end{center}

\end{table}

\vspace{-0.1cm}


\section{Code and Data Disclosure}\label{sec:Code and Data Disclosure}The code and data to support our numerical experiments have been uploaded as a separate zip file. 


%
%
%



\begingroup
    \linespread{1.0}\selectfont 
    
    \setlength{\bibsep}{0pt} 

    \bibliographystyle{informs2014}
    \bibliography{CDOReference}
\endgroup





%
\ECSwitch

\noindent{\bf {\large E-Companion to ``Perfectly Fitting CDO Prices Across Tranches: A Theoretical Framework with Efficient Algorithms"}}
\vspace{0.3cm}


\section{Proofs of Propositions \ref{equiv} and  \ref{pricingprop}}\label{proof11}


\noindent{\textit{\textbf{Proof of Proposition \ref{equiv}.}}}
We will establish the equivalence among (i), (ii), and (iii) through the following implications: (i) $\Rightarrow$ (ii), (ii) $\Rightarrow$ (iii), and (iii) $\Rightarrow$ (i).

\noindent \textbf{(I)} Proof of (i) $\Rightarrow$ (ii).
Assume the existence of non-negative random variables $\hat\tau_j$ (for $1\leq j \leq n$) as specified in (i).
Let $\mathcal{G}_n$ denote the set of all permutations on the set $\{1, \ldots, n\}$, and let $\pi$ be a random permutation uniformly distributed on $\mathcal{G}_n$ that is independent of $(\hat\tau_1, \ldots, \hat\tau_n)$. 

Define a new random vector 
\begin{align}\label{eq:tauprihat}
	(\hat\tau_1', \ldots, \hat\tau_n') := (\hat\tau_{\pi(1)}, \ldots, \hat\tau_{\pi(n)}).
\end{align}
We can deduce that 
\begin{align}
\p(\hat\tau_1' \le x_1, \ldots, \hat\tau_n' \le x_n)
= &\ \E \left[
\p(\hat\tau_{\pi(1)} \le x_1, \ldots, \hat\tau_{\pi(n)} \le x_n|\pi)\right]\notag \\
= &\ \sum_{\sigma \in \mathcal{G}_n}  \left[\p(\pi=\sigma)
\p(\hat\tau_{\pi(1)} \le x_1, \ldots, \hat\tau_{\pi(n)} \le x_n|\pi=\sigma)\right]\notag \\
= &\ \frac{1}{n!}\sum_{\sigma \in \mathcal{G}_n} 
\p(\hat\tau_{\sigma(1)} \le x_1, \ldots, \hat\tau_{\sigma(n)} \le x_n),\label{eq:ec11sy}
\end{align}
where the third equality holds because $\pi$ is uniformly distributed on $\mathcal{G}_n$ and is independent of $(\hat\tau_1, \ldots, \hat\tau_n)$. \eqref{eq:ec11sy} implies that the joint distribution of $(\hat\tau_1', \ldots, \hat\tau_n')$ is the symmetrized version of that of $(\hat\tau_1, \ldots, \hat\tau_n)$.
Using \eqref{eq:ec11sy} and the substitution $\sigma' = \sigma \circ \rho^{-1}$, we can obtain that for any $\rho \in \mathcal{G}_n$, 
\[
\begin{aligned}
\p(\hat\tau_1' \le x_{\rho(1)}, \ldots, \hat\tau_n' \le x_{\rho(n)}) 
= & \ \frac{1}{n!}\sum_{\sigma \in \mathcal{G}_n} 
\p(\hat\tau_{\sigma(1)} \le x_{\rho(1)}, \ldots, \hat\tau_{\sigma(n)} \le x_{\rho(n)}) \\
=& \ \frac{1}{n!}\sum_{\sigma' \in \mathcal{G}_n} 
\p(\hat\tau_{\sigma'(1)} \le x_1, \ldots, \hat\tau_{\sigma'(n)} \le x_n)\\
=&\ \p(\hat\tau_1' \le x_1, \ldots, \hat\tau_n' \le x_n),
\end{aligned}
\]
which implies that $(\hat\tau_1', \ldots, \hat\tau_n')$ is exchangeable.

Next, we shall show that $(\hat\tau_1', \ldots, \hat\tau_n')$ shares the same DPM with $(\hat\tau_1, \ldots, \hat\tau_n)$. 
Define \[\hat N_t' := \sum_{j=1}^n \I_{\{\hat\tau_j' \le t\}}\equiv \sum_{j=1}^n \I_{\{\hat\tau_{\pi(j)} \le t\}} \quad\text{for $t\geq 0$}.\] 
Since the two sets $\{\pi(1),\cdots,\pi(n)\}$ and $\{1,\cdots,n\}$ are identical, it follows that 
\[
\hat N_t' 
= \sum_{j=1}^n \I_{\{\hat\tau_j \le t\}} = \hat N_t\quad\text{for $t\geq 0$}.
\]
Therefore, we have $\p(\hat N'_{T_i} = j) = \p(\hat N_{T_i} = j) = \hat q_{ij}$ for $1 \leq i \leq m$ and $0 \leq j \leq n$.
Accordingly, we conclude that $(\hat\tau_1', \ldots, \hat\tau_n')$ shares the same DPM with $(\hat\tau_1, \ldots, \hat\tau_n)$. 
This completes the proof of (i) $\Rightarrow$ (ii).

\vspace{0.2cm}
\noindent \textbf{(II)} Proof of (ii) $\Rightarrow$ (iii).
The non-negativity and sum-to-one constraints are apparent from the probabilistic definition of $(\hat\tau_1, \ldots, \hat\tau_n)$.  
Moreover, since $\hat N_{T_i} \leq \hat N_{T_{i+1}}$ for  $1 \leq i \leq m-1$, we can obtain 
\[
\p(\hat N_{T_i} \geq j) \leq \p(\hat N_{T_{i+1}} \geq j) 
\quad \text{for $1 \leq i \leq m-1$ and $0 \leq j \leq n$},
\]
which yields the monotonicity constraints in (\ref{qcond}). 
Hence, (ii) implies (iii).

\vspace{0.2cm}
\noindent \textbf{(III)} Proof of (iii) $\Rightarrow$ (i). 
We will use a constructive proof. 
Augment the matrix $\hat Q$ by adding boundary rows corresponding to $T_0 = 0$ and $T_{m+1}$ and setting $\hat q_{0j} = \I_{\{j=0\}}$ and $\hat q_{m+1,j} = \I_{\{j=n\}}$ for $0 \le j \le n$. 
Consider a random variable $U \sim \UU(0,1)$, and define a sequence of non-negative random variables 
\begin{equation}
	\hat\tau_j := \sum_{i=1}^{m+1} \left(\frac{T_i + T_{i-1}}{2} \I_{V_{ij}}  \right)
	\quad \text{for $1 \le j \le n$},
	\label{def_tau_hat}
\end{equation}
where 
$V_{ij}:=\left\{\theta_{i-1,j} < U \le \theta_{ij}\right\}$ for $1 \leq i \leq m+1$ and $1 \leq j \leq n$ and $\theta_{ij}:=\sum_{k \ge j}\hat q_{ik}$  for $0 \leq i \leq m+1$ and $1 \leq j \leq n$. 

We shall complete the proof of this part by taking three steps. First, let us present some properties of $\{\theta_{ij}\}$ and $\{V_{ij}\}$.
Using the third set of constraints (i.e., the non-negativity constraints) of \eqref{qcond}, we know from the definition that  $\theta_{ij}$ is non-increasing in $j$ for any $0\leq i\leq m+1$. 

Note that the second set of constraints (i.e., the monotonicity constraints) of \eqref{qcond} implies that $\theta_{1j}\leq \theta_{2j}\leq \cdots\leq \theta_{mj}$ for all $1\leq j\leq n$. Furthermore, we know from the definitions that $\theta_{0j}=0$ and $\theta_{m+1,j}=1$. In addition, using the first set of constraints (i.e., the sum-to-one constraints) and the third set of constraints (i.e., the non-negativity constraints) of \eqref{qcond}, we can obtain that $\theta_{mj}\leq 1$ and $\theta_{1j}\geq 0$, respectively. Therefore, we obtain that $\theta_{0j}\leq \theta_{1j}\leq \cdots\leq \theta_{mj}\leq \theta_{m+1,j}$ for all $1\leq j\leq n$. This further implies that for any $1\leq j\leq n$, $V_{1j}, V_{2j},\cdots,V_{m+1,j}$ are mutually disjoint. On the other hand, we can deduce that for any $1\leq j\leq n$,  
\begin{align}
	\bigcup_{i=1}^{m+1} V_{ij}=\bigcup_{i=1}^{m+1}\left\{\theta_{i-1,j} < U \le \theta_{ij}\right\}
	=\left\{\theta_{0j} < U \le \theta_{m+1,j}\right\}=\left\{0 < U \leq 1\right\},
\end{align}
which indicates that for any $1\leq j\leq n$, the union of $V_{1j}, V_{2j},\cdots,V_{m+1,j}$ is the complete set. Consequently, we conclude that for any $1\leq j\leq n$, $\hat\tau_j$ is a discrete random variable, taking a positive value $\frac{T_i+T_{i-1}}{2}$ on $V_{ij}$ for $1 \leq i \leq m+1$, where $V_{1j}, V_{2j},\cdots,V_{m+1,j}$ is a partition of the complete set and the $m+1$ positive values $\frac{T_1+T_{0}}{2}, \frac{T_2+T_{1}}{2},\cdots, \frac{T_{m+1}+T_{m}}{2}$ are increasing and mutually different.   

Second, we shall show that $(\hat\tau_1, \ldots, \hat\tau_n)$ possesses the following two properties. 

\noindent (i) They are ordered, i.e., $\hat\tau_1 \le \hat\tau_2 \le \cdots \le \hat\tau_n$.

\noindent (ii) It holds that $\{\hat\tau_j \le T_i\}\equiv \{U \le \sum_{k \ge j} \hat q_{ik}\}$ for $1 \leq i \leq m+1$ and $1 \leq j \leq n$. 

We start with the proof of property (i).  For any $1\leq j\leq n-1$ and $i=1,\cdots, m+1$, if $\hat\tau_j=\frac{T_i+T_{i-1}}{2}$, this means that $U\in(\theta_{i-1,j}, \theta_{ij}]$, which further implies that $U>\theta_{i-1,j+1}$ 
thanks to the fact that $\theta_{ij}$ is non-increasing in $j$ for any $0\leq i\leq m+1$. Then according to the definition of $\hat\tau_{j+1}$, we have that $\hat\tau_{j+1}\geq \frac{T_i+T_{i-1}}{2}=\hat\tau_j$. Thus, we have proved property (i). Regarding property (ii), we can deduce that  for $1 \leq i \leq m+1$ and $1 \leq j \leq n$,
\begin{align}
	\{\hat\tau_j \le T_i\}=&\ \bigcup_{l=1}^i \left\{\hat\tau_j =\frac{ T_l+T_{l-1}}{2}\right\}=\bigcup_{l=1}^i V_{lj}=\ \left\{U \le \sum_{k \ge j} \hat q_{ik}\right\}.\notag
\end{align}

Finally, we are ready to prove that the DPM of $(\hat\tau_1, \ldots, \hat\tau_n)$ is ${\{\hat q_{ij}\}}_{1 \leq i \leq m, 0 \leq j \leq n}$. Define $\hat N_{t}  
:=
\sum_{l=1}^n \I_{\{\hat\tau_l \le t\}}$ for all $t\geq 0$. 
Then we can obtain 
\begin{align}
	\p(\hat N_{T_i} = j)
	= &\ \p\left(\sum_{l=1}^n \I_{\{\hat\tau_l \le T_i\}} =  j\right)
	=\p\left(\hat\tau_j\leq T_i,\ \hat\tau_{j+1}>T_i\right) 
	= \p\left(\hat\tau_j\leq T_i\right) -\p\left(\hat\tau_{j+1}\leq T_i\right), \notag 
\end{align}
where the second and third equalities hold due to property (i) above. Then applying property (ii) above yields  
\begin{align}
	\p(\hat N_{T_i} = j)
	=&\ \p\left\{U \le \sum_{k \ge j} \hat q_{ik}\right\}-\p\left\{U \le \sum_{k \ge j+1} \hat q_{ik}\right\}
	=\hat q_{ij}. \notag
\end{align}
%
This completes the proof of (iii) $\Rightarrow$ (i).  

\vspace{0.2cm}
Thus, we have shown that (i) $\Rightarrow$ (ii), (ii) $\Rightarrow$ (iii), and (iii) $\Rightarrow$ (i). Accordingly,  the whole proof of Proposition \ref{equiv} is completed.
\hfill $\square$

\vspace{0.2cm}
\noindent\textit{\textbf{Proof of Proposition \ref{pricingprop}.}}
By definitions of $\dl^{[a_l, b_l]}$ and $\pl^{[a_l, b_l]}$ in Section \ref{sec_sub:value_CDO}, we can obtain
\[
\begin{aligned}
V^{[a_l, b_l]}=&\ \dl^{[a_l, b_l]} - \pl^{[a_l, b_l]}\\
=& \ \sum_{i=1}^m \!\left[D\!\left(\frac{T_{i-1}+T_i}{2}\right)
   \big(L_{T_i}^{[a_l,b_l]} - L_{T_{i-1}}^{[a_l,b_l]}\big)\right]
   \!-\! s^{[a_l,b_l]} \left[\sum_{i=1}^m \!D(T_i)(T_i - T_{i-1})(b_l - a_l - L_{T_i}^{[a_l,b_l]})\right]  \\
& \ - \uf^{[a_l,b_l]}(b_l - a_l).
\end{aligned}
\]
Collecting the coefficients of $L_{T_i}^{[a_l,b_l]}$ for all $i =1,\ldots, m-1$ and treating $i=m$ separately yields
\begin{align}
V^{[a_l, b_l]}
=&\ \sum_{i=1}^{m-1} 
   \!\left[\Bigg(
      D\!\left(\frac{T_{i-1}+T_i}{2}\right)
      - D\!\left(\frac{T_{i+1}+T_i}{2}\right)
      + s^{[a_l,b_l]} D(T_i)(T_i - T_{i-1})
   \Bigg)L_{T_i}^{[a_l,b_l]}\right] \notag \\
&\ + \Bigg(
      D\!\left(\frac{T_{m-1}+T_m}{2}\right)
      + s^{[a_l,b_l]} D(T_m)(T_m - T_{m-1})
   \Bigg)L_{T_m}^{[a_l,b_l]} \notag \\
&\ - (b_l - a_l)s^{[a_l,b_l]} \sum_{i=1}^m D(T_i)(T_i - T_{i-1})
   - \uf^{[a_l,b_l]}(b_l - a_l). \label{eq:vv}
\end{align}
Using the definitions of $\lambda_i^{[a_l, b_l]}(\mathbf{s})$ and $\gamma^{[a_l, b_l]}(\mathbf{s})$ in
\eqref{vspread_1} and \eqref{vspread_2}, \eqref{eq:vv} can be compactly written as
\begin{equation}\label{NPV_formula}
V^{[a_l, b_l]} = \sum_{i=1}^m \lambda_i^{[a_l, b_l]}(\mathbf{s}) L_{T_i}^{[a_l, b_l]}
- \gamma^{[a_l, b_l]}(\mathbf{s}).
\end{equation}

Next, we will express the expected tranche loss $\E[L_{T_i}^{[a_l,b_l]}]$ in terms of the DPM.
Indeed, we have 
\begin{align}
\E[L_{T_i}^{[a_l,b_l]}]
=&\ \E[(L_{T_i} - a_l)^+] - \E[(L_{T_i} - b_l)^+]\notag  \\
=&\ \E\!\left[\left(\frac{N_{T_i}(1-R)}{n} - a_l\right)^+\right]
   - \E\!\left[\left(\frac{N_{T_i}(1-R)}{n} - b_l\right)^+\right]\notag  \\
=&\ \sum_{j=0}^n 
   \left[
      \left(\frac{j(1-R)}{n} - a_l\right)^+
      - \left(\frac{j(1-R)}{n} - b_l\right)^+
   \right]
   \p(N_{T_i} = j)\notag  \\
=&\ \sum_{j=0}^n \beta_j^{[a_l, b_l]} q_{ij},\label{eq:vvv}
\end{align}
where the second equality holds because $L_{T_i} = N_{T_i}(1-R)/n$. 
Then taking expectation on both sides of \eqref{NPV_formula} and using \eqref{eq:vvv} immediately yields
\[
\begin{aligned}
v^{[a_l, b_l]}
&= \E[V^{[a_l, b_l]}]
 = \sum_{i=1}^m \lambda_i^{[a_l, b_l]}(\mathbf{s})\, \E[L_{T_i}^{[a_l,b_l]}]
   - \gamma^{[a_l, b_l]}(\mathbf{s}) = \sum_{i=1}^m \sum_{j=0}^n 
   \left[\lambda_i^{[a_l, b_l]}(\mathbf{s})\, \beta_j^{[a_l, b_l]} q_{ij}\right]
   - \gamma^{[a_l, b_l]}(\mathbf{s}),
\end{aligned}
\]
which can be written in matrix form as in \eqref{vc}.
The proof is completed.
\hfill $\square$

\section{Proofs of Theorems \ref{ImpliedMatrix} and \ref{thm:construct}}
\label{proof12}

\noindent\textit{\textbf{Proof of Theorem \ref{ImpliedMatrix}.}}
The sufficiency part (i.e., \textit{LP Feasibility} $\Rightarrow$ \textit{Weak Compatibility}) follows directly from Theorem~\ref{thm:construct}. 
We now prove the necessity part (i.e., \textit{Weak Compatibility} $\Rightarrow$ \textit{LP Feasibility}).

Assume that the market prices of CDO tranches $\mathbf{s}$ is weakly compatible. 
By definition, there exists a copula $C \in \mathcal{C}_0$ such that the following fair value condition holds.
\[
v^{[a_l, b_l]}(C; \mathbf{s}) = 0 \quad \text{for $1\leq l \leq M$}.
\]
Let $Q = \{q_{ij}\}_{1 \leq i \leq m,\, 0 \leq j \leq n}$ denote the DPM associated with this copula $C$.
We shall show that $Q$ constitutes a feasible solution to the system of linear constraints~\eqref{linearsystem}.
According to Proposition~\ref{pricingprop}, the condition $v^{[a_l, b_l]}(C; \mathbf{s}) = 0$ for $1\leq l \leq M$ is equivalent to the first set of linear constraints in~\eqref{linearsystem}. 
Since $Q$ is a valid DPM, it satisfies the sufficient and necessary condition given in (iii) of  Proposition~\ref{equiv}, which corresponds to the second, fourth, and fifth sets of constraints in~\eqref{linearsystem}. 
Moreover, by the definition of the DPM, the expected number of defaults at each $T_i$ satisfies
\begin{align}\label{eq:app2-1}
	\E[N_{T_i}] = \sum_{j=0}^{n} j\, q_{ij} \quad \text{for  $1 \leq i \leq m$}.
\end{align}
On the other hand, by the definition of $N_{T_i}$, 
simple algebra yields 
\begin{align}\label{eq:app2-2}
\E[N_{T_i}] =\E\left[\sum_{j=1}^n\I_{\{\tau_j\leq t\}}\right]= n\,F(T_i)\quad \text{for  $1 \leq i \leq m$}.
\end{align}
Then we conclude from \eqref{eq:app2-1} and \eqref{eq:app2-2} that $Q$ also satisfies the third set of constraints in~\eqref{linearsystem}. Therefore, we have proved that the DPM $Q$ associated with the copula $C$ satisfies all five sets of linear constraints in~\eqref{linearsystem} and thus constitutes a feasible solution to ~\eqref{linearsystem}. 
\hfill $\square$
\vspace{0.2cm}

\noindent\textit{\textbf{Proof of Theorem \ref{thm:construct}.}}
For ease of exposition, define $\hat Q := \{\hat q_{ij}\}_{1\leq i\leq m,0\leq j\leq n}$. 
By taking the following four steps, we will construct an exchangeable copula $\hat C \in \mathcal{C}_0$ (with DPM equal to $\hat Q$) that perfectly fits all market prices. 


\noindent \emph{Step 1. Constructing a sequence of ordered, non-negative random variables  $(\tilde\tau_1,\ldots,\tilde\tau_n)$.}
Augment $\hat Q$ by setting $\hat q_{0j} = \I_{\{j=0\}}$ for $T_0 = 0$ and $\hat q_{m+1,j} = \I_{\{j=n\}}$ for $T_{m+1} > T$. 
Consider a random variable $U \sim \UU(0,1)$, and define a sequence of non-negative random variables as follows:
\[
\tilde\tau_j 
:= \sum_{i=1}^{m+1} \frac{T_i + T_{i-1}}{2}\,
\I_{\left\{\sum_{k \ge j} \hat q_{i-1,k} < U \le \sum_{k \ge j} \hat q_{ik}\right\}}
\quad \text{for $1 \le j \le n$}.
\]
Since  $(\tilde\tau_1,\ldots,\tilde\tau_n)$ are defined in the same way as  $(\hat\tau_1,\ldots,\hat\tau_n)$ in \eqref{def_tau_hat}, they possess the same properties (i) and (ii) as for $(\hat\tau_1,\ldots,\hat\tau_n)$ (see part {\bf (III)} of the proof of Proposition \ref{equiv}). 
Specifically, $(\tilde\tau_1,\ldots,\tilde\tau_n)$ have the following two properties, where property (i) exactly means that $(\tilde\tau_1,\ldots,\tilde\tau_n)$ are ordered. 

\noindent (i) They are ordered, i.e., $\tilde\tau_1 \le \tilde\tau_2 \le \cdots \le \tilde\tau_n$.

\noindent (ii) It holds that $\{\tilde\tau_j \le T_i\}\equiv \{U \le \sum_{k \ge j} \hat q_{ik}\}$ for $1 \leq i \leq m+1$ and $1 \leq j \leq n$. 


Because $(\tilde\tau_1,\ldots,\tilde\tau_n)$ are ordered, non-negative random variables, we can regard them as ordered default times. Then we know from part {\bf (III)} of the proof of Proposition \ref{equiv} that $(\tilde\tau_1,\ldots,\tilde\tau_n)$ possesses the following property (iii).  

\noindent (iii) The DPM of $(\tilde\tau_1,\ldots,\tilde\tau_n)$ equals $\hat Q$. 


Applying the above property (ii), we can obtain that for $y_1,\ldots,y_n \in \{0,\ldots,m+1\}$, 
\begin{equation}
	\begin{aligned}
		\p(\tilde\tau_1 \le T_{y_1},\ldots,\tilde\tau_n \le T_{y_n})
		&= \p\!\left(U \le \sum_{k \ge 1}\hat q_{y_1,k},\ldots,U \le \sum_{k \ge n}\hat q_{y_n,k}\right) \\
		&= \p\!\left(U \le \min_{1 \le j \le n}\!\left\{\sum_{k \ge j} \hat q_{y_j,k}\right\}\right)
		= \min_{1 \le j \le n}\!\left\{\sum_{k \ge j} \hat q_{y_j,k}\right\}.
	\end{aligned}
	\label{tilde_dist}
\end{equation}

\vspace{0.1cm}
\noindent \emph{Step 2. Constructing exchangeable default times $(\hat\tau_1,\ldots,\hat\tau_n)$ with DPM equal to $\hat Q$.}
Let $\mathcal{G}_n$ denote the set of all permutations on the set $\{1,\ldots,n\}$, and let $\pi$ be a random permutation uniformly distributed on $\mathcal{G}_n$ that is independent of~$U$. 
Define a sequence of non-negative random variables 
\[
(\hat\tau_1,\ldots,\hat\tau_n) := (\tilde\tau_{\pi(1)},\ldots,\tilde\tau_{\pi(n)}),
\]
which can be regarded as default times. 
Because $(\hat\tau_1,\ldots,\hat\tau_n)$ is defined in the same way as $(\hat\tau'_1,\ldots,\hat\tau'_n)$ in \eqref{eq:tauprihat}, we conclude that $(\hat\tau_1,\ldots,\hat\tau_n)$ is also exchangeable (see part {\bf (I)} of the proof of Proposition \ref{equiv}). 


Furthermore, define $\hat N_t := \sum_{j=1}^n \I_{\{\hat\tau_j \le t\}}\equiv \sum_{j=1}^n \I_{\{\tilde\tau_{\pi(j)} \le t\}}$ and $\tilde N_t := \sum_{j=1}^n \I_{\{\tilde\tau_j \le t\}}$ for $t\geq 0$. 
Since the two sets $\{\pi(1),\cdots,\pi(n)\}$ and $\{1,\cdots,n\}$ are identical, it follows that 
\[
\hat N_t 
= \sum_{j=1}^n \I_{\{\tilde\tau_j \le t\}} = \tilde N_t\quad\text{for $t\geq 0$}.
\]
Therefore, we have $\p(\hat N_{T_i} = j) = \p(\tilde N_{T_i} = j) = \hat q_{ij}$ for $1 \leq i \leq m$ and $0 \leq j \leq n$.
Namely, the DPM of $(\hat\tau_1, \ldots, \hat\tau_n)$ is equal to $\hat Q$.
%
%

\vspace{0.1cm}
\noindent \emph{Step 3. Proving that for $y_1,\ldots,y_n \in \{0,\ldots,m+1\}$, }
\begin{align}\label{eq:step3}
\p(\hat\tau_1 \le T_{y_1},\ldots,\hat\tau_n \le T_{y_n}) = \frac{1}{n!} \sum_{\sigma \in \mathcal{G}_n}
\min_{1 \le j \le n}\!\left\{\sum_{k \ge \sigma(j)} \hat q_{y_j,k}\right\}.
\end{align}
We can deduce that for $y_1,\ldots,y_n \in \{0,\ldots,m+1\}$,
\begin{align}
	\p(\hat\tau_1 \le T_{y_1},\ldots,\hat\tau_n \le T_{y_n}) 
	= &\ \E \left[
	\p(\tilde\tau_{\pi(1)} \le T_{y_1}, \ldots, \tilde\tau_{\pi(n)} \le T_{y_n}|\pi)\right]\notag \\
	= &\ \sum_{\sigma \in \mathcal{G}_n}  \left[\p(\pi=\sigma)
	\p(\tilde\tau_{\pi(1)} \le T_{y_1}, \ldots, \tilde\tau_{\pi(n)} \le T_{y_n}|\pi=\sigma)\right]\notag \\
	= &\frac{1}{n!} \sum_{\sigma \in \mathcal{G}_n}
	\p(\tilde\tau_{\sigma(1)} \le T_{y_1},\ldots,\tilde\tau_{\sigma(n)} \le T_{y_n}),
	\label{eq:step3-2}
\end{align}
where the third equality holds as $\pi$ is uniformly distributed on $\mathcal{G}_n$ and is independent of $(\tilde\tau_1, \ldots, \tilde\tau_n)$. 

Rearranging the terms inside the probabilities on the right hand side of \eqref{eq:step3-2} and then applying ~\eqref{tilde_dist} yields
\begin{equation*}
	\begin{split}
	\p(\hat\tau_1 \le T_{y_1},\ldots,\hat\tau_n \le T_{y_n})
	&= \frac{1}{n!} \sum_{\sigma \in \mathcal{G}_n}
	\p(\tilde\tau_1 \le T_{y_{\sigma^{-1}(1)}},\ldots,\tilde\tau_n \le T_{y_{\sigma^{-1}(n)}}) \\
	&= \frac{1}{n!} \sum_{\sigma \in \mathcal{G}_n}
	\min_{1 \le l \le n}\!\left\{\sum_{k \ge l} \hat q_{y_{\sigma^{-1}(l)},k}\right\}.
	\end{split}
\end{equation*}
By setting $j = \sigma^{-1}(l)$ (and thus $l=\sigma(j)$), we can obtain \eqref{eq:step3} immediately.

\vspace{0.1cm}
\noindent \emph{Step 4. Constructing an exchangeable copula $\hat C \in \mathcal{C}_0$ (with DPM equal to $\hat Q$) that perfectly fits all market prices.}
Define $\hat C$ to be the copula of $(\hat\tau_1,\ldots,\hat\tau_n)$ constructed in Step 2 above. 
For $u_j := F(T_{y_j})$ ($1 \le j \le n$), applying Sklar's theorem and then using \eqref{eq:step3} yields
\[
\hat C(u_1,\ldots,u_n)
= \p(\hat\tau_1 \le T_{y_1},\ldots,\hat\tau_n \le T_{y_n})
= \frac{1}{n!} \sum_{\sigma \in \mathcal{G}_n}
\min_{1 \le j \le n}\!\left\{\sum_{k \ge \sigma(j)} \hat q_{y_j,k}\right\}.
\]
Furthermore, we know from Step 2 above that $(\hat\tau_1,\ldots,\hat\tau_n)$ is exchangeable with DPM equal to $\hat Q$. This implies that $\hat C$ is exchangeable and its associated DPM equals $\hat Q$. Therefore, $\hat C$ can perfectly fit all market prices. The proof is completed. 
\hfill $\square$

\section{Proofs of Propositions \ref{ExchangeProp}, \ref{propx} and \ref{prop:generator}}\label{proof21}

\noindent\textit{\textbf{Proof of Proposition \ref{ExchangeProp}}.}
\noindent\textbf{(I)} Proof of (i) $\Rightarrow$ (ii).
Assume that ${\{U_k\}}_{k \in \mathbb{N}^+}$ are i.i.d. conditional on a $\sigma$-field $\mathcal{H}\subseteq \mathcal F$.
Define a stochastic process $X(u) := \mathbb P(U_1 \leq u | \mathcal{H})$ for $u\in[0,1]$. 
Since $U_1 \sim \mathcal{U}(0, 1)$, it is straightforward to check that $\{X(u): u\in[0,1]\}$ is a stochastic distortion that satisfies $\E [X(u)] = u$ for any $u\in[0,1]$. Therefore, we know that $C^X(u_1, \ldots, u_n):=\E\left[\prod_{j=1}^n X(u_j)\right]$ is the transformed copula by the stochastic distortion $X$. Then 
we can obtain that for any $n\geq 2$, the copula $C$ of $(U_1,\ldots, U_n)$ is given by 
\[
\begin{split}
	C(u_1, \ldots, u_n) &= \p (U_1 \leq u_1, \ldots, U_n \leq u_n)  = \E [\p(U_1 \leq u_1, \ldots, U_n \leq u_n |\mathcal{H})] \\
	&= \E\left[\prod_{j=1}^n \p(U_j \leq u_j | \mathcal{H})\right] = \E\left[\prod_{j=1}^n X(u_j)\right]=C^X(u_1, \ldots, u_n), 
\end{split}
\]
where the third equality follows from the conditional independence of  ${\{U_k\}}_{k \in \mathbb{N}^+}$. 

\noindent\textbf{(II)} Proof of (ii) $\Rightarrow$ (i).
Assume that there exists a stochastic distortion $X$ with $\E [X(u)] = u$ for any $u\in[0,1]$ such that for any $n \geq 2$, the copula $C$ of $(U_1, \ldots, U_n)$ admits the stochastic distortion representation given by $\eqref{Dist2}$.
Let $\sigma$ be any permutation on the set $\{1, \ldots, n\}$.
Then we have
\[
C(u_{\sigma(1)}, \ldots, u_{\sigma(n)}) = \E\left[\prod_{j=1}^n X(u_{\sigma(j)})\right] = \E\left[\prod_{j=1}^n X(u_j)\right] = C(u_1, \ldots, u_n), 
\] 
which implies that $(U_1, \ldots, U_n)$ is exchangeable. Since this holds for any arbitrary $n \geq 2$, we conclude that ${\{U_k\}}_{k \in \mathbb{N}^+}$ is infinitely exchangeable by definition. According to de Finetti's Theorem (see, e.g., \citealt{mai2020infinite}), we obtain the result (i) immediately. The proof is completed.
\hfill $\square$

\vspace{0.2cm}
\noindent\textit{\textbf{Proof of Proposition \ref{propx}.}}
By the properties of gamma processes, both $\{\xi_t: t \geq 0\}$ and $\{\eta_t: t \geq 0\}$ are non-negative and non-decreasing processes. Therefore, it holds that 
\[\xi_{\phi(u_1)}\leq \xi_{\phi(u_2)}\quad\text{and}\quad \eta_{N - \phi(u_1)}\geq \eta_{N - \phi(u_2)}\quad \text{a.s.} \quad \text{for any $0 < u_1 < u_2 < 1$},\]
where we also use the fact that $\{\phi(u)\}$ is non-decreasing. Then we can deduce that for any $0 < u_1 < u_2 < 1$, 
\[
X(u_1)
= \frac{\xi_{\phi(u_1)}}{\xi_{\phi(u_1)} + \eta_{N - \phi(u_1)}}
\le \frac{\xi_{\phi(u_2)}}{\xi_{\phi(u_2)} + \eta_{N - \phi(u_1)}}
\le \frac{\xi_{\phi(u_2)}}{\xi_{\phi(u_2)} + \eta_{N - \phi(u_2)}}
= X(u_2) \quad \text{a.s.},
\]
which means that $\{X(u): u \in [0, 1]\}$ is non-decreasing a.s.. In addition, note that $\xi_0 = \eta_0 = 0$, $\phi(0)=0$, $\phi(1)=N$ a.s.. It follows that $X(0) = 0$ and $X(1) = 1$ a.s.. Therefore, $\{X(u)\}$ is a stochastic distortion. 
 
Furthermore, we can deduce that 
\[
\E[X(u)]
 = \E\big[\E[X(u) \mid \phi(u)]\big]
 = \E\!\left[\frac{\phi(u)}{\phi(u) + (N - \phi(u))}\right]
 = u \quad \text{for  $0 < u < 1$},
\]
where the second equality holds because conditional on $\phi(u)$, $X(u)$ follows a Beta distribution with parameters $\phi(u)$ and $N-\phi(u)$. The proof is completed.  
\hfill $\square$

\vspace{0.2cm}
\noindent\textit{\textbf{Proof of Proposition \ref{prop:generator}}.}
Define $V_{ik}:=\left\{\theta_{i,k-1} < U \le \theta_{ik}\right\}$ for $0 \leq i \leq m+1$ and $0 \leq k \leq N$, where $\theta_{ik}:=\sum_{j \leq k}\hat p_{ij}$  for $0 \leq i \leq m+1$ and $0 \leq k \leq N$ and $\theta_{i,-1}:=0$. 
Applying the fourth set of constraints (i.e., the non-negativity constraints) of \eqref{linearsystem2_delete_1} yields directly that $0\equiv \theta_{i,-1}\leq \theta_{i0}\leq \cdots\leq \theta_{iN}$ for all $0\leq i\leq m+1$.
Moreover, we know from the first set of constraints (i.e., the sum-to-one constraints) of \eqref{linearsystem2_delete_1} that $\theta_{iN}\equiv 1$. 
Therefore, we conclude that for any $i=0,\cdots, m+1$,  $V_{i0}, V_{i1},\cdots,V_{iN}$ is a partition of the complete set. 
Using this partition result and noticing that $\phi(F(T_i))$ can be rewritten as 
$\phi(F(T_i))=\sum_{k=0}^N k\I_{V_{ik}}$ for any $i=0,\cdots,m+1$,
we can obtain that for any $i=0,\cdots,m+1$ and $k=0,\cdots, N$,  
\begin{align}\label{eq:iff}
\text{ $\phi(F(T_i))=k$ if and only if  
$\theta_{i,k-1} < U \leq \theta_{ik}$ (i.e., $V_{ik}$ happens).}
\end{align} 
For any $i=0,\cdots,m+1$ and $k=0,\cdots, N$, if $\phi(F(T_i))=k$, then we have $\theta_{i,k-1} < U \leq \theta_{ik}$. 
This further implies that $U>\theta_{i+1,k-1}$ because $\theta_{i+1,k-1}\leq \theta_{i,k-1}$ due to the third set of constraints (i.e., the monotonicity constraints) of \eqref{linearsystem2_delete_1}. 
Then using the result \eqref{eq:iff} again yields that $\phi(F(T_{i+1}))\geq k=\phi(F(T_i))$. 
Therefore, we have $\phi(F(T_0))\leq \phi(F(T_1))\leq \cdots \leq \phi(F(T_{m+1}))$. 
Since $\phi(u)$ is a linear interpolation of $\phi(F(T_i))$ and $\phi(F(T_{i+1}))$ for any $u\in (F(T_i),F(T_{i+1}))$ and $i=0,\cdots,m$, we conclude that $\{\phi(u)\}$ is a non-decreasing stochastic process.

In addition, it is straightforward to verify that $\phi(0)=0$ and $\phi(1)=N$. 
Furthermore, we know from the result \eqref{eq:iff} that the distribution of $\phi(F(T_i))$ is given by \eqref{eq:phi_distribution}.  
Then we can deduce that 
\begin{align}\label{eq:exp}
	\E(\phi(F(T_i))) = \sum_{k=0}^Nk\hat p_{ik}=NF(T_i)\quad\text{for any $i=1,\cdots,m$}, 
\end{align} 
where the second equality holds due to the second set of constraints of \eqref{linearsystem2_delete_1}. 
Using the fact that $\phi(u)$ is a linear interpolation of $\phi(F(T_i))$ and $\phi(F(T_{i+1}))$ for any $u\in (F(T_i),F(T_{i+1}))$ and $i=0,\cdots,m$, we obtain that $\E(\phi(u))=u$ for any $0\leq u\leq 1$. 

Thus, we have shown that $\{\phi(u): 0 \le u \le 1\}$ defined in \eqref{implied_phi} is a non-decreasing stochastic process satisfying \eqref{cond_phipsi} and \eqref{eq:phi_distribution}. \hfill $\Box$

\section{Proofs of Theorems \ref{strongly_compatible} and \ref{construct_copula}}
\label{proof22}

\noindent\textit{\textbf{Proof of Theorem \ref{strongly_compatible}.}}
The first implication ($\mathbf{s} \in \mathcal{A}^\circ \Rightarrow$ \textit{LP feasibility}) follows directly from Proposition~\ref{prop_interior}, whose proof will be given in Section \ref{proof23}. 
Conversely, the second implication (\textit{LP feasibility} $\Rightarrow \mathbf{s} \in \mathcal{A}$) is ensured by Theorem~\ref{construct_copula}, which constructs a concrete conditionally i.i.d. copula that achieves a perfect fit to the market prices $\mathbf{s}$. 
\hfill $\square$

\vspace{0.2cm}

\noindent\textit{\textbf{Proof of Theorem \ref{construct_copula}}.}
Given that $\{\hat p_{ik}\}_{1 \leq i \leq m, 0 \leq k \leq N}$ is a feasible solution to the system of linear constraints \eqref{linearsystem2} (the last four sets of which are exactly the linear constraints in \eqref{linearsystem2_delete_1}), we know from Proposition \ref{prop:generator} that $\{\phi(u): 0 \le u \le 1\}$ defined in \eqref{implied_phi} via this feasible solution is a non-decreasing stochastic process satisfying \eqref{cond_phipsi}, and thus it is qualified to serve as a generator of a discrete gamma distortion.
Moreover, the distribution of $\phi(F(T_i))$ is given by \eqref{eq:phi_distribution} for any $i=1, \cdots, m$. 

Accordingly, using $\{\phi(u)\}$ defined in \eqref{implied_phi} and two independent gamma processes $\{\xi_t\}$ and $\{\eta_t\}$ that are also independent of $\{\phi(u)\}$, we can construct a discrete gamma distortion $\{X(u)\}$ based on \eqref{Xdef} and then denote by $C^X$ the corresponding discrete gamma-distorted copula. It follows from \eqref{qij_discrete} and \eqref{eq:phi_distribution} that the DPM $Q$ associated with $C^X$ is given by
\begin{align}\label{eq:last}
q_{ij} = \sum_{k=0}^N h_{jk}^{(N)} \hat p_{ik} \quad \text{for  $1\leq i \leq m$ and $0 \leq j \leq n$.}
\end{align}
Then we can deduce that 
\begin{align}
	v^{[a_l, b_l]}(C^{X}; \mathbf{s}) =&\ \sum_{i = 1}^m \sum_{j = 0}^n \lambda^{[a_l, b_l]}_i(\mathbf{s}) \beta_j^{[a_l ,b_l]}  q_{ij} - \gamma^{[a_l, b_l]}(\mathbf{s})\\
	=& \ \sum_{i=1}^m \sum_{j=0}^n \sum_{k=0}^N  h_{jk}^{(N)} \lambda_i^{[a_l,b_l]}(\mathbf{s}) \beta_j^{[a_l,b_l]} \hat p_{ik} - \gamma^{[a_l, b_l]}(\mathbf{s}) \notag\\
	=&\ 0 \quad \text{for $1\leq l \leq M$,}
\end{align}
where the first equality follows from \eqref{vc},  the second equality is obtained from \eqref{eq:last}, and  the third equality holds due to the first set of constraints in \eqref{linearsystem2}. 
Hence, we conclude that the copula $C^X$ achieves a perfect fit to the market prices $\mathbf{s}$. The proof is completed. 
\hfill $\square$

\section{Proof of Proposition \ref{prop_interior}}\label{proof23}

\subsection{Preliminary Preparations}

Before proving Proposition \ref{prop_interior}, we first introduce some notations and present two lemmas as preparation.

For any $N\in\mathbb N^+$, denote by $\mathcal{P}^{(N)}$ the feasible region to the following system of linear constraints \eqref{ls3} (note that \eqref{ls3} is the same as \eqref{linearsystem2_delete_1} except that the superscript $(N)$ is added to highlight the dependence of all the variables on $N$). 
\begin{equation}
	\begin{cases}
		\sum_{k=0}^N \hat p_{ik}^{(N)} = 1 \, & \text{for } 1 \leq i\leq m,\\
		\sum_{k=0}^N k \hat p_{ik}^{(N)} = N F(T_i)\, & \text{for }  1 \leq i \leq m, \\
		\sum_{j \geq k} \hat p_{ij}^{(N)} \leq \sum_{j \geq k} \hat p_{i+1, j}^{(N)}\, & \text{for   $0 \leq k \leq N$ and  $1\leq i \leq m - 1$},\\
		\hat p_{ik}^{(N)} \geq 0, \, & \text{for $1 \leq i \leq m$ and $0\leq k\leq N$}.
		\label{ls3}
	\end{cases}
\end{equation}
For each matrix $\hat P^{(N)} = \{\hat p^{(N)}_{ik}\}_{1 \leq i \leq m, 0 \leq k \leq N} \in \mathcal{P}^{(N)}$, Proposition \ref{prop:generator} ensures the existence of a discrete gamma-distorted copula $C^X$ whose generator $\phi$ satisfies 
\begin{equation}\label{eq:phi_distribution_revised}
\p(\phi(F(T_i)) = k) = \hat p_{ik}^{(N)} \quad \text{for $1\leq i \leq m, 0 \leq k \leq N$}.
\end{equation}
Consequently, given $\mathbf{s} \in \mathbb{R}^{2M}$, we can abuse the notation of $v^{[a_l,b_l]}$, using $v^{[a_l,b_l]}(\hat P^{(N)}; \mathbf{s})$ to denote the value of $v^{[a_l,b_l]}(C^X; \mathbf{s})$ when the copula $C^X$ is a discrete gamma distorted copula whose generator satisfies \eqref{eq:phi_distribution_revised}.
Specifically, by substituting \eqref{qij_discrete} into \eqref{vc}, we have
\begin{equation}\label{vPdef}
v^{[a_l,b_l]}(\hat P^{(N)}; \mathbf{s}) = \sum_{i=1}^m \sum_{j=0}^n \sum_{k=0}^N  h_{jk}^{(N)} \lambda_i^{[a_l,b_l]}(\mathbf{s}) \beta_j^{[a_l,b_l]} \hat p_{ik}^{(N)} - \gamma^{[a_l, b_l]}(\mathbf{s}) \quad \text{for $1\leq l \leq M$}.
\end{equation}
Furthermore, we define $\mathbf{v}({\hat P^{(N)}; \mathbf{s}}) := (v^{[a_1, b_1]}(\hat P^{(N)}; \mathbf{s}), \ldots, v^{[a_M, b_M]}(\hat P^{(N)}; \mathbf{s})) \in \mathbb{R}^{M}$.

The following lemma shows that for each $\mathbf{s} \in \mathcal{A}^\circ$, there exists a sequence of matrices in $\mathcal{P}^{(N)}$ such that the corresponding values of $\mathbf{v}(\hat P^{(N)}; \mathbf{s})$ converge to the zero vector as $N$ goes to infinity.

\begin{Lemma}\label{lemma:convergence}
	For each $\mathbf{s} \in \mathcal{A}$, let $X_{\mathbf{s}}$ be a stochastic distortion such that its transformed copula $C^{X_{\mathbf{s}}}$ provides a perfect fit for $\mathbf{s}$. 
	Then there exists a sequence of matrices $\hat{P}^{(N)}(X_\mathbf{s}) \in \mathcal{P}^{(N)}$ for $N \in \mathbb{N}^+$ such that
	\begin{equation}\label{eq:convergence}
	\lim_{N\to +\infty}\|\mathbf{v}(\hat P^{(N)}(X_{\mathbf{s}}); \mathbf{s})\| = 0,
\end{equation}
    where $\|\cdot\|$ is the maximum norm.
\end{Lemma}
\noindent{\it Proof.} 
The proof is divided into three steps.
First, we construct the matrices $\hat{P}^{(N)}(X_\mathbf{s}) \in \mathcal{P}^{(N)}$ for $N \in \mathbb{N}^+$ based on the stochastic distortion $X_{\mathbf{s}}$.
Second, we show that these matrices can be used to construct gamma distortions $X_\mathbf{s}^{(N)}$ for $N \in \mathbb{N}^+$, and that $X_\mathbf{s}^{(N)}(F(T_i))$ converges in distribution to $X_{\mathbf{s}}(F(T_i))$ as $N$ goes to $+\infty$ for $1 \leq i \leq m$.
Third, we prove the convergence result \eqref{eq:convergence}.

\vspace{0.2cm}
\noindent\textit{Step 1. Constructing matrices $\hat{P}^{(N)}(X_\mathbf{s}) \in \mathcal{P}^{(N)}$ for $N \in \mathbb{N}^+$.}
For each $N \in \mathbb{N}^+$, we define the matrix $\hat{P}^{(N)}(X_\mathbf{s}) = \{\hat p_{ik}^{(N)}(X_\mathbf{s})\}_{1 \leq i \leq m, 0 \leq k \leq N}$ as follows:
\[
\hat p_{ik}^{(N)}(X_\mathbf{s}) := \int_{((k-1)/ N, (k + 1)/ N)} (1 - |Ny - k|) \diff G_{X_\mathbf{s}(F(T_i))}(y) \quad \text{for $0 \leq k \leq N$ and $1 \leq i \leq m$},
\]
where $G_{X_\mathbf{s}(F(T_i))}$ is the distribution function of the random variable $X_\mathbf{s}(F(T_i))$.
We shall verify that $\hat{P}^{(N)}(X_\mathbf{s}) \in \mathcal{P}^{(N)}$. 
To this end, we check the four sets of constraints in \eqref{ls3} one by one. 

First, we note the identity
\begin{equation}
	\label{identity1}
\sum_{k=0}^N \I_{\{(k-1)/ N < y < (k + 1)/ N\}} (1 - |Ny - k|) = 1 \quad \text{for $y \in [0, 1]$ and  $N \in \mathbb{N}^+$}.
\end{equation}
We can verify the first set of linear constraints as follows:
\begin{align*}
	\sum_{k=0}^N \hat p_{ik}^{(N)}(X_\mathbf{s}) &= \sum_{k=0}^N \int_{((k-1)/ N, (k + 1)/ N)} (1 - |Ny - k|) \diff G_{X_\mathbf{s}(F(T_i))}(y)\\
	&= \int_{[0, 1]}\sum_{k=0}^N \I_{\{(k-1)/ N < y < (k + 1)/ N\}} (1 - |Ny - k|) \diff G_{X_\mathbf{s}(F(T_i))}(y)\\
	&= \int_{[0, 1]} \diff G_{X_\mathbf{s}(F(T_i))}(y) = 1 \quad \text{for $1 \leq i \leq m$},
\end{align*}
where the first equality follows from the definitions of $ \{\hat p_{ik}^{(N)}(X_\mathbf{s})\}_{1 \leq i \leq m, 0 \leq k \leq N}$ and the third equality holds due to \eqref{identity1}.
Similarly, for the second set of linear constraints, we note the identity
\[
	\sum_{k=0}^N \frac{k}{N} \I_{\{(k-1)/ N < y < (k + 1)/ N\}} (1 - |Ny - k|) = y \quad \text{for $y \in [0, 1]$ and $N \in \mathbb{N}^+$}.
\]
Then it follows that 
\begin{align*}
	\sum_{k=0}^N k \hat p_{ik}^{(N)}(X) &= \sum_{k=0}^N k\int_{((k-1)/ N, (k + 1)/ N)} (1 - |Ny - k|) \diff G_{X_\mathbf{s}(F(T_i))}(y) \\
	&= N \int_{[0, 1]}\sum_{k=0}^N \frac{k}{N} \I_{\{(k-1)/ N < y < (k + 1)/ N\}}  (1 - |Ny - k|) \diff G_{X_\mathbf{s}(F(T_i))}(y)\\
	&= N\int_{[0, 1]} y \diff G_{X_\mathbf{s}(F(T_i))}(y) = N \E[X_\mathbf{s}(F(T_i))] = N F(T_i) \quad \text{for } 1 \leq i \leq m.
\end{align*}
To verify the third set of linear constraints, we define
\[
f_k^{(N)}(y) := (Ny + 1 - k)\I_{\{(k-1)/N < y < k / N\}} + \I_{\{k / N \leq y\}} \quad \text{for $0 \leq k \leq N$ and $N \in \mathbb{N}^+$}.
\]
Using this definition, we have
\begin{align*}
	\sum_{j \geq k} \hat p_{ij}^{(N)}(X)
	&= \sum_{j\geq k} \int_{((j - 1)/ N, (j + 1)/ N)} (1 - |Ny - j|) \diff G_{X_\mathbf{s}(F(T_i))}(y)\\
	&= \int_{((k - 1)/N, k/N)} (Ny + 1 - k) \diff G_{X_\mathbf{s}(F(T_i))}(y) + \int_{[k/N, 1]} \diff G_{X_\mathbf{s}(F(T_i))}(y)\\
	&= \E [f_k^{(N)}(X_\mathbf{s}(F(T_i)))] \quad \text{for $0 \leq k \leq N$ and $1 \leq i \leq m$}.
\end{align*}
Since $f_k^{(N)}(y)$ is a non-decreasing function of $y$ and $X_\mathbf{s}(F(T_i)) \leq X_\mathbf{s}(F(T_{i + 1}))$ a.s., it follows that 
\[
\sum_{j \geq k} \hat p_{i+1, j}^{(N)}(X_\mathbf{s}) - \sum_{j \geq k} \hat p_{ij}^{(N)}(X_\mathbf{s}) = \E [f_k^{(N)}(X_\mathbf{s}(F(T_{i + 1})))] - \E [f_k^{(N)}(X_\mathbf{s}(F(T_i)))] \geq 0
\]
for $0 \leq k \leq N$ and $1 \leq i \leq m-1$.
Finally, the fourth set of linear constraints (i.e., the non-negativity constraints) are satisfied trivially since all integrands are non-negative.
We have verified all the four sets of linear constraints in \eqref{ls3}.
Thus, we obtain that $\hat P^{(N)}(X_\mathbf{s}) \in \mathcal{P}^{(N)}$.

\vspace{0.2cm}
\noindent
\textit{Step 2. Constructing gamma distortions $X_\mathbf{s}^{(N)}$ for $N \in \mathbb{N}^+$ and showing the convergence in distribution.}
Since $\hat P^{(N)}(X_\mathbf{s}) \in \mathcal{P}^{(N)}$ for $N \in \mathbb{N}^+$, Proposition \ref{prop:generator} ensures the existence of a discrete gamma distortion $X_\mathbf{s}^{(N)}$ whose generator $\phi_\mathbf{s}^{(N)}$ satisfies
\begin{equation}\label{eq:phi_distribution_N}
\p(\phi_\mathbf{s}^{(N)}(F(T_i)) = k) = \hat p_{ik}^{(N)}(X_\mathbf{s}) \quad \text{for $1\leq i \leq m$ and $0 \leq k \leq N$}.
\end{equation}
We shall show that 
\begin{equation}\label{eq:distribution_convergence}
X_\mathbf{s}^{(N)}(F(T_i)) \xrightarrow{d} X_\mathbf{s}(F(T_i)) \quad \text{as $N \to +\infty$ for $1 \leq i \leq m$}.
\end{equation}
To this end, we fix $1 \leq i \leq m$, and let $y_0 \in [0, 1]$ be an arbitrary continuity point of the distribution function $G_{X_\mathbf{s}(F(T_i))}$. We need to prove that 
\begin{equation}\label{eq:cdf_convergence}
\lim_{N\to+\infty}G_{X_\mathbf{s}^{(N)}(F(T_i))}(y_0) = G_{X_\mathbf{s}(F(T_i))}(y_0).
\end{equation}

Note that $X_{\mathbf{s}}^{(N)}(F(T_i))$ follows a Beta mixture distribution with mixing weight $\hat p_{ik}^{(N)}(X_{\mathbf{s}})$ on the $\Beta(k, N - k)$ component for each $0 \leq k \leq N$, where $\Beta(\alpha, \beta)$ denotes a Beta distribution with parameters $\alpha > 0$ and $\beta > 0$.
Then the cumulative distribution function (CDF) of $X_\mathbf{s}^{(N)}(F(T_i))$ is given by
\begin{equation}\label{G_XN_distribution}
G_{X_\mathbf{s}^{(N)}(F(T_i))}(y_0) = \sum_{k=0}^N \hat p_{ik}^{(N)}(X_\mathbf{s}) F^B(y_0; k, N - k),
\end{equation}
where $F^B(\cdot; \alpha, \beta)$ denotes the CDF of the $\Beta(\alpha, \beta)$ distribution for $\alpha > 0$ and $\beta > 0$. 
For the degenerate cases, $F^{B}(y; 0, \beta) := \I_{\{y \geq 0\}}$ and $F^{B}(y; \alpha, 0) := \I_{\{y \geq 1\}}$ represent the degenerate distributions on $0$ and $1$, respectively.
Since $y_0$ is a continuity point, for any $\epsilon > 0$, there exists $\delta > 0$ such that 
\begin{equation}\label{eq:continuity}
	|G_{X_\mathbf{s}(F(T_i))}(y) - G_{X_\mathbf{s}(F(T_i))}(y_0)| < \epsilon \text{ for all $|y - y_0| \leq \delta$}.
\end{equation}

To analyse the limit of the right hand side of \eqref{G_XN_distribution} as $N$ goes to $+\infty$, we shall first establish upper and lower bounds for $F^{B}(x; \alpha, \beta)$. 
Let $W_{\alpha, \beta} \sim \mathrm{Beta}(\alpha, \beta)$. 
The mean and variance of $W_{\alpha, \beta}$ are given by $\frac{\alpha}{\alpha + \beta}$ and $\frac{\alpha\beta}{(\alpha + \beta)^2(\alpha + \beta + 1)}$, respectively.
For $x < \frac{\alpha}{\alpha + \beta}$, applying Chebyshev's inequality yields an upper bound for $F^B(x; \alpha, \beta)$.
\begin{align}\label{chebyshev_upper_bound}
	F^B(x; \alpha, \beta) 
	= & \ \p (W_{\alpha, \beta} \leq x) 
	=  \p \left(W_{\alpha, \beta} - \frac{\alpha}{\alpha + \beta} \leq x - \frac{\alpha}{\alpha + \beta} \right) \notag\\
	\leq &\ \frac{\operatorname{Var}(W_{\alpha, \beta})}{\left(x - \frac{\alpha}{\alpha + \beta}\right)^2} 
	=  \frac{\alpha\beta}{\left(x - \frac{\alpha}{\alpha + \beta}\right)^2(\alpha + \beta)^2(\alpha + \beta + 1)} \notag \\
	\leq &\ \frac{1}{4\left(x - \frac{\alpha}{\alpha + \beta}\right)^2(\alpha + \beta + 1)}.
\end{align}
Similarly, for $x > \frac{\alpha}{\alpha + \beta}$, Chebyshev's inequality implies a lower bound for $F^B(x; \alpha, \beta)$.
\begin{align}\label{chebyshev_lower_bound}
1 - F^B(x; \alpha, \beta) &= \p(W_{\alpha, \beta} > x) = \p\left(W_{\alpha, \beta} - \frac{\alpha}{\alpha + \beta} > x - \frac{\alpha}{\alpha + \beta}\right) \notag\\
&\leq \frac{\operatorname{Var}(W_{\alpha, \beta})}{\left(x - \frac{\alpha}{\alpha + \beta}\right)^2} = \frac{\alpha\beta}{\left(x - \frac{\alpha}{\alpha + \beta}\right)^2(\alpha + \beta)^2(\alpha + \beta + 1)} \notag \\
&\leq \frac{1}{4\left(x - \frac{\alpha}{\alpha + \beta}\right)^2(\alpha + \beta + 1)}.
\end{align}
Note that \eqref{chebyshev_upper_bound} and \eqref{chebyshev_lower_bound} also hold for the degenerate cases when either $\alpha = 0$ or $\beta = 0$.

Now, we are ready to establish a lower bound for the right hand side of \eqref{G_XN_distribution}.
Define $K_0 = \lfloor N(y_0 - \delta/2) \rfloor$. 
Truncating the summation in \eqref{G_XN_distribution} at $K_0$ 
, we can obtain
\begin{align}\label{eq:lower_bound_GXN}
	G_{X_\mathbf{s}^{(N)}(F(T_i))}(y_0)
	&\geq \sum_{k=0}^{K_0} \hat p_{ik}^{(N)}(X_\mathbf{s}) F^B(y_0; k, N - k) \geq \sum_{k=0}^{K_0} \hat p_{ik}^{(N)}(X_\mathbf{s})\left(1 - \frac{1}{4(N + 1)(y_0 - k / N)^2}\right) \notag \\
	&\geq \left(1 - \frac{1}{(N + 1)\delta^2}\right) \sum_{k=0}^{K_0} \hat p_{ik}^{(N)}(X_\mathbf{s}),
\end{align}
where the second inequality follows from \eqref{chebyshev_lower_bound} and the last inequality follows from the fact that $y_0 - k / N \geq \delta / 2$ for $k \leq K_0$.

Using the inequality
\[
	\sum_{k = 0}^{K_0} \I_{\{(k-1)/ N < y < (k + 1)/ N\}} (1 - |Ny - k|) \geq \I_{\{0 \leq y \leq K_0 / N\}} \quad \text{for $y \in [0, 1]$ and $N \in \mathbb{N}^+$}
\]
and exchanging the order of integration and summation, we yield the inequality
\begin{align}\label{sum_pk_lower_bound}
\sum_{k=0}^{K_0} \hat p_{ik}^{(N)}(X_\mathbf{s}) 
=& \sum_{k=0}^{K_0} \int_{(\frac{k-1}{N}, \frac{k+1}{N})} (1 - |Ny - k|) \diff G_{X_\mathbf{s}(F(T_i))}(y) \notag\\
=& \int_{[0, 1]} \sum_{k=0}^{K_0} \I_{\{(k-1)/ N < y < (k + 1)/ N\}} (1 - |Ny - k|) \diff G_{X_\mathbf{s}(F(T_i))}(y) \notag \\
\geq& \int_{[0, K_0 / N]} \diff G_{X_\mathbf{s}(F(T_i))}(y) = G_{X_\mathbf{s}(F(T_i))}\left(\frac{K_0}{N}\right).
\end{align}

For $N > \max\{2 / \delta, 1 / \delta^2\}$, it holds that $1 - \frac{1}{(N + 1)\delta^2} > 0$, and the argument of the CDF in the right hand side of \eqref{sum_pk_lower_bound} satisfies 
\begin{equation}\label{eq:K0_bound}
\frac{K_0}{N} > \frac{N(y_0 - \delta / 2) - 1}{N} = y_0 - \delta/2 - 1 / N > y_0 - \delta.
\end{equation}
Since $G_{X_\mathbf{s}(F(T_i))}$ is non-decreasing, using \eqref{eq:continuity} and \eqref{eq:K0_bound} yields
\begin{equation}\label{G_sum_pk_bound}
	G_{X_\mathbf{s}(F(T_i))}\left(\frac{K_0}{N}\right) \geq G_{X_\mathbf{s}(F(T_i))}(y_0 - \delta) \geq G_{X_\mathbf{s}(F(T_i))}(y_0) - \epsilon.
\end{equation}

Combining the results in \eqref{eq:lower_bound_GXN}, \eqref{sum_pk_lower_bound} and \eqref{G_sum_pk_bound}, we find the following lower bound for $G_{X_\mathbf{s}^{(N)}(F(T_i))}(y_0)$:
\begin{align*}
	G_{X^{(N)}(F(T_i))}(y_0)
	&\geq \left(1 - \frac{1}{(N + 1)\delta^2}\right) (G_{X_\mathbf{s}(F(T_i))}(y_0) - \epsilon).
\end{align*}
By taking the limit inferior as $N$ goes to $+\infty$, it follows that 
\begin{equation}
	\liminf_{N \to +\infty} G_{X_\mathbf{s}^{(N)}(F(T_i))}(y_0) \geq G_{X_\mathbf{s}(F(T_i))}(y_0) - \epsilon.
	\label{lower_lim}
\end{equation}

Next, we establish an upper bound for the right hand side of \eqref{G_XN_distribution}.
To this end, we define $K_1 = \lceil N(y_0 + \delta/2) \rceil$, and split the summation in \eqref{G_XN_distribution} into two parts ($k \leq K_1 - 1$ and $k \geq K_1$) and use the trival upper bound $F^B(y_0; k, N-k) \leq 1$ for the first part to obtain
\begin{align}
	G_{X^{(N)}(F(T_i))}(y_0) 
	&= \sum_{k=0}^{N} \hat p_{ik}^{(N)}(X_\mathbf{s}) F^B(y_0; k, N - k)\notag \\
	&\leq \sum_{k = 0}^{K_1 - 1} \hat p_{ik}^{(N)}(X_\mathbf{s}) + \sum_{k = K_1}^N \hat p_{ik}^{(N)}(X_\mathbf{s}) F^B(y_0; k, N - k).
	\label{G_XN_split}
\end{align}

Using the inequality 
\[
	\sum_{k = 0}^{K_1 - 1} \I_{\{(k-1)/ N < y < (k + 1)/ N\}} (1 - |Ny - k|) \leq \I_{\{0 \leq y \leq K_1 / N\}} \quad \text{for $y \in [0, 1]$ and $N \in \mathbb{N}^+$},
\]
we find that the first sum in \eqref{G_XN_split} is bounded by
\begin{align}\label{sum_pk_upper_bound}
\sum_{k=0}^{K_1 - 1} \hat p_{ik}^{(N)}(X_\mathbf{s}) 
= & \sum_{k=0}^{K_1 - 1} \int_{((k-1)/ N, (k + 1)/ N)} (1 - |Ny - k|) \diff G_{X_\mathbf{s}(F(T_i))}(y) \notag\\
=& \int_{[0, 1]} \sum_{k=0}^{K_1 - 1} \I_{\{(k-1)/ N < y < (k + 1)/ N\}} (1 - |Ny - k|) \diff G_{X_\mathbf{s}(F(T_i))}(y) \notag \\
\leq& \int_{[0, K_1/N]} \diff G_{X_\mathbf{s}(F(T_i))}(y) = G_{X_\mathbf{s}(F(T_i))}\left(\frac{K_1}{N}\right).
\end{align}
For $N > 2/\delta$, the argument of the CDF on the right hand side of \eqref{sum_pk_upper_bound} satisfies 
\begin{equation}\label{eq:K1_bound}
\frac{K_1}{N} < \frac{N(y_0 + \delta / 2) + 1}{N} = y_0 + \delta / 2 + 1 / N < y_0+\delta. 
\end{equation}
Consequently, we can bound the first sum in \eqref{G_XN_split} as
\begin{equation}\label{G_sum_pk_upper_bound}
	\sum_{k = 0}^{K_1 - 1} \hat p_{ik}^{(N)}(X_\mathbf{s}) \leq G_{X_\mathbf{s}(F(T_i))}\left(\frac{K_1}{N}\right) \leq G_{X_\mathbf{s}(F(T_i))}(y_0 + \delta) \leq G_{X_\mathbf{s}(F(T_i))}(y_0) + \epsilon,
\end{equation}
where the first inequality follows from \eqref{sum_pk_upper_bound}, the second inequality follows from \eqref{eq:K1_bound}, 
and the last inequality follows from \eqref{eq:continuity}.

For the second sum in \eqref{G_XN_split}, we use the upper bound established in \eqref{chebyshev_upper_bound} to obtain
\[
\sum_{k = K_1}^N \hat p_{ik}^{(N)}(X_\mathbf{s}) F^B(y_0; k, N - k) 
\leq \sum_{k = K_1}^N \hat p_{ik}^{(N)}(X_\mathbf{s}) \frac{1}{4(N+1)(k/N-y_0)^2}.
\]
Noticing that the condition $k \geq K_1$ implies $k/N - y_0 \geq \delta/2$ and that $\sum_{k = K_1}^N \hat p_{ik}^{(N)}(X_\mathbf{s}) \leq 1$, we can bound the second sum in \eqref{G_XN_split} as
\begin{equation}\label{second_sum_upper_bound}
\sum_{k = K_1}^N \hat p_{ik}^{(N)}(X_\mathbf{s}) F^B(y_0; k, N - k) \leq \frac{1}{(N+1)\delta^2}.
\end{equation}
Combining the bounds in \eqref{G_sum_pk_upper_bound} and \eqref{second_sum_upper_bound}, we can obtain the following upper bound for $G_{X_\mathbf{s}^{(N)}(F(T_i))}(y_0)$:
\begin{align*}
	G_{X_\mathbf{s}^{(N)}(F(T_i))}(y_0) \leq G_{X_\mathbf{s}(F(T_i))}(y_0) + \epsilon + \frac{1}{(N + 1)\delta^2}.
\end{align*}
By taking the limit superior as $N$ goes to $+\infty$, it follows that
\begin{equation}
	\limsup_{N \to +\infty} G_{X_\mathbf{s}^{(N)}(F(T_i))}(y_0) \leq G_{X_\mathbf{s}(F(T_i))}(y_0) + \epsilon.
	\label{upper_lim}
\end{equation}

Since $\epsilon > 0$ is arbitrary, \eqref{lower_lim} and \eqref{upper_lim} imply the convergence in \eqref{eq:cdf_convergence}.
This completes the proof of the convergence in distribution in \eqref{eq:distribution_convergence}.

\vspace{0.2cm}
\noindent
\textit{Step 3. Proving the convergence in \eqref{eq:convergence}.}
Recall that conditional on $\phi_{\mathbf{s}}^{(N)}(F(T_i)) = k$, $X_{\mathbf{s}}^{(N)}(F(T_i))$ follows a $\Beta(k, N - k)$ distribution for $0 \leq k \leq N$.
By the property of the Beta distribution and the definition of $h_{jk}^{(N)}$, we have 
\[
	\binom{n}{j} \E \left[(X_{\mathbf{s}}^{(N)}(F(T_i)))^j(1 - X_{\mathbf{s}}^{(N)}(F(T_i)))^{n - j} \big| \phi_{\mathbf{s}}^{(N)}(F(T_i)) = k\right] = h_{jk}^{(N)}
\]
for $0 \leq j \leq n$ and $0 \leq k \leq N$.
Conditional on $\phi_{\mathbf{s}}^{(N)}(F(T_i))$, we can derive from the law of total expectation that
\[
\binom{n}{j} \E \left[X_{\mathbf{s}}^{(N)}(F(T_i))^j(1 - X_{\mathbf{s}}^{(N)}(F(T_i)))^{n - j}\right] = \sum_{k=0}^{N} h_{jk}^{(N)} \hat p_{ik}^{(N)}(X_{\mathbf{s}}) \quad \text{for $0 \leq j \leq n$}.
\]
Substituting this identity into \eqref{vPdef}, we can obtain
\begin{align}
    v^{[a_l,b_l]}(\hat P^{(N)}(X_{\mathbf{s}}); \mathbf{s}) 
    =&\ \sum_{i=1}^m \sum_{j=0}^n  \lambda_i^{[a_l,b_l]}(\mathbf{s}) \beta_j^{[a_l,b_l]} \binom{n}{j} \E \left[X^{(N)}_{\mathbf{s}}(F(T_i))^j(1 - X^{(N)}_{\mathbf{s}}(F(T_i)))^{n - j}\right] \notag \\
    &\ -  \gamma^{[a_l, b_l]}(\mathbf{s}) \quad \text{for } 1\leq l \leq M. \label{eq:vn_expectation}
\end{align}
We have shown in Step 2 that $X^{(N)}_{\mathbf{s}}(F(T_i))$ converges in distribution to $X_{\mathbf{s}}(F(T_i))$ as $N$ goes to $+\infty$ for $1 \leq i \leq m$. 
Since the function $y \mapsto y^j (1-y)^{n-j}$ is continuous and bounded on $[0, 1]$, the convergence in distribution implies the convergence of expectations.
\begin{equation}\label{eq:expectation_convergence}
\lim_{N \to +\infty} \E \left[X^{(N)}_{\mathbf{s}}(F(T_i))^j(1 - X^{(N)}_{\mathbf{s}}(F(T_i)))^{n - j}\right] 
= \E \left[X_{\mathbf{s}}(F(T_i))^j(1 - X_{\mathbf{s}}(F(T_i)))^{n - j}\right].
\end{equation}
Letting $N$ go to $+\infty$ in \eqref{eq:vn_expectation} and using \eqref{eq:expectation_convergence} yields
\begin{align}\label{eq:temp}
&\lim_{N \to +\infty} v^{[a_l,b_l]}(\hat P^{(N)}(X_{\mathbf{s}}); \mathbf{s}) \notag \\
=&\ \sum_{i=1}^m \sum_{j=0}^n \lambda_i^{[a_l,b_l]}(\mathbf{s}) \beta_j^{[a_l,b_l]} \binom{n}{j} \E \left[X_{\mathbf{s}}(F(T_i))^j(1 - X_{\mathbf{s}}(F(T_i)))^{n - j}\right] - \gamma^{[a_l, b_l]}(\mathbf{s})
\end{align}
for $1 \leq l \leq M$.
Note that the right hand side of \eqref{eq:temp} is exactly $v^{[a_l, b_l]}(C^{X_{\mathbf{s}}}; \mathbf{s})$ given in \eqref{distortion_v}.
Furthermore, since $C^{X_{\mathbf{s}}}$ achieves a perfect fit to the market prices $\mathbf{s}$, we have $v^{[a_l, b_l]}(C^{X_{\mathbf{s}}}; \mathbf{s}) = 0$ for all $l=1, \ldots, M$. 
Consequently, we have proved the convergence in \eqref{eq:convergence}.
\hfill $\square$

\vspace{0.2cm}
Define $\mathcal{O}^{I}$ to be the open orthant in $\mathbb{R}^M$ as follows: 
\[
	\mathcal{O}^{I} := \{\mathbf{x} = (x_1, \ldots, x_M) \in \mathbb{R}^M: x_l > 0 \text{ for } l \in I, \, x_l < 0 \text{ for } l \notin I\} \quad \text{for } I \subseteq \{1, \ldots, M\}.
\]
The following lemma characterizes a sufficient condition for a convex set to contain the origin.
\begin{Lemma}
	\label{lemma_convex}
	Let $\mathcal{E}\subseteq \mathbb{R}^M$ be a convex set. 
	If $\mathcal{E} \cap \mathcal{O}^{I} \neq \emptyset$ for all $I \subseteq \{1, \ldots, M\}$,
	then $\mathbf{0}_M \in \mathcal{E}$.  
\end{Lemma}
{\it Proof.}
We use proof by contradiction. 
Suppose that $\mathbf{0}_M \notin \mathcal{E}$. 
By the hyperplane separation theorem, there exists a nonzero vector $\mathbf{x} = (x_1, \ldots, x_M) \in \mathbb{R}^M$ such that $\mathbf{x} \cdot \mathbf{z} \geq 0$ for all $\mathbf{z} \in \mathcal{E}$. 
Define $I_{\mathbf{x}} := \{l: x_l < 0\}$. 
By hypothesis, there exists a vector $\mathbf{y} = (y_1, \ldots, y_M) \in \mathcal{E} \cap \mathcal{O}^{I_\mathbf{x}}$. 
By the definition of $\mathcal{O}^{I_\mathbf{x}}$, we have $x_l y_l \leq 0$ for all $1 \leq l \leq M$. 
Since $\mathbf{x} \neq \mathbf{0}_M$ and all components of $\mathbf{y}$ are nonzero, the dot product $\mathbf{x} \cdot \mathbf{y}$ is strictly negative. 
This contradicts the fact that $\mathbf{x} \cdot \mathbf{y} \geq 0$. 
Therefore, we conclude that $\mathbf{0}_M \in \mathcal{E}$.
\hfill $\square$

\subsection{Proof of Proposition \ref{prop_interior}}

We are now ready to prove Proposition \ref{prop_interior}.
Consider an arbitrary vector $\mathbf{s} \in \mathcal{A}^\circ$. 
Define the set
\[
	\mathcal{E}^{(N)}(\mathbf{s}) = \{ \mathbf{v}({\hat P^{(N)}; \mathbf{s}}): \hat P^{(N)} \in \mathcal{P}^{(N)}\}\quad\text{for any $N\in\mathbb N^+$}.
\]
Note that the feasibility of \eqref{linearsystem2} for $N \in \mathbb{N}^+$ is equivalent to the condition that $\mathbf{0}_M \in \mathcal{E}^{(N)}(\mathbf{s})$. Therefore, proving Proposition \ref{prop_interior} is equivalent to showing that $\mathbf{0}_M \in \mathcal{E}^{(N)}(\mathbf{s})$ for all sufficiently large $N$.     
According to Lemma \ref{lemma_convex}, it suffices to show that for any $I \subseteq \{1, \ldots, M\}$, we have $\mathcal{E}^{(N)}(\mathbf{s}) \cap \mathcal{O}^{I} \neq \emptyset$ for all sufficiently large $N$. 


Since $\mathbf{s}$ lies in the interior $\mathcal{A}^\circ$, there exists an open ball centered at $\mathbf{s}$ and contained in $\mathcal{A}$. 
For sufficiently small $\delta_l > 0$ (for $1\leq l \leq M$), we define a perturbed price $\mathbf{\tilde s}$ lying within this ball as follows:
\[
\tilde s^{[a_l,b_l]} = s^{[a_l,b_l]}\ \, \text{for} \, 1\leq l \leq M; \quad
\tilde \uf^{[a_l,b_l]} = \uf^{[a_l,b_l]} + \delta_l\ \, \text{for } l \in I; \quad
\tilde \uf^{[a_l,b_l]} = \uf^{[a_l,b_l]} - \delta_l\ \, \text{for } l \notin I.
\]
Since $\lambda_i^{[a_l,b_l]}$ depends only on the tranche spread, this definition ensures that
\begin{equation}
\lambda_i^{[a_l,b_l]}(\mathbf{\tilde s}) = \lambda_i^{[a_l,b_l]}(\mathbf{s}) \quad \text{for  $1 \leq i \leq m$ and $1 \leq l \leq M$}. 
\label{lambda_perturb}
\end{equation}
Recalling the definition of $\gamma^{[a_l, b_l]}$ in \eqref{vspread_2}, we have 
\[
\gamma^{[a_l, b_l]}(\mathbf{\tilde s}) - \gamma^{[a_l, b_l]}(\mathbf{s}) = \begin{cases}
	\delta_l (b_l - a_l) > 0 &\text{for }  l \in I, \\
	-\delta_l (b_l - a_l) < 0 &\text{for } l \notin I.
\end{cases}
\]
Define $\epsilon = \frac{1}{2}\min_{1\leq l \leq M}\{|\gamma^{[a_l, b_l]}(\mathbf{\tilde s}) - \gamma^{[a_l, b_l]}(\mathbf{s})|\} > 0$.
Let $X_{\mathbf{\tilde s}}$ be a stochastic distortion whose transformed copula achieves a perfect fit to $\mathbf{\tilde s}$. Applying Lemma \ref{lemma:convergence} to $\mathbf{\tilde s}$ and $X_{\mathbf{\tilde s}}$ yields 
\begin{equation*}
	\|\mathbf{v}(\hat P^{(N)}(X_{\mathbf{\tilde s}}); \mathbf{\tilde s})\| \to 0 \quad \text{as $N \to +\infty$}.
\end{equation*}
Thus, there exists $N_0\in \mathbb{N}^+$ such that $\|\mathbf{v}(\hat P^{(N)}(X_{\mathbf{\tilde s}}); \mathbf{\tilde s})\| < \epsilon$ for all $N \geq N_0$.

Next, we consider $\mathbf{v}(\hat P^{(N)}(X_{\mathbf{\tilde s}}); \mathbf{s})$ for $N \geq N_0$. 
Recalling the definition of $v^{[a_l, b_l]}$ in \eqref{vPdef} and the identity \eqref{lambda_perturb}, we can obtain
\begin{align*}
	&v^{[a_l, b_l]}(\hat P^{(N)}(X_{\mathbf{\tilde s}}); \mathbf{s})\\ 
	=& \sum_{i=1}^m \sum_{j=0}^n \sum_{k=0}^N  h_{jk}^{(N)} \lambda_i^{[a_l,b_l]}(\mathbf{s}) \beta_j^{[a_l,b_l]} \hat p_{ik}^{(N)}(X_{\mathbf{\tilde s}}) - \gamma^{[a_l, b_l]}(\mathbf{s})\\
	=& \left[\sum_{i=1}^m \sum_{j=0}^n \sum_{k=0}^N  h_{jk}^{(N)} \lambda_i^{[a_l,b_l]}(\mathbf{\tilde s}) \beta_j^{[a_l,b_l]} \hat p_{ik}^{(N)}(X_{\mathbf{\tilde s}}) - \gamma^{[a_l, b_l]}(\mathbf{\tilde s})\right] + \gamma^{[a_l, b_l]}(\mathbf{\tilde s}) - \gamma^{[a_l, b_l]}(\mathbf{s})\\
	=& v^{[a_l, b_l]}(\hat P^{(N)}(X_{\mathbf{\tilde s}}); \mathbf{\tilde s}) + \gamma^{[a_l, b_l]}(\mathbf{\tilde s}) - \gamma^{[a_l, b_l]}(\mathbf{s}).
\end{align*}
For any $l \in I$, we have $\gamma^{[a_l, b_l]}(\mathbf{\tilde s}) - \gamma^{[a_l, b_l]}(\mathbf{s}) \geq 2\epsilon$, which ensures that
\[
v^{[a_l, b_l]}(\hat P^{(N)}(X_{\mathbf{\tilde s}}); \mathbf{s}) \geq -\epsilon + 2\epsilon = \epsilon > 0.
\]
Similarly, for any $l \notin I$, we have $\gamma^{[a_l, b_l]}(\mathbf{\tilde s}) - \gamma^{[a_l, b_l]}(\mathbf{s}) \leq -2\epsilon$, thereby yielding
\[
v^{[a_l, b_l]}(\hat P^{(N)}(X_{\mathbf{\tilde s}}); \mathbf{s}) \leq \epsilon - 2\epsilon = -\epsilon < 0.
\]
Thus we have shown that $\mathbf{v}(\hat P^{(N)}(X_{\mathbf{\tilde s}}); \mathbf{s}) \in \mathcal{O}^I$ for all $N \geq N_0$. In addition, for any $N\in\mathbb N^+$, we have $\mathbf{v}(\hat P^{(N)}(X_{\mathbf{\tilde s}}); \mathbf{s}) \in \mathcal{E}^{(N)}(\mathbf{s})$ because $\hat P^{(N)}(X_{\mathbf{\tilde s}}) \in \mathcal{P}^{(N)}$. Accordingly, we can obtain that  $\mathbf{v}(\hat P^{(N)}(X_{\mathbf{\tilde s}}); \mathbf{s}) \in \mathcal{E}^{(N)}(\mathbf{s}) \cap \mathcal{O}^I$ for all $N \geq N_0$. Therefore, we conclude that $\mathcal{E}^{(N)}(\mathbf{s}) \cap \mathcal{O}^I \neq \emptyset$ for all $N \geq N_0$.
This completes the whole proof of Proposition \ref{prop_interior}. $\hfill$ $\Box$

\section{Proofs of Proposition \ref{bound_convergence} and Proposition \ref{strong_bid-ask-corollary}}\label{bound_convergence_proof}

\noindent\textit{\textbf{Proof of Proposition \ref{bound_convergence}.}} 
We will focus only on the establishment of the convergence of the upper bound for the up-front payment, i.e.,
\[
\lim_{N \to +\infty} 
\overline{\uf}^{[a_{l+1}, b_{l+1}]}_N 
= \overline{\uf}^{[a_{l+1}, b_{l+1}]}.
\]
The convergence of the other three bounds follows by analogous arguments.

For each $N \in\mathbb N^+$, let $\mathcal{C}_{DG}^{(N)}$ denote the family of discrete gamma-distorted copulas with fixed parameter~$N$, and let $\mathcal{A}_N^{l+1}$ denote the set of price vectors of the first $(l+1)$ tranches that are compatible with respect to~$\mathcal{C}_{DG}^{(N)}$.  
Since $\mathcal{C}_{DG}^{(N)} \subseteq \mathcal{C}_1$, it follows that $\mathcal{A}_N^{l+1} \subseteq \mathcal{A}^{l+1}$ for all $N \in \mathbb{N}^+$ (recall that $\mathcal{A}^{l+1}$ denotes the set of strongly compatible market prices for the first $l+1$ tranches).

Moreover, by Proposition~\ref{prop_interior}, every interior point of $\mathcal{A}^{l+1}$ belongs to $\mathcal{A}_N^{l+1}$ for all sufficiently large~$N$.  
Hence, we can deduce that
\begin{equation}\label{inclusion_chain}
(\mathcal{A}^{l+1})^\circ 
\subseteq \liminf_{N \to +\infty} \mathcal{A}_N^{l+1}
\subseteq \limsup_{N \to +\infty} \mathcal{A}_N^{l+1}
\subseteq \mathcal{A}^{l+1}.
\end{equation}

For $x \in \mathbb{R}$, define the vector $\mathbf{s}(x)$ as
\[
\mathbf{s}(x) := (s^{[a_1,b_1]}, \uf^{[a_1,b_1]}, \ldots, s^{[a_l,b_l]}, \uf^{[a_l,b_l]}, s^{[a_{l+1},b_{l+1}]}, x).
\]
Furthermore, define 
\[
\mathcal{S}_N := \{x \in \mathbb{R} : \mathbf{s}(x) \in \mathcal{A}_N^{l+1}\}\quad\text{for any $N\in\mathbb N^+$} 
\quad\text{and}\quad
\mathcal{S} := \{x \in \mathbb{R} : \mathbf{s}(x) \in \mathcal{A}^{l+1}\}.
\] 
Using these notations, the upper bounds defined in  \eqref{eq:uf_ub_lb} and \eqref{eqs:upfront_LP} can be expressed as
\[
\overline{\uf}^{[a_{l+1}, b_{l+1}]}_N = \sup \mathcal{S}_N\quad\text{for any $N\in\mathbb N^+$} 
\quad\text{and}\quad 
\overline{\uf}^{[a_{l+1}, b_{l+1}]} = \sup \mathcal{S}.
\]
Given that 
$(s^{[a_1,b_1]}, \uf^{[a_1,b_1]}, \ldots, s^{[a_l,b_l]}, \uf^{[a_l,b_l]}) \in (\mathcal{A}^l)^\circ$, Proposition~\ref{prop_interior} ensures that $\mathcal{P}_N^l$ is non-empty for all sufficiently large~$N$. 
This guarantees that for all sufficiently large $N$, $\mathcal{S}_N$ is non-empty, and thus $\overline{\uf}^{[a_{l+1}, b_{l+1}]}_N$ is well-defined. 
Moreover, the inclusion $\mathcal{S}_N \subseteq \mathcal{S}$ for any $N\in\mathbb N^+$ implies that $\mathcal{S}$ is also non-empty, and thus $\overline{\uf}^{[a_{l+1}, b_{l+1}]}$ is well-defined. 

Taking one-dimensional sections in the inclusion chain \eqref{inclusion_chain} yields
\begin{equation}\label{set_inclusion_chain}
\mathcal{S}^\circ
\subseteq \liminf_{N \to +\infty} \mathcal{S}_N
\subseteq \limsup_{N \to +\infty} \mathcal{S}_N
\subseteq \mathcal{S}.
\end{equation}
For $x_1, x_2 \in \mathcal{S}$ and $\theta \in (0,1)$, the convex combination $\theta \mathbf{s}(x_1) + (1-\theta) \mathbf{s}(x_2)$ belongs to $\mathcal{A}^{l+1}$ since a perfect fit can be achieved by mixing the two conditionally i.i.d. copulas that achieve perfect fits with $\mathbf{s}(x_1)$ and $\mathbf{s}(x_2)$, respectively. 
This implies that $\theta x_1 + (1-\theta) x_2 \in \mathcal{S}$, implying that $\mathcal{S}$ is convex.
Thus, $\mathcal{S}$ must be an interval or a singleton.
In both cases, \eqref{set_inclusion_chain} implies that
\[
\lim_{N \to +\infty}
\overline{\uf}^{[a_{l+1}, b_{l+1}]}_N
= \overline{\uf}^{[a_{l+1}, b_{l+1}]}.
\]
This completes the proof.
\hfill$\square$

\vspace{0.2cm}
\noindent\textit{\textbf{Proof of Proposition \ref{strong_bid-ask-corollary}.}}
We first prove the second implication ($(\mathbf{s}_{bid}, \mathbf{s}_{ask}) \in \mathcal{A}_{bid-ask}^\circ$ $\Rightarrow$ LP feasibility). Suppose that there exists an $N \in \mathbb{N}^{+}$ such that \eqref{linearsystem_bid_ask} has a feasible solution $\{\hat{p}_{ik}\}_{1\leq i \leq m, 0 \leq k \leq N}$.
Note that the last four sets of linear constraints in \eqref{linearsystem_bid_ask} are exactly the constraints of \eqref{linearsystem2_delete_1}.
By applying Proposition \ref{prop:generator}, we can construct a discrete gamma-distorted copula $C \in \mathcal{C}_1$ from $\{\hat{p}_{ik}\}_{1\leq i \leq m, 0 \leq k \leq N}$.
The first two sets of linear constraints in \eqref{linearsystem_bid_ask} ensure that the expected NPVs satisfy $v^{[a_l, b_l]}(C; \mathbf{s}_{bid}) \ge 0$ and $v^{[a_l, b_l]}(C; \mathbf{s}_{ask}) \le 0$ for all $l=1,\ldots,M$. Namely, $(\mathbf{s}_{bid}, \mathbf{s}_{ask})$ satisfies strong bid-ask compatibility. 

Next, we shall prove the first implication (LP feasibility $\Rightarrow$ strong bid-ask compatibility).
Since $(\mathbf{s}_{bid}, \mathbf{s}_{ask}) \in \mathcal{A}_{bid\text{-}ask}^{\circ}$, there exists $\mathbf{s}_{bid}'$ and $\mathbf{s}_{ask}'$ such that $\mathbf{s}_{bid} < \mathbf{s}_{bid}' \leq \mathbf{s}_{ask}' < \mathbf{s}_{ask}$ and $(\mathbf{s}_{bid}', \mathbf{s}_{ask}') \in \mathcal{A}_{bid\text{-}ask}$.
By the definition of $\mathcal{A}_{bid\text{-}ask}$, there exists a stochastic distortion $X_*$ satisfying
\[
	v^{[a_l, b_l]}(C^{X_*}; \mathbf{s}_{bid}') \geq 0 \quad \text{and} \quad v^{[a_l, b_l]}(C^{X_*}; \mathbf{s}_{ask}') \leq 0 \quad \text{for $1 \leq l \leq M$}.
\]
Since $v^{[a_l, b_l]}(C^{X_*}; \mathbf{s})$ is strictly decreasing in each component of $\mathbf{s}$, it follows that 
\begin{align}\label{eq:ttemp}
	v^{[a_l, b_l]}(C^{X_*}; \mathbf{s}_{bid}) > 0 \quad \text{and}\quad v^{[a_l, b_l]}(C^{X_*}; \mathbf{s}_{ask}) < 0 \quad \text{for  $1 \leq l \leq M$}.
\end{align}

Using the same argument as in Lemma~\ref{lemma:convergence},  we can construct a sequence of matrices $\hat{P}^{(N)}(X_*)\in \mathcal{P}^{(N)}$ for $N \in \mathbb{N}^+$ such that 
\begin{equation*}
	\lim_{N \to +\infty} v^{[a_l, b_l]}(\hat{P}^{(N)}(X_*); \mathbf{s}_{bid}) = v^{[a_l, b_l]}(C^{X_*}; \mathbf{s}_{bid}) \quad \text{for } 1 \leq l \leq M
\end{equation*}
and 
\begin{equation*}
	\lim_{N \to +\infty} v^{[a_l, b_l]}(\hat{P}^{(N)}(X_*); \mathbf{s}_{ask}) = v^{[a_l, b_l]}(C^{X_*}; \mathbf{s}_{ask})\quad \text{for } 1 \leq l \leq M.
\end{equation*}
These two convergence results along with \eqref{eq:ttemp} imply that 
there exists a sufficiently large $N_0 \in \mathbb{N}^+$ such that
\[
	v^{[a_l, b_l]}(\hat{P}^{(N_0)}(X_*); \mathbf{s}_{bid}) \geq 0 \quad \text{and} \quad v^{[a_l, b_l]}(\hat{P}^{(N_0)}(X_*); \mathbf{s}_{ask}) \leq 0 \quad \text{for } 1 \leq l \leq M.
\]	
These inequalities correspond exactly to the first two sets of linear constraints in \eqref{linearsystem_bid_ask} with $N=N_0$. On the other hand, the fact that $\hat{P}^{(N)}(X_*)\in \mathcal{P}^{(N)}$ for all $N \in \mathbb{N}^+$ implies that  
$\hat{P}^{(N_0)}(X_*)$ satisfies the last four sets of linear constraints in \eqref{linearsystem_bid_ask} with $N=N_0$. Thus, we conclude that $\hat{P}^{(N)}(X_*)$ is a feasible solution to \eqref{linearsystem_bid_ask}. This completes the proof.
\hfill$\square$

\section{Transforming an LFP Problem into an LP Problem}\label{sec_app:LFP2LP}

Remark \ref{CCtransform} mentioned that the computation of $\overline{s}^{[a_{l+1},b_{l+1}]}_N$ and $\underline{s}^{[a_{l+1},b_{l+1}]}_N$ requires solving LFP problems, which can be readily transformed into equivalent LP problems via the Charnes-Cooper transformation (\citealp{cooper1962programming}).
In this section, we outline the transformation procedure for the upper bound $\overline{s}^{[a_{l+1},b_{l+1}]}_N$, and the lower bound can be treated in an analogous manner.
		
		First, we express the LFP problem associated with $\overline{s}^{[a_{l+1},b_{l+1}]}_N$ in a more compact form:
		\[
		\overline{s}^{[a_{l+1},b_{l+1}]}_N =  \sup_{\{\hat{p}_{ik}\}\in\mathcal{P}_{N}^l} \frac{\sum_{i=1}^{m}\sum_{k=0}^{N} B_{ik} \hat p_{ik}}{C_1 - \sum_{i=1}^{m}\sum_{k=0}^{N} A_{ik} \hat{p}_{ik}},
		\]
		where $A_{ik}$, $B_{ik}$, and $C_1$ are pre-computed constants defined as follows:
		\begin{align*}
			A_{ik} &= \sum_{j=0}^{n} D(T_{i})(T_{i}-T_{i-1})\beta_{j}^{[a_{l+1}, b_{l+1}]}h_{jk}^{(N)}, &\text{for } 1 \leq i\leq m 0 \leq k\leq N, \\
			B_{ik} &= \sum_{j=0}^{n} \left[D\left(\frac{T_{i}+T_{i-1}}{2}\right)-D\left(\frac{T_{i}+T_{i+1}}{2}\right)1_{\{i<m\}}\right]\beta_{j}^{[a_{l+1}, b_{l+1}]}h_{jk}^{(N)} &\text{for } 1 \leq i\leq m, 0 \leq k\leq N,\\
			C_1 &= \sum_{i=1}^{m}D(T_{i})(T_{i}-T_{i-1})(b_{l+1}-a_{l+1}).
		\end{align*}
		The transformation introduces a new scalar decision variable $\hat t$ and a set of new decision variables $\hat{Y} = \{\hat y_{ik}\}_{1\leq i \leq m, 0 \leq k \leq N}$, which are related to the original decision variables $\hat p_{ik}$:
		\[
		\hat t := \frac{1}{C_1 - \sum_{i=1}^{m}\sum_{k=0}^{N} A_{ik} \hat{p}_{ik}} \quad \text{and} \quad \hat y_{ik} := \hat p_{ik}\hat t.
		\]
	Plugging these new decision variables transforms into the the LFP problem yields the following equivalent LP problem in terms of the new decision variables $\{\hat Y, \hat t\}$:\\
		\textbf{Maximize:}
		\[
		\sum_{i=1}^{m}\sum_{k=0}^{N} B_{ik} \hat y_{ik}
		\]
		\\ \textbf{Subject to:}
		\begin{equation*}
			\begin{cases}
				C_1 \hat t - \sum_{i=1}^{m}\sum_{k=0}^{N} A_{ik}  \hat y_{ik} - 1 = 0, \\
				\sum_{i=1}^m \sum_{k=0}^N g_{ikl'}^{(N)}(\mathbf{s}) \hat y_{ik} - \gamma^{[a_{l'} b_{l'}]}(\mathbf{s}) \hat t = 0 \quad &\text{for $1\leq l' \leq l$}, \\
				\sum_{k=0}^N k \hat y_{ik} - N F(T_i) \hat t = 0 \quad &\text{for $1 \leq i \leq m$}, \\
				\sum_{j \geq k} \hat y_{ij} - \sum_{j \geq k} \hat y_{i + 1, j} \leq 0 \quad &\text{for $1 \leq i \leq m - 1$ and $1 \leq k \leq N$}, \\
				\hat y_{ik} \geq 0\quad &\text{for $1 \leq i \leq m$ and $0\leq k\leq N$}, \\
				\hat t \ge 0.
			\end{cases}
			\label{linearsystem2_transformed}
		\end{equation*}

\section{Standard Methods for CDS Index Calculations}
\label{sec_hazard_rate}

This section details two methods used in our paper for CDS index calculation, and these two methods are standard in the literature (see, e.g., \citealp{cont2011dynamic}, and \citealp{hull2016options}). The first method (i.e., Algorithm \ref{cds_marginal_default}) shows how to derive the marginal default distribution $F(\cdot)$ from the market index spread $s^{\rm idx}$. 


\begin{breakablealgorithm}{Calculating the Marginal Distribution of Default Times}\label{cds_marginal_default}
	\item 
	For any $\mu>0$, define 
	\[
	\begin{split}
		J_1(\mu) &:= \sum_{i=1}^m \left[e^{-\mu T_i}(T_i - T_{i-1})D(T_i)\right]   , \\
		J_2(\mu) &:= \frac{1}{2} \sum_{i=1}^m \left[\left(e^{-\mu T_{i-1}}-e^{-\mu T_i}\right)(T_i - T_{i-1})  D\left(\frac{T_i + T_{i-1}}{2}\right)\right], \\
		J_3(\mu) &:= (1 - R) \sum_{i=1}^m \left[\left(e^{-\mu T_{i-1}}-e^{-\mu T_i}\right) D\left(\frac{T_i + T_{i-1}}{2}\right)\right].
	\end{split}
	\]
	Here $\mu$ represents the hazard rate of default times and is to be determined in {\bf Step 2}.
	\item Solve for the hazard rate $\mu$ that equates the present values of the premium and default legs of the CDS index:
	\[
	s^{\rm idx} (J_1(\mu) + J_2(\mu)) = J_3(\mu).
	\]
	The marginal default distribution is then $F(t) = 1 - \exp(-\mu t)$.
\end{breakablealgorithm}

Algorithm \ref{cds_marginal_default_2} below details the standard method for calculating the change in the CDS index's value $\Delta v^{CDS}$ in response to a small CDS index spread change $\Delta s$. Note that here $\Delta v^{CDS}$ is the denominator of the spread delta $\delta^{[a_l, b_l]}$, which Algorithm \ref{alg:hedging} aims to compute.


\begin{breakablealgorithm}{Calculating the CDS Index's Value Change $\Delta v^{CDS}$ for a Spread Change $\Delta s$}\label{cds_marginal_default_2}
	\item Given the marginal default distribution $F(\cdot)$, calculate the value of $\mathrm{PV01}$ (risky annuity):
	\[
	\mathrm{PV01} =  \sum_{i=1}^m (T_i - T_{i-1}) (1-F(T_i)) D(T_i)+\frac{1}{2}\sum_{i=1}^m (T_i - T_{i-1}) (F(T_{i}) - F(T_{i-1})) D\left(\frac{T_i + T_{i-1}}{2}\right).
	\] 
	\item The change in the CDS index's value $\Delta v^{CDS}$ for an index spread change $\Delta s$ is then computed as follows:
	\[
	\Delta v^{CDS} = \mathrm{PV01} \times \Delta s.
	\]
\end{breakablealgorithm}

\end{document}